%% file: NumComponentPanelMain.tex
\documentclass{article}

\usepackage{amssymb,amsmath,amsfonts,eurosym,geometry,ulem,graphicx,caption,color,setspace,sectsty,comment,footmisc,caption,natbib,pdflscape,subfigure,array,hyperref,upgreek,bbm}

\usepackage{tabularx}
\usepackage{amsthm}
\usepackage{amssymb}
\usepackage{graphicx}
\usepackage{amsmath}
\usepackage{verbatim}
\usepackage{setspace}     
\usepackage{ulem}
\usepackage{textpos}
\usepackage{changepage}
\usepackage{url}
\usepackage{float}
\usepackage[para,online,flushleft]{threeparttable}
\usepackage{booktabs}
\usepackage{multirow}
\usepackage{palatino}
\usepackage{enumitem}
\usepackage[table,xcdraw]{xcolor}
\DeclareMathOperator*{\argmax}{arg\,max}

\def\bs{\boldsymbol}

\def\t{^{\top}}

\newtheorem{theorem}{Theorem}

\newtheorem{proposition}{Proposition}
\newtheorem{lemma}[theorem]{Lemma}

\usepackage{natbib} 

\usepackage{etoolbox}
\newcounter{bibcount}
\makeatletter
\patchcmd{\@lbibitem}{\item[}{\item[\hfil\stepcounter{bibcount}{[\thebibcount]}}{}{}
\renewcommand\NAT@bibsetup%
[1]{\setlength{\leftmargin}{\bibhang}\setlength{\itemindent}{-\parindent}%
  \setlength{\itemsep}{\bibsep}\setlength{\parsep}{\z@}}
\makeatother

\newcommand{\R}{\mathbb{R}}
\newcommand{\E}{\mathbb{E}}

\newtheorem{assumption}{Assumption}

\newcolumntype{L}[1]{>{\raggedright\let\newline\\arraybackslash\hspace{0pt}}m{#1}}
\newcolumntype{C}[1]{>{\centering\let\newline\\arraybackslash\hspace{0pt}}m{#1}}
\newcolumntype{R}[1]{>{\raggedleft\let\newline\\arraybackslash\hspace{0pt}}m{#1}}

\begin{document}

\begin{titlepage}
\title{\textbf{Testing the Number of Components in Finite Mixture Normal Regression Models with Panel Data}
\author{Yu Hao\thanks{Address for correspondence: Yu (Jasmine) Hao, Faculty of Business and Economics, The University of Hong Kong. We are very grateful for the comments from Chun Pang Chow, Vadim Marmer, and Kevin Song. We are also grateful for the IAAE grant at the 2019 IAAE Conference. This research is supported by the Natural Science and Engineering Research Council of Canada. }\\
Faculty of Business and Economics\\
The University of Hong Kong\\
haoyu@hku.hk \and
Hiroyuki Kasahara\\
Vancouver School of Economics\\
The University of British Columbia\\
hkasahar@mail.ubc.ca
}}
\date{\today}
\maketitle
\begin{abstract}
\noindent  This paper introduces a likelihood ratio-based test for examining the null hypothesis of an $M_0$-component model versus an alternative $(M_0+1)$-component model in the context of normal mixture panel regression. We demonstrate that the first-order derivative of the density function for the variance parameter in the normal mixture panel regression model is linearly independent from its second-order derivative for the mean parameter, unlike the cross-sectional normal mixture. However, similar to the cross-sectional normal mixture, the likelihood ratio test statistic for the panel normal mixture remains unbounded. To manage this unboundedness, we use a penalized maximum likelihood estimator and derive the asymptotic distribution of the penalized likelihood ratio test and expectation-maximization  test statistics using a fourth-order Taylor expansion of the log-likelihood function for reparameterized parameters. A sequential hypothesis testing approach is developed for consistently estimating the number of components.
Simulation experiments reveal that the proposed tests have good finite sample performance. We apply these tests to estimate the number of production technology types for a finite mixture Cobb--Douglas production function model. Our findings suggest heterogeneity in output elasticities for intermediate goods, indicating firm-specific variation in production functions beyond Hicks-neutral productivity terms.
  \\
\bigskip
\end{abstract}
\setcounter{page}{0}
\thispagestyle{empty}
\end{titlepage}

\section{Introduction}\label{sec:1_intro}
\input{sections/FM_1_introduction}

\section{Heteroskedastic finite mixture panel normal regression model}\label{sec:2_model}
\input{sections/FM_2_finite_mixture}

\section{Likelihood ratio test for   $H_0:M = 1$ against $H_A:M=2$}\label{sec:3_test_1}
\input{sections/FM_3_m1}

\section{Likelihood ratio test for   $H_0:M = M_0$ against $H_A:M=M_0+1$}\label{sec:4_test}
\input{sections/FM_4_m}

\section{EM test for $H_0: M = M_0$ against $H_A: M = M_0 +1$}\label{sec:5_EM}
\input{sections/FM_5_EM}

\section{Asymptotic distribution under local alternatives }\label{sec:local}
We derive the asymptotic distribution of the PLRTS and EM test statistic under local alternatives. For brevity, we focus on testing $H_0: M=1$ against $H_A: M=2$. Consider the following local alternative to the homogeneous model $f(\bs w;\bs\gamma^*,\bs\theta^*)$ with $\bs\theta^*=(\mu^*,\sigma^{*2},(\bs{\beta}^*)\t)\t$. For brevity, we omit the common parameter $\bs \gamma$ in this section.
In a reparameterized parameter, $\bs\psi^* = ((\bs\nu^*)\t, (\bs{\lambda}^*)\t)\t$.  
For $\alpha^*\in (0,1)$  and a local parameter $\bs h=(\bs {h}_{\bs\nu}\t,\bs {h}_{\bs\lambda}\t)\t$ with $\bs h_{\bs \lambda}\in v(\bs\Theta_{\bs\lambda})$,  we consider a sequence of contiguous local alternatives 
$(\alpha_n,\bs{\psi}_n\t)\t = (\alpha_n,\bs\nu_n\t,\bs\lambda_n\t)\in  \bs{\Theta}_\alpha\times \bs{\Theta}_{\bs\nu}\times\bs{\Theta}_{\bs\lambda}$ such that, with $\bs{t}_{\bs\lambda}(\bs\lambda,\alpha)$ given by (\ref{eq:t_1}),
\begin{equation}\label{eq:h}
\bs {h_\nu}= \sqrt{n}(\bs\nu_n -\bs\nu^*),\quad \bs {h_\lambda} =\sqrt{n} \bs{t}_{\bs\lambda}(\bs\lambda_n,\alpha_n),\quad \text{and} \quad \alpha_n = \alpha^* + o(1).
\end{equation}
Equivalently, the non-reparameterized contiguous local alternatives are given by  
\begin{equation}\label{eq:local-alt}
\bs\theta_{1,n} =\bs \nu_n + (1-\alpha_n) \bs\lambda_n\quad\text{and}\quad \bs\theta_{2,n} = \bs\nu_n - \alpha_n \bs\lambda_n
\end{equation}
for  $\bs\nu_n = \bs \nu^* + n^{-1/2} \bs{h}_{\nu}$ and
$\bs\lambda_n= (\lambda_{1,n},\lambda_{2,n},....,\lambda_{q+2,n})\t$ with
\begin{align*}
\lambda_{j,n} = n^{-1/4} (\alpha_n(1-\alpha_n))^{-1/2} h_{\lambda,j}\quad \text{for $j=1,...,q+2$},
\end{align*}
where $\bs{h}_{\bs\lambda}=(h_{\lambda,1}^2,...,h_{\lambda,q+2}^2,h_{\lambda,1}h_{\lambda,2},....,h_{\lambda,q+1}h_{\lambda,q+2})\t$. The local alternatives are of order $n^{1/4}$ rather than $n^{1/2}$. See 
  the discussion following Proposition \ref{prop:t2_distribution}.

The following proposition provides the asymptotic distribution of the PLRT and EM test statistics under contiguous local alternatives.
\begin{proposition}\label{prop:local-power}
Suppose that the assumptions in Proposition \ref{prop:KS18_prop4} hold for $M_0=1$. Consider a sequence of contiguous local alternatives $\bs\vartheta_{2,n} = (\alpha_n,\bs\theta_{1,n}^\top,\bs\theta_{2,n}^\top)\t$ given in (\ref{eq:local-alt}), where  $\alpha_n$ and $\bs\lambda_n$ satisfy (\ref{eq:h}). Then, under $H_{1,n}: \bs\vartheta=  \bs\vartheta_{2,n}$, we have $PLR_n(1)$,
$EM_n(1)\overset{d}{\rightarrow} (\tilde{\bs{t}}_{\bs\lambda} )\t \bs{\mathcal{I}}_{\bs\lambda,\bs\eta}  \tilde{\bs{t}}_{\bs\lambda}$,  
where $\tilde{\bs t}_{\bs\lambda}$ has the same distribution as $\hat {\bs t}_{\bs\lambda}$ in Proposition \ref{prop:t2_distribution} but $\bs{G}_{\bs{\lambda,\bs{\eta}}}$ is replaced with $ ( \bs{\mathcal{I}}_{\bs\lambda,\bs\eta}  )^{-1} \bs{S}_{\bs\lambda,\bs\eta} + \bs h_{\bs\lambda}$.
\end{proposition}
Importantly, a set of contiguous local alternatives considered in (\ref{eq:local-alt}) excludes a sequence such that $\alpha_n\rightarrow 0$ or $1$.

\section{Sequential hypothesis testing}\label{sec:sht}
 To estimate  the number of components, we sequentially test $H_0: M=r$ against $H_1: M=r+1$  starting from $r=1$, and then $r=2,\ldots,\bar M$, where $\bar M$ is the upper bound for the number of components, which is assumed to be larger than $M_0$. The first value for $r$ that leads to a nonrejection of $H_0$ gives our estimate for $M_0$.  \cite{Robin2000} develop a similar sequential hypothesis test for estimating the rank of a matrix.

For $M=1, \ldots, \bar M$, let $c^{M}_{1-q_n}$ denote the $100(1-q_n)$ percentile of the  cumulative distribution function of a random variable $ \max \{(\hat{\bs{t}}^1_{\bs\lambda})^\top \bs{\mathcal{I}}^1_{\bs\lambda,\bs\eta}  \hat{\bs{t}}^1_{\bs\lambda} ,   \ldots, (\hat{\bs{t}}^{M}_{\bs\lambda})^\top \bs{\mathcal{I}}^{M}_{\bs\lambda,\bs\eta}  \hat{\bs{t}}^{M}_{\bs\lambda}  \}$ for $M=M_0$ in Propositions \ref{prop:tm0_distribution} and \ref{prop:KS18_prop4}. Let $\hat c^{M}_{1-q_n}$ be a consistent estimator of $c^{M}_{1-q_n}$. Then, our estimator based on sequential hypothesis testing (SHT, hereafter) is defined as
\begin{align}
\hat M_{\text{PLR}} & = \min_{M \in \{0, \ldots ,\bar M\}} \{M: PLR_n(r) \geq \hat c^r_{1-q_n}, r=0, \ldots ,M-1,  PLR_n({M})\ < \hat c^M_{1-q_n} \},\nonumber\\
\hat M_{\text{EM}} & = \min_{M \in \{0, \ldots ,\bar M\}} \{M: EM_n(r) \geq \hat c^r_{1-q_n}, r=0, \ldots ,M-1, EM_n({M}) \ < \hat c^M_{1-q_n} \}.\label{M-hat}
\end{align}
The estimators $\hat M_{\text{PLR}}$ and $\hat M_{\text{EM}}$  depend on the choice of the significance level $q_n$. The following proposition states that $\hat M_{\text{PLR}}$ and $\hat M_{\text{EM}}$ converge to $M_0$ in probability as $n \rightarrow \infty$ when $-n^{-1}\ln q_n=o(1)$ and  $q_n=o(1)$.

 Let $Q_n^M( \bs{\vartheta}_M):=n^{-1} \sum_{i=1}^n \ln f_M(\bs{w}_i; \bs{\vartheta}_M)$ and $Q^M( \bs{\vartheta}_M):= \E[\ln f_M(\bs{w}_i; \bs{\vartheta}_M)]$, where $f_M(\bs{w}_i; \bs{\vartheta}_M)$ is defined in (\ref{eq:fm}) for $M=1,...,\bar M$. 
 
 \begin{assumption} \label{assumption:4} For $M=1,...,M_0-1$,  
(a)  $Q^M( \bs{\vartheta}_M)$ has a unique maximum at $\bs{\vartheta}_M^*$ in $\Theta_{\bs{\vartheta}_M}$;  (b) $\Theta_{\bs{\vartheta}_M}$ is compact; 
(c) $\bs{\vartheta}_M^*$ is interior to $\Theta_{\bs{\vartheta}_M}$; (d) $B^M(\bs{\vartheta}_M^*):= \E \left\{\nabla_{\bs{\vartheta}_M}  \ln f_M(\bs{w}_i; \bs{\vartheta}_M) \nabla_{\bs{\vartheta}_M^\top} \ln f_M(\bs{w}_i; \bs{\vartheta}_M)\right\}$ is non-singular; and (e) $A^M(\bs{\vartheta}_M^*):= \E \left\{\nabla_{\bs{\vartheta}_M\bs{\vartheta}_M^\top}  \ln f_M(\bs{w}_i; \bs{\vartheta}_M)\right\}$ has a constant rank in some open neighborhood of $\bs{\vartheta}_M^*$;
(f)  $Q^{M+1}( \bs{\vartheta}_{M+1}^*)  - Q^M( \bs{\vartheta}_M^*) > 0$.
\end{assumption}

\begin{proposition}\label{prop:sht}
Suppose that $M_0<\bar M$ and Assumptions \ref{assumption:1}--\ref{assumption:4} hold. If  we choose $q_n$ such that  $-n^{-1}\ln q_n=o(1)$ and  $q_n=o(1)$, then $\hat M_{\text{PLR}}-M_0=o_p(1)$ and $\hat M_{\text{EM}}-M_0=o_p(1)$.
\end{proposition}

Assumptions \ref{assumption:4}(a)--(e) ensure the consistency and asymptotic normality of $\hat{\bs{\vartheta}}_M$, where (c)--(e) correspond to Assumption A6 of \cite{white82em}. Per Assumption \ref{assumption:4}(f), the Kullback--Leibler information criterion of the model relative to the true $M_0$-component model strictly decreases as the number of components $M$ increases for $M < M_0$.

\section{Simulation}\label{sec:6_sim}
\input{sections/FM_6_simulated}

\section{Empirical application}\label{sec:7_app}
\input{sections/FM_7_application}

\section{Conclusion}
\input{sections/FM_8_conclusion}

\bibliographystyle{apalike}
\bibliography{references}

\appendix

\section{Proofs of propositions} \label{sec:appendixa}
\input{sections/FM_appendixa}

\section{Auxiliary results and their proofs}\label{sec:appendixc}
\input{sections/FM_appendixb}

\section{Other tables}\label{sec:appendixd}
\input{sections/FM_appendixd}

\end{document}

%% file: sections/FM_1_introduction.tex

Finite mixture models offer a natural representation of heterogeneity across a finite number of classes. Because of their flexibility, they have been used in empirical applications in various fields since the proposal of a two-component normal mixture model by \cite{Pearson1894}. In economics, finite mixtures are frequently used to model unobserved individual-specific effects in labor economics, health economics, and industrial organization, as well as in other fields.\footnote{For example,  \cite{Heckman1984} use the finite mixture model to provide an alternative method of accounting for unobserved heterogeneity in the analysis of the single-spell duration times of unemployed workers.
\cite{Keane1997} and \cite{Cameron1998} analyze a dynamic model of schooling and occupational choices with unobserved heterogeneous human capital.
Likewise, finite mixture models have been applied in health economics. \cite{Deb1997} develop a finite mixture negative binomial count model that accounts for the unobserved dispersion of medical care utilization by the elderly.
 \cite{Kamakura1989} and \cite{Andrews2003} model consumer segmentation in marketing in industrial organizations.} The theoretical properties of finite mixture models and examples of their applications have been discussed by numerous authors, such as \cite{Titterington1985}, \cite{Lindsay1995},  and \cite{McLachlan2000}.

The number of components is a crucial parameter in finite mixture models. In economic applications, the number of components often represents the number of unobservable types or abilities of individuals. Choosing an arbitrary number of parameters may lead to overestimation or underestimation of the level of heterogeneity. Using too few components may result in biased estimation, while using too many components can be computationally costly and the model becomes singular because of identification problems. Thus, developing a statistical procedure to determine the number of components is essential.

 Testing for the number of components in normal mixture regression models has long been an unsolved problem. The regularity conditions of the likelihood ratio test (LRT) for standard asymptotic analysis fail in finite mixture models because of issues such as non-identifiable parameters, the singularity of the Fisher information matrix, and the true parameter being on the boundary of the parameter space. Numerous papers have been written on the subject of the LRT for the number of components \citep[see, e.g.,][]{ghoshsen85book, chernofflander95jspi, lemdanipons97spl, chenchen01cjstat, chenchen03sinica, cck04jrssb, garel01jspi, garel05jspi, Chen14joe}, and the asymptotic distribution of the LRT statistic for general finite mixture models has been derived as a function of the Gaussian process \citep{dacunha99as, liushao03as, zhuzhang04jrssb, azais09esaim}. However, the key assumptions in these works are violated in cross-sectional normal regression models because normal mixtures possess additional undesirable mathematical properties: (i) the Fisher information for testing is not finite, (ii) the log-likelihood function is unbounded, and (iii) the second derivative of the density function for the mean parameter is linearly dependent on its first derivative for the variance parameter. \cite{Kasahara2015a} analyze the asymptotic distribution of the LRT statistics of a cross-sectional univariate finite mixture normal regression model, and \cite{Kasahara2019} develop a multivariate extension. \cite{Amengual2022} develop a score-type test for a cross-sectional normal mixture model.

This paper develops a likelihood ratio-based test for determining the number of components in finite mixture normal regression models with panel data, where outcome variables are conditionally independent across periods given the latent type within each unit. To the best of our knowledge, it is not known in the literature whether the aforementioned problems (i)--(iii) of the cross-sectional normal mixture still arise in the panel normal mixture. Furthermore, no likelihood-based test has yet been developed for testing the null hypothesis of an $M_0$-component model against an alternative $(M_0+1)$-component model for $M_0\geq 1$ in the panel normal regression mixture models with conditional independent errors.\footnote{\cite{Kasahara2014} develop a procedure to estimate a lower bound on the number of components consistently in finite mixture models in which each component distribution has independent marginals, which includes panel normal regression mixture models with conditionally independent errors as a special case.}

We show that problems related to (i) and (ii) arise, but the higher-order degeneracy of problem (iii) disappears in panel normal mixture models with conditional independence. Following \cite{chenli09as} and \cite{Kasahara2015a}, we consider a penalized likelihood ratio test (PLRT) and an expectation-maximization (EM) test to deal with the unboundedness and analyze the asymptotic distribution of the PLRT using reparameterization orthogonal to the direction in which the Fisher information matrix is singular. The likelihood ratio of an $(M_0+1)$-component model against the $M_0$-component model is approximated with local quadratic-form expansion with squares and cross-products of the reparameterized parameters. We demonstrate that the asymptotic null distributions of the penalized likelihood ratio test statistic (PLRTS) and the EM test statistic are characterized by the maximum of $M_0$ random variables, which we can easily simulate. Building on the PLRT and EM tests, we propose a sequential hypothesis testing approach for consistently estimating the number of components.
In simulations, our proposed PLRT and EM tests demonstrate favorable finite sample properties. Moreover, a sequential hypothesis testing approach accurately selects the correct number of components with high frequency, surpassing selection procedures based on the Akaike information criterion (AIC) and the Bayesian information criterion (BIC).

This paper makes several contributions to the literature. First, it analyzes the likelihood ratio-based test for the number of components in panel normal regression mixture models with conditional independence. \cite{Kasahara2015a} and \cite{Kasahara2019} analyze likelihood ratio-based tests for the number of components in cross-sectional univariate normal mixture regression models and multivariate normal mixture models, respectively. We demonstrate that the asymptotic distribution of the PLRT and EM tests for panel normal regression mixture models with conditionally independent errors differs from that of the univariate/multivariate normal mixture models in the aforementioned two papers because  higher-order dependency does not occur when the repeated measurement of the outcome variables is available in the panel data. Furthermore, we develop  a sequential hypothesis testing approach for consistently estimating the number of components.

Second, while it is well known that the log-likelihood function of normal mixture models is unbounded \citep{Hartigan1985}, it is unknown whether the related unboundedness problem arises in panel data. We show that the likelihood ratio test statistic is unbounded in panel normal mixture models with conditionally independent errors when the time dimension of panel data is finite. This unboundedness causes over-rejection of the LRT. We introduce a penalty function to prevent the likelihood ratio test statistics from being unbounded, using computational experiments to determine the data-driven penalty function. We also develop an R package {NormalRegPanelMixture} \citep{Hao2017} that contains the EM test module and asymptotic distribution simulation module.

Third, we conduct an empirical analysis of the number of production technology types using panel data from Japanese and Chilean manufacturing firms and provide strong evidence of substantial heterogeneity in production function coefficients across firms within narrowly defined industries. This is an important contribution to the literature on production function estimation, as most empirical applications assume the homogeneity of production function coefficients across firms using the standard production function estimation methods developed by \cite{olley1996dynamics}, \cite{levinsohn2003estimating}, and \cite{Ackerberg2015}. Our empirical finding suggests that it is essential to incorporate unobserved heterogeneity into the production function coefficients across firms in applications \citep{LiSasaki17arxiv, doraszelski2018measuring, Balat19mimeo, Kasahara2022esri}.

The EM test approach was introduced by \cite{lcm09bm} and \cite{chenli09as} to test homogeneity in finite mixture models. \cite{Li2010} develop an EM test for the null hypothesis of $M_0$ components applicable to general $M_0\geq 2$, and \cite{Kasahara2015a} propose an EM test for normal regression mixture models to test the null of $M_0\geq 2$. The EM approach is also applied to test homogeneity in multivariate mixtures \citep{Niu2011} and subgroup analyses \citep{Shen2015}. More recently, \cite{Liu2018} extend the EM test to mixtures of the general location-scale family distribution, and \cite{Kasahara2019} develop an EM test for multivariate normal mixture models. Building upon the literature, this paper develops an EM test for panel normal regression mixture models with conditionally independent errors. 

The identification and estimation of latent group structures in panel data has received attention in recent studies \citep{Kasahara2009, LinNg12jem, Bonhomme15ecma, AndoBai16jae, SuShiPhillips16ecma, LuSu17qe}. Finite mixture modeling provides a practical, model-based approach to determining unobserved group structures. Choosing the number of groups is often a prerequisite for classifying each individual's group membership. We can estimate the number of groups in panel data regression models by applying our proposed sequential hypothesis testing approach.

The rest of this paper is organized as follows.
In Section \ref{sec:2_model}, we define the finite normal mixture panel regression model.
In Section \ref{sec:3_test_1}, we demonstrate the PLRT for testing the homogeneity of a  normal mixture panel regression against a two-component model as a precursor to obtaining the general test of $M_0$ components.
Section \ref{sec:4_test} generalizes the result to testing $M_0$ components against $M_0 + 1$ components.
Section \ref{sec:5_EM} introduces the EM test for testing $M_0$ components against $M_0 + 1$ components. Section \ref{sec:local} derives the asymptotic distribution of the PLRT and EM tests under local alternatives.
Section \ref{sec:sht} develops a consistent estimator for the number of components based on sequential hypothesis testing.
Section \ref{sec:6_sim} presents the simulated results of the tests.
Section \ref{sec:7_app} provides an empirical application. In the following, $:=$ denotes ``equals by definition,'' and boldface letters denote vectors or matrices.  


%% file: sections/FM_2_finite_mixture.tex
We consider finite mixture normal regression models with panel data,  where the panel length $T$ is fixed and the number of cross-sectional observations $n$ goes to infinity. Define $\bs{w} := \{y_{t},\bs{x}_{t},\bs{z}_{t}\}_{t=1}^T$ with $y_t \in \R, \bs{x}_t \in \R^q, \bs{z}_t \in \R^p$.
Given $M \ge 2$, denote the density of an $M$-component model that represents the conditional density function of $\{y_t\}_{t=1}^T$ given $\{\bs{x}_{t},\bs{z}_{t}\}_{t=1}^T$
as
\begin{equation}\label{eq:fm}
	\begin{split}
	f_M(\bs{w}; \bs{\vartheta}_M) & =   \sum_{j=1}^M \alpha_j  f(\bs{w};\bs{\gamma},\bs{\theta}_j),
	\end{split}
\end{equation}
where $\bs{\vartheta}_M=(\bs{\alpha}^\top,\bs\theta_1^\top,...,\bs\theta_M^\top,\gamma^\top)^\top\in\Theta_{\bs{\vartheta}_M}$, $\bs\alpha^\top :=(\alpha_1,...,\alpha_{M-1})$, $\alpha_M=1-\sum_{j=1}^{M-1}\alpha_j$ 
and
\begin{equation}\label{eq:f1}
	f(\bs{w};\bs{\gamma},\bs{\theta}_j) = \prod_{t=1}^T \frac{1}{\sigma_j} \phi\left(\frac{ y_{t} -\mu_j - \bs{x}^\top_{t}\bs{\beta}_j - \bs{z}^\top_{t} \bs{\gamma} }{\sigma_j } \right)
\end{equation}
is the $j$-th component density function with $\mu_j \in \Theta_{\mu} \subset \R$ , $\sigma_j^2 \in \Theta_{\sigma} \subset \R_{++}$, $\bs\beta_j \in \Theta_{\bs\beta} \subset \R^{q}$,   $\bs\gamma \in \Theta_{\bs\gamma} \subset \R^p$, and  $\phi(t) = (2\pi)^{-1/2} \exp(-\frac{t^2}{2})$ is the standard normal probability density function. We collect the component-specific  parameters into 
 $\bs{\theta}_j := (\mu_j, \sigma_j^2, \bs{\beta}_j^\top)^\top  \in \Theta_{\bs\theta}$, and the regression coefficient $\bs\gamma$ for a vector $\bs z$ is assumed to be common across components. 

The number of components, denoted by $M_0$, is defined as the smallest integer $M$ such that the data density of $\bs{w}$ admits the representation (\ref{eq:fm}). Consider a random sample of $n$ with a panel length of $T$ independent observations $\{\bs{W}_{i}\}_{i=1}^n$, where  $\bs{W}_i = \{ (Y_{it},\bs{X}\t_{it},\bs{Z}\t_{it} )\t \}_{t=1}^T$ from a true $M_0$-component density $f_M(\bs{w};\bs\vartheta_{M_0}^*)$ defined in equation (\ref{eq:fm}) with $\bs\vartheta_{M_0}^*=((\bs{\alpha}^*)^\top,(\bs\theta_1^*)^\top,...,(\bs\theta_{M_0}^*)^\top,(\gamma^*)^\top)^\top$. The superscript $*$ signifies the true parameter value. Because component distributions can be identified only up to permutation, we assume that $\mu_1^*<\mu_2^*<\cdots <\mu_{M_0}^*$ for identification.\footnote{More generally, we may consider a lexicographical order: $\bs\theta_1^*<\bs\theta_2^*<\cdots<\bs\theta_{M_0}^*$.}

Our goal is to test
$$
H_0: M = M_0\	\text{ against }\ H_A: M = M_0 +1.
$$

%% file: sections/FM_3_m1.tex

We begin by developing the PLRT to test the null hypothesis $H_0: M=1$ against the alternative hypothesis $H_1: M=2$. Consider a random sample of $n$ with a panel length of $T$ independent observations $\{ \boldsymbol{W}_{i} \}_{i=1}^n$, where $\boldsymbol{W}_i = \{ (Y_{it}, \boldsymbol{X}_{it}^{\top}, \boldsymbol{Z}_{it}^{\top})^{\top} \}_{t=1}^T$, drawn from a true one-component density $f(\boldsymbol{w}; \boldsymbol{\gamma}, \boldsymbol{\theta})$ defined in equation (\ref{eq:f1}). Now consider a two-component mixture density function
\begin{equation*}
f_2(\boldsymbol{w}; \boldsymbol{\vartheta}_2) = \alpha f(\boldsymbol{w}; \boldsymbol{\gamma}, \boldsymbol{\theta}_1 ) + (1 - \alpha) f(\boldsymbol{w}; \boldsymbol{\gamma}, \boldsymbol{\theta}_2),
\end{equation*}
where $\boldsymbol{\vartheta}_2 = (\alpha, \boldsymbol{\theta}_1^{\top}, \boldsymbol{\theta}_2^{\top}, \boldsymbol{\gamma}^{\top})^{\top} \in \Theta_{\boldsymbol{\vartheta}_2}$ and $\alpha$ is the mixing probability of the first component. The two-component model can generate the true one-component density in two cases: (1) $\boldsymbol{\theta}_1 = \boldsymbol{\theta}_2 = \boldsymbol{\theta}^*$ and (2) $\alpha = 0$ or $1$. Consequently, the null hypothesis $H_0: M=1$ can be partitioned into two sub-hypotheses: $H_{01}: \boldsymbol{\theta}_1 = \boldsymbol{\theta}_2$ and $H_{02}: \alpha (1-\alpha) = 0$. The regularity conditions of the LRTS for a standard asymptotic analysis fail in finite mixture models: under $H_{01}$, $\alpha$ is not identified, and the Fisher information matrix for the other parameters becomes singular; under $H_{02}$, $\alpha$ is on the boundary of the parameter space, and either $\boldsymbol{\theta}_1$ or $\boldsymbol{\theta}_2$ is not identified.

As discussed in the introduction, analyzing the asymptotic distribution of the LRTS for the cross-sectional normal mixture is challenging because of its undesirable mathematical properties \citep[cf.][]{chenli09as}: (i) the Fisher information for testing $H_{02}$ is not finite, (ii) the log-likelihood function is unbounded \citep{Hartigan1985}, and (iii) the first-order derivative of $f_2(\boldsymbol{w}; \boldsymbol{\vartheta}_2)$ with respect to $\sigma_j^2$ is linearly dependent on its second-order derivative with respect to $\mu_j$. The presence of problems (i)--(iii) in panel normal mixture models with $T\geq 2$ is not well understood in the literature because, to the best of our knowledge, no studies have examined them so far.

Regarding problem (i), we note that the issue of the infinite Fisher information for testing $H_{02}$ also arises in the panel normal mixture model. For brevity, let us consider the case without $(\boldsymbol{X}, \boldsymbol{Z})$. The score for testing $H_{02}: \alpha=0$ takes the form  
\[
\left.\frac{\partial  f_2(\boldsymbol{W} ; \mu_1,\sigma_1^2,\mu_2,\sigma_2^2)}{\partial \alpha} \right|_{\alpha=0,\mu_2=\mu^*,\sigma_2^2=\sigma^{*2}}= 
 \frac{f(\boldsymbol{W};\mu_1,\sigma_1^2)}{f(\boldsymbol{W};\mu^*,\sigma^{*2} )}-1,
 \] 
where $f(\boldsymbol{W};\mu,\sigma^2)= \prod_{t=1}^T \phi((y_t-\mu)/\sigma)/\sigma$ and $\phi(\cdot)$ is the standard normal density function. When $\sigma_1^2> 2\sigma^{*2}$, $\mathbb{E}[\{f(\boldsymbol{W};\mu_1,\sigma_1^2)/f(\boldsymbol{W};\mu^*,\sigma^{*2} )-1\}^2]=\infty$. For more details, please refer to Proposition \ref{prop:infinite_fisher}. Because the infinite Fisher information causes difficulty in deriving the asymptotic distribution under $H_{02}$, this paper focuses on testing $H_{01}$. We define $ \Upsilon^*_1 := \{ (\alpha, \boldsymbol{\gamma}, \boldsymbol{\theta}_1, \boldsymbol{\theta}_2) \in  \Theta_{\vartheta_2} : \boldsymbol{\theta}_1 = \boldsymbol{\theta}_2 = \boldsymbol{\theta}^* \text{ and } \boldsymbol{\gamma} = \boldsymbol{\gamma}^* \}$, which is the subspace of $\Theta_{\vartheta_2}$ that corresponds to $H_{01}$. Note that because we focus on $H_{01}$, our test may not have power against the local alternatives with $\alpha_n\rightarrow 0$. We analyze the asymptotic distribution of the PLRTS under the contiguous local alternatives in Section \ref{sec:local}.

 Related to problem (ii), the LRTS in normal mixture models with panel data becomes unbounded as the sample size $n$ goes to $\infty$. Define the likelihood ratio statistic with respect to the true parameter under $H_0$ as
$$
LR_n^*(\boldsymbol{\vartheta}_{2}) := 2\left\{ \sum_{i=1}^n \log f_2(\boldsymbol{W}_i;\boldsymbol{\vartheta}_{2}) - \sum_{i=1}^n\log f(\boldsymbol{W}_i;\boldsymbol{\gamma}^*,\boldsymbol{\theta}^*)\right\},
$$
where $f_2$ is the density of the two-component finite mixture distribution in (\ref{eq:fm}) with $M=2$ and $((\boldsymbol{\gamma}^*)^{\top},(\boldsymbol{\theta}^*)^{\top})^{\top}$ is the true parameter value under $H_0$.   Let $\tilde{\boldsymbol{\vartheta}}_{2,n}$ be the maximum likelihood estimator for the two-component model, i.e., $\tilde{\boldsymbol{\vartheta}}_{2,n}=\arg\max_{\boldsymbol{\vartheta}_{2}\in \Theta_{\boldsymbol{\vartheta}_{2}}} LR_n^*(\boldsymbol{\vartheta}_{2})$. 
 \begin{proposition}\label{prop:unbounded_likelihood}
Suppose that the true model is described by the one-component model with $(\boldsymbol{\gamma},\boldsymbol{\theta})=(\boldsymbol{\gamma}^*,\boldsymbol{\theta}^*)$. Then, for any positive constant $M > 0$, $\Pr \Big( LR_n^*(\tilde{\boldsymbol{\vartheta}}_{2,n}) \le M \Big) \to 0$ as $n \to \infty$.
\end{proposition}

To deal with unboundedness, we consider a penalized maximum likelihood estimator (PMLE) as in \cite{Chen2009a} using the following penalty function: 
\begin{equation}\label{eq:penalty}
\tilde p_n({\boldsymbol{\vartheta}}_M):= \sum_{j=1}^M p_n(\sigma_j^2)\quad\text{with}\quad p_n(\sigma_j^2) := - a_n  \{ {{\sigma}_{0}^2}/{\sigma_j^2} + \log({\sigma_j^2}/{{\sigma}_{0}^2}) - 1\}.  
\end{equation}
This penalty function circumvents the problem of unbounded log likelihood by preventing a variance parameter estimate from nearing zero. The parameter $a_n$ is selected such that the penalty's impact becomes asymptotically negligible for the distribution of the PMLE. Refer to conditions C1--C3 in the proof of Proposition \ref{prop:vartheta_convergence_M}. 

Let 
\[
\hat{\boldsymbol{\vartheta}}_2 = \arg \max_{{\boldsymbol{\vartheta}}_2 \in \Theta_{\boldsymbol{\vartheta}_2} } \sum_{i=1}^n \log f_2(\boldsymbol{W}_i ; \boldsymbol{\vartheta}_2 )+\tilde p_n({\boldsymbol{\vartheta}}_2)
\]
denote the PMLE under the two-component model. Define a set of parameter values for the two-component density that generates the true one-component density by $\Theta^*_2:= \{ (\alpha,\boldsymbol{\gamma},\boldsymbol{\theta}_1,\boldsymbol{\theta}_2) \in  \Theta_{\vartheta_2}: \boldsymbol{\theta}_1 = \boldsymbol{\theta}_2 = \boldsymbol{\theta}^* \text{ and } \boldsymbol{\gamma} = \boldsymbol{\gamma}^*; \alpha=1  \text{ and } \theta_1=\theta^*; \alpha=0 \text{ and } \theta_2=\theta^* \}$.
$\boldsymbol{\theta}_1$ and $\boldsymbol{\theta}_2$ are component-specific parameters and $\boldsymbol{\gamma}$ is a parameter vector common across components. 
The following proposition establishes the consistency of the PMLE.

\begin{assumption}\label{assumption:1}
(a) $\boldsymbol{X}$ and $\boldsymbol{Z}$ have finite second moments, and $\text{Pr}(\boldsymbol{X}^\top \boldsymbol{\beta}  + \boldsymbol{Z}^\top \boldsymbol{\gamma}\neq \boldsymbol{X}^\top \boldsymbol{\beta}^* + \boldsymbol{Z}^\top \boldsymbol{\gamma}^*) > 0 $ for $ (\boldsymbol{\beta}^\top, \boldsymbol{\gamma}^\top)^\top \neq ( (\boldsymbol{\beta}^*)^\top, (\boldsymbol{\gamma}^* )^\top )^\top$. (b) $a_n> 0$ and $a_n = o(n^{1/4})$ in the penalty function (\ref{eq:penalty}).
\end{assumption}

\begin{proposition}\label{prop:vartheta_convergence}
Suppose that Assumption \ref{assumption:1} holds. Then, under the null hypothesis $H_0: M_0 = 1$, $\inf_{\boldsymbol{\vartheta}_2 \in \Theta_{2}^*} |\hat{\boldsymbol{\vartheta}}_2 - \boldsymbol{\vartheta}_2| \to_p 0$.
\end{proposition}
It should be noted that $f_2(\boldsymbol{w}; \boldsymbol{\vartheta}_2^*) = f(\boldsymbol{w};  \bs\gamma^*,\boldsymbol\theta^*)$ for any $\boldsymbol{\vartheta}_2^* \in \Theta_{2}^*$. Consequently, Proposition \ref{prop:vartheta_convergence} suggests that the PMLE $\hat{\boldsymbol{\vartheta}}_2$ converges in probability to a set of parameters for which the true density function $f(\boldsymbol{w}; \bs\gamma^*,\boldsymbol\theta^*)$ emerges within the space of two-component density functions.

For problem (iii), we show that in normal mixture models with panel data, the first-order derivative of $f_2(\bs{w}; \bs{\vartheta}_2)$ with respect to $\sigma_j^2$ is \textit{not} linearly dependent with its second-order derivative with respect to $\mu_j$ (See Proposition \ref{prop:expansion}(c)). 
Consequently, the panel mixture model (\ref{eq:fm}) with the component density function (\ref{eq:f1}) is strongly identifiable, and  the best rate of convergence for estimating the mixing distribution is $n^{-1/4}$ when the number of components is unknown \citep[cf.][]{chen95as}.  See Proposition \ref{prop:t2_distribution}(a).
 In contrast, the strong identifiability does not hold for the cross-sectional normal mixture, and its convergence rate becomes as slow as $n^{-1/8}$ when the number of components is over-specified  \citep[cf.][]{Kasahara2015a}.
 
As in any finite mixture models, however, the standard asymptotic analysis breaks down in testing $H_{01}: \bs{\theta}_1 = \bs{\theta}_2 = \bs{\theta}^*$  because $\alpha$ is not identified under $H_{01}$; in addition, the first-order derivative at the true value $\bs\vartheta^*_2 = (\alpha, (\bs\theta^*)\t, (\bs\theta^*)\t, (\bs\gamma^*)\t)\t$ is linear dependent as
\begin{align}\label{eq:dependence}
 \nabla_{\bs{\theta}_1} \log f_2(\bs{w};\bs\vartheta^*_2)
= \frac{\alpha}{1 - \alpha}   \nabla_{\bs{\theta}_2} \log f_2(\bs{w};\bs\vartheta^*_2).
\end{align}
To deal with this linear dependency, we analyze the asymptotic distribution of the LRTS by developing a higher-order approximation for the log-likelihood function.

To extract the direction of the Fisher information matrix singularity, we adapt the reparameterization approach by \cite{Kasahara2012} and consider the following one-to-one reparameterization of $\bs{\theta}_1$ and $\bs{\theta}_2$ given $\alpha$:
\begin{equation}\label{eq:m0_repar2}
\begin{pmatrix}
\bs{\lambda} \\
\bs{\nu}
\end{pmatrix} := \begin{pmatrix}
\bs{\theta}_1 - \bs{\theta}_2 \\
\alpha \bs{\theta}_1 + (1 - \alpha) \bs{\theta}_2
\end{pmatrix} \text{ so that }
\begin{pmatrix}
\bs{\theta}_1 \\
\bs{\theta}_2
\end{pmatrix} = \begin{pmatrix}
\bs{\nu} + ( 1- \alpha) \bs{\lambda} \\
\bs{\nu} - \alpha \bs{\lambda}
\end{pmatrix},
\end{equation}
where $\bs{\nu}$ and $\bs{\lambda}$ are both $(q+2) \times 1$ reparameterized parameter vectors with $\bs{\nu}=(\nu_\mu,\nu_\sigma,\bs{\nu}_{\bs\beta}\t)\t$ and  $\bs\lambda = (\lambda_{\mu},\lambda_{\sigma}, (\bs\lambda_{\beta})\t)\t= ( \mu_1 - \mu_2, \sigma_1^2 - \sigma_2^2, (\bs \beta_1 - \bs\beta_2)\t)\t$. We also write $\bs{\theta}$ and $\bs{\lambda}$ as $\bs{\theta} =(\theta_1,\theta_2,\theta_3,...,\theta_{q+2})\t:= (\mu,\sigma^2,\beta_1,...\beta_q)\t$ and $\bs{\lambda}=(\lambda_1,\lambda_2,\lambda_3,...,\lambda_{q+2})\t := (\lambda_\mu,\lambda_\sigma,\lambda_{\beta_1},...,\lambda_{\beta_q})\t$. 

This reparameterization is essential for analyzing the asymptotic distribution of the PLRTS in light of the linear dependency in  (\ref{eq:dependence}). The reparameterized parameter $\bs{\lambda}$ captures a deviation from the one-component model, where its first-order derivatives of the log density are identically equal to zero under $H_0: M=1$. Consequently, this reparameterization facilitates the derivation of an approximate quadratic-form criterion function, which is based on the fourth-order Taylor series approximation of the log-likelihood function, to characterize the asymptotic distribution of the LRTS.

  
Define the space for reparameterized parameters as
$$
\bs{\psi} := (\bs{\gamma}\t,\bs{\nu}\t,\bs{\lambda}\t)\t \in \Theta_{\bs\psi},
$$
where $\Theta_{\bs{\psi}} = \{ \bs{\psi}: \bs{\gamma} \in \Theta_{\bs\gamma}, \bs{\nu} + ( 1 - \alpha) \bs{\lambda} \in \Theta_{\bs\theta}, \bs{\nu} - \alpha \bs{\lambda} \in \Theta_{\bs\theta}\}.$
Under the null hypothesis $H_{01}: \bs{\theta}_1  = \bs{\theta}_2 = \bs{\theta}^*$, we have $\bs{\lambda} = (0,\ldots, 0)\t$ and $\bs{\nu} = \bs{\theta}^*$. We rewrite the reparameterized parameters under the null hypothesis as $(\bs{\psi}^{*})\t = ((\bs{\gamma}^*)\t, (\bs{\theta}^*)\t,0,\ldots,0)\t$.
Under the reparameterized parameter space, the density function and its logarithm are expressed as
\begin{align}\label{eq:repar}
	g(\bs{w};\bs{\psi},\alpha) & = \alpha f(\bs{w};\bs{\gamma},\bs{\nu}  + (1 - \alpha) \bs{\lambda}) + (1 - \alpha)  f(\bs{w};\bs{\gamma},\bs{\nu} - \alpha \bs{\lambda})\quad\text{and}\\
 l(\bs{w};\bs{\psi},\alpha)  & = \log g(\bs{w};\bs{\psi},\alpha).\nonumber
\end{align}
Write $\bs{\psi}$ as $\bs{\psi} = (\bs{\eta}\t,\bs{\lambda}\t)\t$ with $\bs{\eta} = (\bs{\gamma}\t,\bs{\nu}\t)\t$, where $\bs{\eta}^* = ((\bs{\gamma}^*)\t,(\bs{\nu}^* )\t)\t$ and $\bs{\lambda}^* = \bs{0}$. Denote the parameter spaces of $\bs\eta$ and $\bs\lambda$ by $\Theta_{\bs\eta}\subset \mathbb{R}^{p+q+2}$ and $\Theta_{\bs\lambda}\subset \mathbb{R}^{q+2}$, respectively.

Under this reparameterization, the first-order derivatives of the reparameterized log density with respect to the reparameterized parameters $\bs{\eta}$ are identical to those under the one-component model, and the first-order derivative with respect to $\bs{\lambda}$ is a zero vector:
\begin{equation}
 \label{eq:g_nabla}
\begin{split}
\nabla_{\bs{\eta}\t} l(\bs{w};\bs{\psi}^*,\alpha)   = \frac{\nabla_{(\bs{\gamma}\t,\bs{\theta}\t)\t} 	f(\bs{w};\bs\gamma^*,\bs\theta^*)}{f(\bs{w};\bs\gamma^*,\bs\theta^*)} \quad\text{and}\quad
 \nabla_{\bs{\lambda}} l(\bs{w};\bs{\psi}^*,\alpha)   = \bs{0} .
 \end{split}
\end{equation}
With $\nabla_{\bs\lambda} l(\bs{w};\bs{\psi}^*,\alpha) = 0$, the Fisher information matrix is singular, and the standard quadratic approximation fails.
Consequently, the information on $\bs{\lambda}$ is provided by the second-order derivative of $l(\bs{w};\bs{\psi},\alpha)$ with respect to $\bs{\lambda}$.
We use the second-order derivative with respect to $\bs{\lambda}$ to identify $\bs{\lambda}$:
\begin{equation}
 \label{eq:s2}
\nabla_{\bs{\lambda} \bs{\lambda}\t} l(\bs{w};\bs{\psi}^*,\alpha)  = \alpha (1-\alpha) \frac{ \nabla_{\bs{\theta} \bs{\theta}\t}f(\bs{w};\bs{\gamma}^*,\bs{\theta}^*)}{f(\bs{w};\bs{\gamma}^*,\bs{\theta}^*)}.
\end{equation}
When $\alpha$ is bounded away from $0$ and $1$,  the elements of $\nabla_{\bs{\lambda} \bs{\lambda}\t} l(\bs{W};\bs{\psi}^*,\alpha)$ are mean-zero random variables.

Note that unlike the cross-sectional models analyzed by \cite{Kasahara2015a},  there exists no collinearity between these first- and second-order derivatives for the panel models.  This distinction is indeed important, as it highlights the differences in the asymptotic distribution of the LRTS for the panel models compared with the cross-sectional models. The absence of collinearity between the first- and second-order derivatives in the panel models leads to different convergence rates and asymptotic properties.

Let $f^*$ and $\nabla f^*$ denote $f(\bs{W};\bs{\gamma}^*,\bs{\theta}^*)$ and $\nabla f(\bs{W};\bs{\gamma}^*,\bs{\theta}^*)$, respectively.
Define the vector $\bs{s}(\bs{W})$ as
\begin{align}\label{eq:s_1}
	 &\bs{s}(\bs{W}) = \begin{pmatrix}
		\bs{s}_{\bs{\eta}}(\bs W)   \\
		\bs{s}_{\bs{\lambda} \bs\lambda}(\bs W)
\end{pmatrix},
\quad \text{where}  \underset{( p + q +2 ) \times 1}{\bs{s}_{\bs{\eta} }(\bs W) }:= \frac{\nabla_{(\bs{\gamma}\t,\bs{\theta}\t)\t} f^*}{f^*} \quad \text{ and } \quad \underset{ (( q + 2 ) (q + 1) / 2 ) \times 1 }{\bs{s}_{\bs{\lambda} \bs\lambda}(\bs W) } :=\frac{\widetilde{\nabla}_{\bs\theta\bs\theta\t} f^*}{f^*}.
\end{align}
The term $\widetilde{\nabla}_{\bs\theta\bs\theta\t} f^*$ denotes the second-order derivatives of the density function $f^*$ with respect to the parameters $\bs{\theta}$. The coefficients $c_{jk}$ are used to adjust the scaling of these second-order derivatives. The function $\bs{s}(\bs{w})$ comprises the second-order derivatives of the log-likelihood function with respect to the reparameterized parameter $\bs\lambda$. This function, $\bs{s}_{\bs{\lambda} \bs\lambda}(\bs w)$, serves as a score function for identifying $\bs{\lambda}$. Consequently, $\bs{s}(\bs{w})$ is referred to as a score function.
An explicit expression for the score function $\bs{s}(\bs{w})$ can be derived using Hermite polynomials, as elaborated in Appendix \ref{sec:appendix_score_1}.  


Collect the relevant normalized reparameterized parameters and define $\bs{t}(\bs{\psi},\alpha)$ as
\begin{equation}
 \label{eq:t_1}
 \bs{t}(\bs{\psi},\alpha)  = \begin{pmatrix}
 \bs{t}_{\bs{\eta}} \\
 \bs{t}_{\bs{\lambda} }(\bs{\lambda},\alpha)
 \end{pmatrix}= \begin{pmatrix}
	 \bs{\eta} - \bs{\eta}^* \\
	  \alpha ( 1- \alpha)  \bs{v} (\bs{\lambda} )
 \end{pmatrix},
\end{equation}
where $v(\bs{\lambda})$  is a vector of unique elements of $\bs\lambda\bs\lambda\t$   given by
\begin{equation}
\label{eq:v} 
v(\bs{\lambda})   = (\lambda_1\lambda_1,...,\lambda_{q+2}\lambda_{q+2},\lambda_{1}\lambda_2,...,\lambda_{q+1}\lambda_{q+2})\t,
\end{equation}
the length of which is $q_\lambda:=(q+2)(q+3)/2$.

Let $L_n(\bs{\psi},\alpha):= \sum_{i=1}^n  l(\bs{W}_i;\bs{\psi}^*,\alpha)$ be the reparameterized log-likelihood function and define the normalized score vector
\[
\bs S_n := n^{-1/2} \sum_{i=1}^n {\bs s}(\bs W_i).
\]
Then, taking the fourth-order Taylor expansion of $L_n(\bs{\psi},\alpha)$  around $(\bs\psi^*,\alpha)$, we may write $2\{L_n(\bs{\psi},\alpha)- L_n(\bs{\psi}^*,\alpha)\}$ as a quadratic function of $\sqrt{n}\bs{t}(\bs{\psi},\alpha)$ as
\begin{align}
2\{L_n(\bs{\psi},\alpha)- &L_n(\bs{\psi}^*,\alpha)\} = 2(\sqrt{n}\bs{t}(\bs{\psi},\alpha))\t \bs S_n - (\sqrt{n}\bs{t}(\bs{\psi},\alpha))\t \bs {\mathcal{I}}_n(\sqrt{n}\bs{t}(\bs{\psi},\alpha)) + R_n(\bs{\psi},\alpha)  \label{eq:LR0} \\
 & \quad=  \bs G_n\t \bs{\mathcal{I}}_n \bs G_n -  \left[ \sqrt{n}\bs{t}(\bs{\psi},\alpha)- \bs G_n\right]\t \bs{\mathcal{I}}_n  \left[ \sqrt{n}\bs{t}(\bs{\psi},\alpha)- \bs G_n\right] + R_n(\bs{\psi},\alpha),\label{eq:LR1}
\end{align}
where  $\bs {\mathcal{I}}_n$ is the negative of the sample Hessian defined in the proof of Proposition \ref{prop:expansion} and $\bs G_n:=\bs{\mathcal{I}}_n^{-1}\bs S_n$.  Let $\bs{\mathcal{I}}=\E[\bs s(\bs W)\bs s(\bs W)\t]$. 

\begin{assumption}\label{assumption:2}
	(a) \(\bs{X}\) and \(\bs{Z}\) have finite $8$-th moments.  (b) $\E[\boldsymbol U \boldsymbol{U\t}]$ is non-singular, where $\boldsymbol U = [1, \boldsymbol{ X\t}, \boldsymbol{ Z\t}]\t$.
\end{assumption}

\begin{proposition}\label{prop:expansion} 
Suppose that Assumptions \ref{assumption:1} and \ref{assumption:2} hold. Then, under $H_0: M=1$, for $\alpha\in (0,1)$, (a) for any $\delta>0$, $\lim\sup_{n\rightarrow \infty} \Pr(\sup_{\bs\psi \in \Theta_{\bs\psi}: ||\bs\psi-\bs\psi^*||\leq \kappa} |R_n(\bs\psi,\alpha)| > \delta (1+||n\bs t(\bs\psi,\alpha)||^2)) \rightarrow 0$ as $\kappa\rightarrow 0$, (b) $ \bs S_n \overset{d}{\to} \bs S\sim N(0,\bs{\mathcal{I}})$, and (c) $\bs{\mathcal{I}}_n \overset{p}{\to} \bs{\mathcal{I}}$, where $\bs{\mathcal{I}}$ is finite and non-singular.
\end{proposition}
The non-singularity of $\bs{\mathcal{I}}$ in Proposition \ref{prop:expansion}(c) highlights the difference between the cross-sectional normal mixture and the panel data normal mixture models. In particular, as shown in equation (\ref{eq:s-vector}) in Appendix  \ref{sec:appendix_score_1}, the first-order derivative of $f_2(\bs w;\bs{\vartheta}_2)$ with respect to $\sigma_j^2$ is linearly independent of its second-order derivative with respect to $\mu_j$ when $T\geq 2$, which ensures that the higher-order degeneracy of problem (iii) does not arise. Intuitively, the availability of repeated observations within each individual unit provides better identification, even for over-parameterized models, and reduces the degree of higher-order degeneracy.

 The set of feasible values of $\sqrt{n} \bs{t}(\bs{\psi},\alpha)$ is given by the shifted and rescaled parameter space for $(\bs\eta,v(\bs\lambda))$ defined as $\Lambda_{n} := \sqrt{n} (\Theta_{\bs\eta}-\eta^*)   \times \sqrt{n} \alpha(1-\alpha)v( \Theta_{\bs\lambda})$, where $v(A):= \{ t \in \mathbb{R}^{q_\lambda}: t = v(\lambda) \text{ for some } \lambda \in A \subset \mathbb{R}^{q+2}\}$. Because $\Lambda_n/\sqrt{n} $ is  locally approximated by a cone $\Lambda:= \mathbb{R}^{p+q+2} \times v(\mathbb{R}^{q+2})$, we can apply Lemma 2 of \cite{Andrews1999}  to approximate the distribution of the supremum of the right-hand side of (\ref{eq:LR1}) as 
 \[
 \max_{\bs\psi\in \Theta_{\bs\psi}}2\{L_n(\bs{\psi},\alpha)-  L_n(\bs{\psi}^*,\alpha)\} \overset{d}{\rightarrow} \bs G \t \bs{\mathcal{I}}  \bs G-\inf_{\bs t\in \Lambda} (\bs t-\bs G)'\bs{\mathcal{I}} (\bs t-\bs G),
 \] 
 where $\bs G = \bs{\mathcal{I}}^{-1} \bs S\sim N(0,\bs{\mathcal{I}}^{-1})$. This allows us to characterize the asymptotic distribution of the LRTS.

For each $\alpha\in (0,1)$, define the reparameterized PMLE as
\begin{equation}
	  \label{eq:pmle}
\hat{\bs{\psi}}= \arg \max_{\bs{\psi} \in \Theta_{\bs\psi}}  L_n(\bs{\psi},\alpha)+ \sum_{j=1}^2 p_n(\sigma_j^2(\bs{\psi},\alpha)) 
\end{equation}
with  $\hat{\bs{\psi}} := (\hat{\bs{\gamma}}\t,\hat{\bs{\nu}}\t,\hat{\bs{\lambda}}\t)\t$,
where $\Theta_{\bs\psi}$ is defined as the space of $\bs{\psi}$ such that the $\bs{\vartheta}_2$ implied is in $\Theta_{\bs{\vartheta}}$ and $\sigma_j^2(\bs{\psi},\alpha)$ is the value of $\sigma_j$ implied by the value of $\bs{\psi}$ and $\alpha$ (e.g., $\sigma_1^2(\bs\psi,\alpha)=\nu_\sigma +(1-\alpha)\lambda_\sigma$). 

Let $(\hat{\bs{\gamma}}_0,\hat{\bs{\theta}}_0)$ be the one-component MLE that maximizes the one-component likelihood function $L_{0,n}(\bs{\gamma},\bs{\theta}) := \sum_{i=1}^n \log f(\bs{W}_i;\bs{\gamma},\bs{\theta}) $. Define the LRTS and the PLRTS  of testing $H_{01}$ with a small positivity constant $\epsilon$ on $\alpha$ as, respectively,
\begin{align}\label{eq:LR_def}
	LR_n &:= \max_{\alpha \in [\epsilon, 1- \epsilon]} 2 \{L_n(\hat{\bs{\psi}},\alpha) -  L_{0,n}(\hat{\bs{\gamma}}_0,\hat{\bs{\theta}}_0) \}\ \text{and}\  PLR_n  := LR_n+ \sum_{j=1}^2 p_n(\sigma_j^2(\hat{\bs{\psi}},\alpha)).
\end{align}
A hard bound is imposed on the values of $\alpha$ to avoid an issue of  the infinite Fisher information for testing $H_{02}$. However, the LRTS may have reduced power if the true value of $\alpha$ does not satisfy the constraint $[\epsilon, 1- \epsilon]$ given an ad hoc constant $\epsilon>0$. 
For this reason, in Section \ref{sec:5_EM}, we also develop the EM test, which does not impose a direct constraint on the value of $\alpha$. 
 
With $\bs{s}(\bs{W})$ in (\ref{eq:s_1}),  partition $\bs{\mathcal{I}}=\E[\bs s(\bs W)\bs s(\bs W)\t]$ and define
\begin{equation*}
\begin{split}
\bs{\mathcal{I}} = 
\begin{pmatrix}
 \bs{\mathcal{I}}_{\bs\eta}  & \bs{\mathcal{I}}_{\bs\eta\bs\lambda}  \\
 \bs{\mathcal{I}}_{\bs\lambda\bs\eta}  & \bs{\mathcal{I}}_{\bs\lambda\bs\lambda} 
\end{pmatrix},  
\quad \bs{\mathcal{I}}_{\bs{\eta}}  = \E[\bs s_{\bs\eta}(\bs W)  \bs s_{\bs\eta}(\bs W)\t], \quad \bs{\mathcal{I}}_{\bs\lambda\bs\eta}   = \E[\bs s_{\bs\lambda\bs\lambda }(\bs W)  \bs s_{\bs\eta}(\bs W)\t ], \quad \bs{\mathcal{I}}_{\bs\eta\bs\lambda}  = \bs{\mathcal{I}}_{\bs\lambda\bs\eta} \t, \\
 \bs{\mathcal{I}}_{\bs\lambda\bs\lambda}  =\E[\bs s_{\bs\lambda\bs\lambda }(\bs W)  \bs s_{\bs\lambda\bs\lambda}(\bs W)\t ],  \quad \bs{\mathcal{I}}_{\bs\lambda,\bs\eta}  = \bs{\mathcal{I}}_{\bs\lambda\bs\lambda}  - \bs{\mathcal{I}}_{\bs\lambda\bs\eta}  \bs{\mathcal{I}}_{\bs\eta} ^{-1} \bs{\mathcal{I}}_{\bs\eta\bs\lambda} , \quad\text{and}\quad \bs{G}_{\bs{\lambda},\bs\eta} := ( \bs{\mathcal{I}}_{\bs\lambda,\bs\eta}  )^{-1} \bs{S}_{\bs\lambda,\bs\eta} , 
\end{split}
\end{equation*}
where $\bs{S}_{\bs\lambda,\bs\eta}  \sim N(0,\bs{\mathcal{I}}_{\bs\lambda,\bs\eta} )$.
Define a set that characterizes the feasible values of $\sqrt{n}\bs{t}_\lambda(\bs{\lambda},\alpha)$ when $n \to \infty$ by the cone
\begin{equation*}
 \Lambda_{\bs\lambda} =  \Big\{ \sqrt{n} \alpha ( 1 - \alpha)  v(\bs{\lambda})  : \bs{\lambda} \in \Theta_{\bs\lambda} \Big\}.
\end{equation*}

Define $\hat{\bs{t}}_{\lambda}$ as
\begin{equation}\label{eq:t_lambda_def}
r_{\lambda}(\hat{\bs t}_{\bs \lambda}) = \inf_{\bs{t}_{\bs{\lambda}} \in \Lambda_{\bs\lambda}} r_{\lambda}(\bs{t}_{\bs{\lambda}}), \quad r_{\lambda} (\bs{t}_{\bs{\lambda}}) : = (\bs{t}_{\bs{\lambda}} - \bs{G}_{\bs{\lambda},\bs\eta})\t \bs{\mathcal{I}}_{\bs\lambda,\bs\eta}  (\bs{t}_{\bs{\lambda}} - \bs{G}_{\bs{\lambda},\bs\eta} ),
\end{equation}
where   $\hat{\bs{t}}_{\lambda}$ is a projection of a random Gaussian random variable $\bs{G}_{\bs{\lambda}}$ on a cone $\Lambda_{\bs\lambda}$.

 The following proposition establishes the asymptotic  distribution of the LRTS  or PLRTS under the null hypothesis $H_0: M=1$.

\begin{proposition}\label{prop:t2_distribution}
Suppose that Assumptions \ref{assumption:1} and \ref{assumption:2} hold. Under the null hypothesis $H_0: M_0 = 1$, (a) $\bs t(\hat{\bs{\psi}},\alpha)=O_p(n^{-1/2})$ for any $\alpha\in (0,1)$,  (b) $LR_n  \overset{d}{\to} (\hat{\bs{t}}_{\bs\lambda} )\t \bs{\mathcal{I}}_{\bs\lambda,\bs\eta}  \hat{\bs{t}}_{\bs\lambda}$ and  $PLR_n \overset{d}{\to} (\hat{\bs{t}}_{\bs\lambda} )\t \bs{\mathcal{I}}_{\bs\lambda,\bs\eta}  \hat{\bs{t}}_{\bs\lambda} +  \text{plim}_{n\rightarrow \infty} \sum_{j=1}^2 p_n(\sigma_j^2(\hat{\bs{\psi}},\alpha))$.
	\end{proposition} 
Proposition \ref{prop:t2_distribution}(a) implies that $\hat{\bs{\theta}}_j - \bs{\theta}^*=O_p(n^{-1/4})$ for $j=1,2$.  The $n^{1/4}$  convergence rate is a consequence of the linear dependency in (\ref{eq:dependence}), where the identification of the parameter $\bs{\theta}$ relies on the fourth-order Taylor approximation of the log-likelihood function. This rate is also the best convergence rate for an over-parameterized mixture under the strong identifiability condition \citep{chen95as}.  When we choose the penalty function  so that $\sum_{j=1}^2 p_n(\sigma_j^2(\hat{\bs{\psi}},\alpha))=o_p(1)$ under the null hypothesis of $M=1$,  $PLR_n$ has the same asymptotic null distribution as  $LR_n$.

%% file: sections/FM_4_m.tex
In this section, we build upon the analysis from the previous section and derive the asymptotic distribution of the PLRTS for testing the null hypothesis of $M_0$ components against an alternative of $(M_0+1)$ components, where $M_0 \ge 2$.

Consider a random sample of $n$ with a panel length of $T$ independent observations $\{\boldsymbol{W}_{i}\}_{i=1}^n$, where $\boldsymbol{W}_i = \{ (Y_{it},\boldsymbol{X}^\top_{it},\boldsymbol{Z}^\top_{it} )^\top \}_{t=1}^T$ from an $M_0$-component density $f_{M_0}(\boldsymbol{w}; \boldsymbol{\vartheta}_{M_0})$ defined in equation (\ref{eq:fm0}):
\begin{equation}\label{eq:fm0}
f_{M_0}(\boldsymbol{w}; \boldsymbol{\vartheta}_{M_0}^*)  = \sum_{j=1}^{M_0} \alpha^*_j  f(\boldsymbol{w};\boldsymbol{\gamma}^*,\boldsymbol{\theta}^*_j),
\end{equation}
where $\boldsymbol{\vartheta}_{M_0}^* = (\boldsymbol{\theta}_1^*,\boldsymbol{\theta}_2^*,\ldots,\boldsymbol{\theta}_{M_0}^*,\alpha_1^*,\ldots,\alpha_{M_0 -1}^*,\boldsymbol{\gamma}^*) \in \Theta_{\boldsymbol{\vartheta}_{M_0}}$ and $\alpha_{M_0}^*=1-\sum_{j=1}^{M-1}\alpha_j^*$. 

Let the density of the $(M_0 + 1)$-component model be defined by
\begin{equation}\label{eq:fm0_1}
		f_{M_0+1}(\boldsymbol{w}; \boldsymbol{\vartheta}_{M_0+1}) =\sum_{j=1}^{M_0+1} \alpha_j  f(\boldsymbol{w};\boldsymbol{\gamma},\boldsymbol{\theta}_j),
\end{equation}
where $\boldsymbol{\vartheta}_{M_0+1} = (\boldsymbol{\theta}_1,\boldsymbol{\theta}_2,\ldots,\boldsymbol{\theta}_{M_0+1},  \alpha_1,\ldots,\alpha_{M_0},\boldsymbol{\gamma}) \in \Theta_{\vartheta_{M_0+1}}$ as defined in (\ref{eq:fm0}). We assume $\mu_{1}^* < \mu_{2}^*, \ldots, < \mu_{M_0}^*$ in the true parameters for identification.

The $(M_0 + 1)$-component model (\ref{eq:fm0_1}) gives rise to the true density (\ref{eq:fm0}) in two cases: (i) two components have the same mixing parameter and (ii) one component has zero mixing proportion. Accordingly, we partition the null hypothesis of $H_0: M=M_0$ into two as $H_0 = H_{01}\cup H_{02}$, with $H_{01}: \boldsymbol{\theta}_h = \boldsymbol{\theta}_{h+1} = \boldsymbol{\theta}_{h}^*$ for some $h=1\ldots,M_0$ and $H_{02}: \alpha_h = 0$ for some $h=1,\ldots,M_0+1$.

We first analyze the infinite Fisher information problem for testing $H_{02}$. Partition  $H_{02}$ as $H_{02}=\cup_{h=1}^{M_0}H_{0,2h}$, where $H_{0,2h}: \alpha_h=0$.  
Define the subset of $\Theta_{\vartheta_{M_0+1}}$  corresponding to $H_{0,2h}$ as 
\[\begin{split}
	\Upsilon_{2h}^* = \{ \bs \vartheta_{M_0+1} \in \Theta_{\bs \vartheta_{M_0+1}} : \alpha_h = 0; (\alpha_j, \mu_j, \sigma_j)  = (\alpha_j^*, \mu_j^*, \sigma_j^*)\text{ for } j < h; \\ (\alpha_j, \mu_j, \sigma_j)  = (\alpha_{j-1}^*, \mu_{j-1}^*, \sigma_{j-1}^*)\text{ for } j > h\}. \end{split}
\]  
The score for testing $H_{0,2h}: \alpha_h = 0$ takes the form $\nabla_{\alpha_h} \log f_{M_0+1}(\bs W_i, \bs \vartheta_{M_0+1}) = [f(\bs W_i; \mu_h, \sigma_h^2 ) - f(\bs W_i; \mu_{M_0}^*, \sigma_{M_0}^{2*} )]  / f_{M_0}(\bs W_i, \bs \vartheta_{M_0}^*)$. 
Because $(\mu_h,\sigma_h^2)$ is not identified when $\alpha_h = 0$, the Fisher information matrix of the LRTS for testing $H_{0,2h}: \alpha_h = 0$ depends on the supremum of the variance of $\nabla_{\alpha_h} \log f_{M_0 + 1}(\bs W_i; \bs \vartheta_{M_0+1})$ over $\bs \vartheta_{M_0+1} \in \Upsilon_{2h}^*$.
The Fisher information is infinite unless there is an a priori restriction on the values of $\sigma_j^2$.
\begin{proposition}\label{prop:infinite_fisher}
	$\sup_{\vartheta_{M_0+1} \in \Upsilon_{2h}^*} \E [  \{ \nabla_{\alpha_h}  \log f_{M_0+1}  (\bs W_i, \vartheta_{M_0+1} )\}^2  ]  < \infty$ if and only if $\max \{ \sigma^2: \sigma \in \Theta_{\sigma} \} < 2 \max \{ \sigma_1^{2*},\ldots,\sigma_{M_0}^{2*}\}.$
\end{proposition}
Because the restriction on the values of $\sigma_j^2$ in Proposition \ref{prop:infinite_fisher} is difficult to justify and not easy to enforce in practice, we focus on testing $H_{01}$.

Partition  $H_{01}$ as $H_{01}=\cup_{h=1}^{M_0}H_{0,1h}$, where $H_{0,1h}: \bs\theta_{h}=\bs\theta_{h+1}$ with
$\mu_1<\cdots < \mu_h=\mu_{h+1}<\cdots < \mu_{M_0+1}$. We impose these inequality constraints on $\mu_j$ for component identification.
There are $M_0$ ways to describe the $M_0$ component null model in the space of $(M_{0}+1)$ component models, and each way corresponds to the null hypothesis of $H_{0,1h}: \bs\theta_{h}=\bs\theta_{h+1}$ for $h=1,2,..., M_0$. Testing $H_{0,1h}: \bs\theta_{h}=\bs\theta_{h+1}$ in the $M_0$-component null models is similar to  testing $H_{01}: \bs\theta_1=\bs\theta_2$ in the one-component null model in Section \ref{sec:3_test_1}.     
 
Define the subset of $\Theta_{\vartheta_{M_0+1}}$ corresponding to  $H_{0,1h}$ as
\begin{equation}\label{eq:Upsilon_1}
\begin{split}
	\Upsilon^*_{1h} : = \Big \{\bs{\vartheta}_{M_0+1} \in \Theta_{\vartheta_{M_0+1}}: \alpha_{h} + \alpha_{h+1} = \alpha_{h}^* \text{ and } \bs{\theta}_{h} = \bs{\theta}_{h+1} = \bs{\theta}_{h}^*  ; \bs{\gamma} = \bs{\gamma}^*;  \alpha_j = \alpha_{j}^* \qquad 
	\\
	\text{ and } \bs{\theta}_j = \bs{\theta}_j^{*} \text{ for } 1 \le j < h ; \alpha_j = \alpha_{j-1}^* \text{ and } \bs{\theta}_j = \bs{\theta}_{j-1}^* \text{ for } h +1 \le j \le M_0 + 1
\Big \}
\end{split}
\end{equation}
for $h = 1,\ldots,M_0$.   The set $\Upsilon^*_1 := \cup_{h=1}^{M_0} \Upsilon^*_{1h}$ corresponds to $H_{01}=\cup_{h=1}^{M_0}H_{0,1h}$.

Suppose that the null hypothesis of $M=M_0$ holds with the true density (\ref{eq:fm0}). Because any parameter in $\Upsilon^{*}_{1} = \cup_{h=1}^{M_0} \Upsilon^*_{1h}$ can generate the true density $f_{M_0}(\bs{w}; \bs{\vartheta}^*_{M_0}) =\sum_{j=1}^{M_0} \alpha_0^{j*}  f(\bs{w};\bs{\gamma}^*,\bs{\theta}_j^{*})$,  we need to restrict the estimators under the $(M_0 + 1)$-component  model to be in a neighborhood of $\Upsilon^*_{1h}$ to test $H_{0,1h}$.

Recall that $\mu_1^{*} < \mu_2^* \ldots < \mu_{M_0}^*$. Let $ \underline{\Theta}_{\mu}$ and $\overline{\Theta}_{\mu}$ denote the lower and upper bounds of $\Theta_{\mu}$, respectively. Define $D_1^* = [\underline{\Theta}_{\mu}, \frac{\mu_1^* + \mu_2^* }{2}] \times \Theta_{\beta} \times \Theta_{\sigma^2}$,  $D_h^* = [\frac{\mu_{h-1}^* + \mu_{h}^* }{2}, \frac{\mu_{h}^* + \mu_{h+1}^* }{2}] \times \Theta_{\beta} \times \Theta_{\sigma^2}$ for $h = 2,\ldots,M_0 - 1$, $D_{M_0 }^* = [\frac{\mu_{M_0-1}^* + \mu_{M_0}^* }{2}, \overline{\Theta}_{\mu}] \times \Theta_{\beta} \times \Theta_{\sigma^2}$.
  Then, $D_h^* \subset \Theta_{\theta}$ is a neighborhood containing $\theta_h^{*}$ but not $\theta_j^{*}$ for $j \neq h$.
For $h=1,\ldots M_0$, given a small positive constant $\epsilon>0$, define a restricted parameter space  $\bs{\Psi}_h^* \subset \Theta_{\vartheta_{M_0 + 1}}(\epsilon)$ as
\begin{equation}\label{Omega}
	\bs{\Psi}_h^* = \left\{ \begin{split}
		\alpha_1,\ldots,\alpha_{M_0+1} \in [\epsilon,1-\epsilon]; \sum_{j=1}^{M_0 + 1} \alpha_j = 1; \gamma \in \Theta_{\gamma} ;
        \theta \in \Theta_{\theta}: \bs{\theta}_j \in D_j^* \text{ for } j=1,\ldots,h-1; \\
		\bs{\theta}_h,\bs{\theta}_{h+1} \in D^*_h;
        \bs{\theta}_j \in D^*_{j-1} \text{ for } j=h+2,\ldots,M_0+1.
	\end{split}  \right\}
\end{equation}
Note that $\bs{\Psi}_h^* \cap \Upsilon_{1h}^* \neq \emptyset$ and $\bs{\Psi}_h^* \cap \Upsilon_{1l}^* = \emptyset$ if $h \neq l$, and $\cup_{h=1}^{M_0}\bs{\Psi}_h^*= \Theta_{\vartheta_{M_0 + 1}}(\epsilon)$.

Let $\hat{\bs{\Psi}}_h^* $ and $\hat{D}^*_h$ be consistent estimators of ${\bs{\Psi}}_h^* $ and ${D}^*_h$, which can be constructed from a consistent estimator of $\bs{\vartheta}_{M_0}^*$ in the $M_0$-component model. We test $H_{0,1h}: \bs\theta_{h}=\bs\theta_{h+1}$ by estimating the $(M_0+1)$-component model under the restriction that $\vartheta^{M_0+1}\in \hat{\bs{\Psi}}_h^*$.

For $h=1,2,..., M_0$,  define the local PMLE that maximizes the log-likelihood function of the $(M_0+1)$-component model under the constraint that $\bs\vartheta_{M_0+1}\in \hat{\bs\Psi}_h^*$ in (\ref{Omega}) by
\[
\hat{\bs{\vartheta}}_{M_0 + 1}^h   = \argmax_{\bs{\vartheta}_{M_0 + 1}  \in \hat{\bs\Psi}_h^* } L_{M_0+1,n}(\bs{\vartheta}_{M_0 + 1}) + \tilde p_n(\bs{\vartheta}_{M_0+1 }),
\] 
where  
$$
 L_{M,n}(\bs{\vartheta}_{M})  : = \sum_{i=1}^n \log f_{M}(\bs{W}_i;\bs{\vartheta}_{M})\quad\text{and}\quad \tilde p_n(\bs\vartheta_{M})  :=\sum_{j=1}^{M} p_n (\sigma_j^2;  \hat{\sigma}_{0,j}^2). 
 $$
 with
\begin{equation}\label{penalty2}
p_n(\sigma_j^2; \hat{\sigma}_{0,j}^2) := - a_n \{ {\hat{\sigma}_{0,j}^2}/{\sigma_j^2} + \log({\sigma_j^2}/{\hat{\sigma}_{0,j}^2}) - 1\},
\end{equation}
where $\hat{\sigma}_{0,j}^2$ is a root-$n$ consistent estimator of $\sigma_{0,j}^2$ from the $M_0$-component model under the null hypothesis. Because $\hat{\sigma}_j^2 - \sigma_{0,j}^2 = O_p(n^{-1/4})$ under the null hypothesis (cf. Proposition \ref{prop:t2_distribution}(a)),  $p_n(\hat\sigma_j^2; \hat{\sigma}_{0,j}^2) =o_p(1)$ when  $a_n$ is chosen to be $o(n^{1/4})$. 

Under $H_{0}: M=M_0$, $\bs\Psi_h^*$ contains a set of  parameters $\Upsilon_{1h}^*$ defined in (\ref{eq:Upsilon_1}) such that $f_{M_0 + 1}(\bs{w}; \bs{\vartheta}_{M_0 + 1})$ is equal to $f_{M_0}(\bs{w}; \bs{\vartheta}_{M_0 }^*)$ for any $\bs{\vartheta}_{M_0 +1}\in\Upsilon_{1h}^*$ and is therefore the density function from which the data are generated. These penalized likelihood estimators are consistent.
\begin{proposition}\label{prop:vartheta_convergence_M}
Suppose that Assumption \ref{assumption:1} holds. Then, under the null hypothesis $H_0: M = M_0$,  $\inf_{\bs{\vartheta}_{M_0+1} \in \bs\Psi_h^*} |\hat{\bs{\vartheta}}_{M_0+1}^h - \bs{\vartheta}_{M_0+1}| \overset{p}{\to} 0$ for $h=1,2,...,M_0$.
\end{proposition}

Consider the local PLRTS for testing $H_{0,1h}: \bs\vartheta_h = \bs\vartheta_{h+1}$ defined by 
\begin{equation*}
PLR^{M_0,h}_n := 2\{ L_{M_0+1,n}(\hat{\bs{\vartheta}}^h_{M_0 + 1}) + \tilde p_n(\hat{\bs{\vartheta}}^h_{M_0+1 }) - L_{M_0,n}(\hat{\bs{\vartheta}}_{M_0 })\}\quad\text{for $h=1,2,...,M_0$}.
\end{equation*}

The test utilizing the local PLRTS, denoted by $PLR^{M_0,h}_n$, possesses power solely against local alternatives within the restricted parameter space of $\bs{\Psi}_h^*$. To guarantee power against local alternatives over a wide range of directions, we consider the PLRTS characterized by the maximum of the local PLRTS for $h=1,...,M_0$, as defined by
\begin{equation}\label{eq:LR_M0_max}
PLR_n(M_0)  := \max\{PLR^{M_0,1}_n,PLR^{M_0,2}_n,...,PLR^{M_0,M_0}_n\}.
\end{equation}
Because $\Theta_{\vartheta_{M_0 + 1}}(\epsilon)= \cup_{h=1}^{M_0} \hat{\bs\Psi}_h^*$, $PLR_n(M_0)$ is identical to $\max_{\bs{\vartheta}_{M_0 + 1} \in \Theta_{\bs{\vartheta}_{M_0 + 1} }(\epsilon) } \{L_{M_0+1,n} (\bs{\vartheta}_{M_0 + 1})+ \tilde p_n(\bs{\vartheta}_{M_0+1 })\} -  L_{M_0,n}(\hat{\bs{\vartheta}}_{M_0 })$.

 To derive the asymptotic null distribution of $PLR_n(M_0)$,  collect the score vector for testing $H_{0,1h}$ for $h = 1,\ldots, M_0$ into one vector as
\begin{equation}
    \label{eq:s_tilde}
\tilde{\bs{s}}(\bs{W}) = \begin{pmatrix}
\tilde{\bs{s}}_{\bs{\eta}} (\bs{W})  \\
\tilde{\bs{s}}_{\bs{\lambda\lambda}} (\bs{W})
\end{pmatrix}, \ \text{ where } \underset{(M_0 +  p + q +1 ) \times 1}{\tilde{\bs{s}}_{\bs{\eta}}(\bs{W})}  = \begin{pmatrix}
\bs{s}_{\bs{\alpha}}(\bs{W})  \\
\bs{s}_{ (\bs{\gamma},\bs{\nu} ) } (\bs{W})
\end{pmatrix}\quad \text{ and } \tilde{\bs{s}}_{\bs{\lambda\lambda}} (\bs{W}) = \begin{pmatrix}
\bs{s}^1_{\bs{\lambda\lambda}} (\bs{W}) \\
\vdots \\
\bs{s}^{M_0}_{\bs{\lambda\lambda}} (\bs{W})
\end{pmatrix},
\end{equation}
where
\begin{equation}
 \begin{split}
    \bs{s}_{\bs{\alpha}} (\bs{W}) & = \begin{pmatrix}
    f(\bs{W};\bs{\gamma}^*,\bs{\theta}^*_1) -    f(\bs{W};\bs{\gamma}^*,\bs{\theta}^*_{M_0}) \\
    \vdots \\
       f(\bs{W};\bs{\gamma}^*,\bs{\theta}^*_{M_0-1}) -    f(\bs{W};\bs{\gamma}^*,\bs{\theta}^*_{M_0})
    \end{pmatrix} \Bigg / f_{M_0}(\bs{W}; \bs{\vartheta}_{M_0}^*) , \\
    \bs{s}_{ (\bs{\gamma},\bs{\nu} ) } (\bs{W}) &  = \sum_{j=1}^{M_0} \alpha_j^*
    \nabla_{ (\bs{\gamma},\bs{\nu} ) }  f(\bs{W};\bs{\gamma}^*,\bs{\theta}^*_{j}) / f_{M_0}(\bs{W}; \bs{\vartheta}_{M_0}^*), \\ 
    \bs{s}^h_{\bs{\lambda\lambda}} (\bs{W}) & = \widetilde{ \nabla}_{\bs\theta_h\bs\theta_h\t}  f(\bs{W};\bs{\gamma}^*,\bs{\theta}^*_h) / f_{M_0}(\bs{W}; \bs{\vartheta}_{M_0}^*)\quad\text{for $h=1,2,...,M_0$} ,
    \end{split}
\end{equation}
with
$\widetilde{ \nabla}_{\bs\theta_h \bs\theta\t_h}  f(\bs{W};\bs{\gamma}^*,\bs{\theta}^*_h)  := (c_{11} \nabla_{\theta_{h1}\theta_{h1}} f^*,...,c_{(q+2)(q+2)}\nabla_{\theta_{h,q+2}\theta_{h,q+2}} f^*,c_{12}\nabla_{\theta_{h1}\theta_{h2}} f^*,...,c_{(q+1)(q+2)}\nabla_{\theta_{h,q+1}\theta_{h,q+2}} f^*)\t$  
for $\bs{\theta}_h:=(\theta_{h1},\theta_{h2},\theta_{h3},...,\theta_{h,q+2})\t:=(\mu_h,\sigma_h^2,\beta_{h1},...,\beta_{hq})\t$ and  $c_{jk}=1/2$ for $j\neq k$ and $c_{jk}=1$ for $j=k$.   
 Define \begin{equation}
\begin{split}
    \tilde{\bs{\mathcal{I}}} := \E[\bs{\tilde{s}}(\bs{W}) \bs{\tilde{s}}(\bs{W})^\top], \quad  \tilde{\bs{\mathcal{I}}}_{\bs{\eta}} := \E[\bs{\tilde{s}}_{\bs{\eta}}(\bs{W})  \bs{\tilde{s}}_{\bs{\eta}}(\bs{W})^\top ], \quad \tilde{\bs{\mathcal{I}}}_{\bs{\lambda} \bs{\eta}} := \E[\bs{\tilde{s}}_{\bs{\lambda\lambda}}(\bs{W}) \bs{\tilde{s}}_{\bs{\eta}}(\bs{W})^\top ], \\
    \tilde{\bs{\mathcal{I}}}_{\bs{\eta} \bs{\lambda}}  := \tilde{\bs{\mathcal{I}}}_{\bs{\lambda} \bs{\eta}}^\top, \quad \tilde{\bs{\mathcal{I}}}_{\bs{\lambda\lambda}} := \E[\bs{\tilde{s}}_{\bs{\lambda\lambda}}(\bs{W}) \bs{\tilde{s}}_{\bs{\lambda\lambda}}(\bs{W})^\top ], \quad \tilde{\bs{\mathcal{I}}}_{\bs{\bs\lambda,\bs\eta}} := \tilde{\bs{\mathcal{I}}}_{\bs{\lambda \lambda}} - \tilde{\bs{\mathcal{I}}}_{\bs{\lambda \eta}} \tilde{\bs{\mathcal{I}}}_{\bs\eta}^{-1} \tilde{\bs{\mathcal{I}}}_{\bs{\eta \lambda}}.
\end{split}
\end{equation} 
Then, the asymptotic distribution of the normalized score function is given by
\[
 \tilde{\bs S}_n: = \frac{1}{\sqrt{n}} \sum_{i=1}^n \tilde{\bs s}(\bs W_i) \overset{d}{\to} \tilde {\bs S} \sim N(\bs 0,   \tilde{\bs{\mathcal{I}}}),
\]
where, in view of (\ref{eq:s_tilde}),  $ \tilde {\bs S} $ may be partitioned as $ \tilde {\bs S} =( \tilde {\bs S}_{\bs\eta}\t, \tilde {\bs S}_{\bs\lambda\bs\lambda}\t)\t$ with $n^{-1/2} \sum_{i=1}^n \tilde{\bs s}_{\bs\eta}(\bs W_i)\overset{d}{\to} \tilde {\bs S}_{\bs\eta}$ and $n^{-1/2} \sum_{i=1}^n \tilde{\bs s}_{\bs\lambda\bs\lambda}(\bs W_i)\overset{d}{\to} \tilde {\bs S}_{\bs\lambda\bs\lambda}$ .

Let $\bs{\tilde{S}}_{\bs\lambda,\bs\eta} := (\bs{{S}}_{\bs\lambda,\bs\eta}^1, \ldots, \bs{{S}}_{\bs\lambda,\bs\eta}^{M_0} )^\top:= \bs{\tilde{S}}_{\bs\lambda\bs\lambda} - 
\tilde{\bs{\mathcal{I}}}_{\bs\lambda\bs\eta} \tilde{\bs{\mathcal{I}}}_{\bs\eta} ^{-1} \bs{\tilde{S}}_{\bs\eta} 
 \sim N(0,   \tilde{\bs{\mathcal{I}}}_{\bs\lambda,\bs\eta})$ be a $\R^{M_0 (q+2)(q+1)/2}$-valued random vector. For $h=1,2,...,M_0$,  define $\tilde{\bs{\mathcal{I}}}_{\bs\lambda,\bs\eta}^h := \E[\bs{{S}}_{\bs\lambda,\bs\eta}^h ( \bs{{S}}_{\bs\lambda,\bs\eta}^h)^\top]$ and $\bs{{G}}_{\bs\lambda,\bs\eta}^h := ({\bs{\mathcal{I}}}_{\bs\lambda,\bs\eta}^h)^{-1} \bs{{S}}_{\bs\lambda,\bs\eta}^h$.

Define $\hat{\bs t}^h_{\bs\lambda}  $ analogously to $\hat{\bs t}_{\bs\lambda} $ as
\begin{equation}
     \label{eq:t_h}
\begin{split}
r^h_{\bs\lambda} (\hat{\bs{t}}^h_{\bs{\lambda}} ) = \inf_{{\bs{t}}^h_{\bs{\lambda}}\in \Lambda_{\bs\lambda}} r^h({\bs{t}}^h_{\bs{\lambda}}); \quad
r^h_{\bs\lambda} ({\bs{t}}^h_{\bs{\lambda}}) := ({\bs{t}}^h_{\bs{\lambda}} - \bs{{G}}^h_{\bs\lambda,\bs\eta} )^\top  {\bs{\mathcal{I}}}_{\bs\lambda,\bs\eta}^h({\bs{t}}^h_{\bs{\lambda}} - \bs{{G}}^h_{\bs\lambda,\bs\eta} )\quad\text{for $h=1,2,...,M_0$}.
\end{split}
\end{equation}

The local quadratic-form approximation of the log-likelihood function  $LR^{M_0,h}_n$ around $\Upsilon^*_{1h} \subset \Theta_{\vartheta_{M_0 + 1}}$  has an identical structure to the approximation that we derive in Section \ref{sec:3_test_1} in testing $H_{01}$ in the test of homogeneity.  Consequently, we can show that $PLR^{M_0,h}_n\overset{d}{\to} (\hat{\bs{t}}^h_{\bs\lambda})^\top \bs{\mathcal{I}}^h_{\bs\lambda,\bs\eta}  \hat{\bs{t}}^h_{\bs\lambda}$.  Then, given (\ref{eq:LR_M0_max}), the asymptotic null distribution of the PLRTS for testing $H_{01}$ is given by the maximum over $(\hat{\bs{t}}^h_{\bs\lambda})^\top \bs{\mathcal{I}}^h_{\bs\lambda,\bs\eta}  \hat{\bs{t}}^h_{\bs\lambda}$s for $h=1,2,..., M_0$.

\begin{assumption}\label{assumption:3}
    (a) $\alpha_j^*  \in (\epsilon, 1 - \epsilon)$ for $j = 1,\ldots, M_0$. (b) $\bs{\tilde{\mathcal{I}}}$  is non-singular. (c) $a_n$ in (\ref{penalty2}) satisfies  $a_n =O(1)$. 
\end{assumption}

\begin{proposition}\label{prop:tm0_distribution}
  Suppose that Assumptions \ref{assumption:1}--\ref{assumption:3} are satisfied. Then, under the null hypothesis $H_0: M = M_0$,   $PLR_n({M_0})  \overset{d}{\to}   \max \{(\hat{\bs{t}}^1_{\bs\lambda})^\top \bs{\mathcal{I}}^1_{\bs\lambda,\bs\eta}  \hat{\bs{t}}^1_{\bs\lambda} ,   \ldots, (\hat{\bs{t}}^{M_0}_{\bs\lambda})^\top \bs{\mathcal{I}}^{M_0}_{\bs\lambda,\bs\eta}  \hat{\bs{t}}^{M_0}_{\bs\lambda}  \}$.
\end{proposition}
The asymptotic null distribution of $PLR_n({M_0})$ is non-standard, but it is straightforward to simulate the random variable from the asymptotic null distribution using the estimates.  Specifically, 
we simulate a draw of $\bs{\tilde{S}}_{\bs\lambda,\bs\eta}= (\bs{{S}}_{\bs\lambda,\bs\eta}^1, \ldots, \bs{{S}}_{\bs\lambda,\bs\eta}^{M_0} )^\top$ from $N(0,   \hat{\tilde{\bs{\mathcal{I}}}}_{\bs\lambda,\bs\eta})$, where $\hat{\tilde{\bs{\mathcal{I}}}}_{\bs\lambda,\bs\eta}$ is a sample analogue estimator of ${\tilde{\bs{\mathcal{I}}}}_{\bs\lambda,\bs\eta}$. Then, compute $\bs{{G}}_{\bs\lambda,\bs\eta}^h = (\hat{\bs{\mathcal{I}}}_{\bs\lambda,\bs\eta}^h)^{-1} \bs{{S}}_{\bs\lambda,\bs\eta}^h$ and obtain $\hat{\bs t}^h_{\bs\lambda}$ analogously to (\ref{eq:t_h}) using an estimator of $\bs{\mathcal{I}}_{\bs\lambda,\bs\eta}^h$ for $h=1,...,M_0$, and a simulated random draw is computed as $  \max \{(\hat{\bs{t}}^1_{\bs\lambda})^\top \hat{\bs{\mathcal{I}}}^1_{\bs\lambda,\bs\eta}  \hat{\bs{t}}^1_{\bs\lambda} ,   \ldots, (\hat{\bs{t}}^{M_0}_{\bs\lambda})^\top \hat{\bs{\mathcal{I}}}^{M_0}_{\bs\lambda,\bs\eta}  \hat{\bs{t}}^{M_0}_{\bs\lambda}  \}$. Appendixes \ref{sec:appendix_score_1} and \ref{sec:appendixb:score_m} present an expression for the score functions using Hermit polynomials.

%% file: sections/FM_5_EM.tex
This section develops an EM test used for testing the hypothesis $H_0: M=M_0$ against the alternative hypothesis $H_A: M = M_0 + 1$. A key limitation of the PLRT, as discussed in the previous section, is that the computation of mixing probabilities, denoted as $\alpha_j$, is subject to a hard constraint, which is dictated by an arbitrary choice of bounds. The EM test, in contrast, circumvents the need to impose an explicit constraint on the $\alpha_j$ values. It achieves this by performing a limited number of EM steps, starting from a predetermined set of $\alpha_j$ values. The EM test approach offers certain advantages, including computational simplicity and less stringent assumptions. 

Let $\mathcal{T}$ be a finite set of numbers in $(0,0.5]$ with $0.5\in \mathcal{T}$,  let $p(\tau)\leq 0$ be a penalty term that is continuous in $\tau$, $p(0.5)=0$, and let $p(\tau)\rightarrow -\infty$ as $
\tau$ goes to 0.  Specifically, we choose  $$ p(\tau) := \log(2\min\{\tau,1-\tau\}).$$

For each $\tau_0 \in \mathcal{T}$, let  $\tau^{(1)} (\tau_0) = \tau_0$, and define the restricted PMLE   by
\[
   \bs{\vartheta}_{M_0 + 1}^{h(1)}(\tau_0)  = \underset{\bs{\vartheta}_{M_0 + 1} \in \Theta_{\bs{\vartheta}_{M_0 +1}}^h(\tau)  }{\arg\max}  {PL}_n(\bs{\vartheta}_{M_0 + 1},\tau_0), 
   \]
where $\Theta_{\bs{\vartheta}_{M_0 +1}}^h(\tau_0) := \{ \bs\theta \in \hat{\bs{\Psi}}_h : \alpha_h / (\alpha_h + \alpha_{h+1} ) =  \tau_0 \}$ and
$$
 {PL}_n(\bs{\vartheta}_{M_0 + 1},\tau):=
   L_{M_0+1,n}(\bs{\vartheta}_{M_0 + 1}) + \tilde p_n(\bs{\vartheta}_{M_0+1}) +p(\tau).
   $$

Starting from $(\bs{\vartheta}_{M_0 + 1}^{h(1)}(\tau_0),\tau^{h(1)}(\tau_0 ) )$ with $\tau^{h(1)}(\tau_0 )=\tau_0$,  update $\bs{\vartheta}_{M_0 + 1}^{h(k)}(\tau_0) $ and $\tau^{h(k)}(\tau_0)$ by the following generalized EM algorithm.
 Denote the estimators  after the $k$-th round of EM algorithm iteration by $\vartheta_{M_0 + 1}^{h(k)}$ and  $\tau^{h(k)}$. 
In the E-step, for $i=1,\ldots,N$ and $j = 1,\ldots, M_0 + 1 $, compute the weight for observation $i$ and type $j$ as
\begin{equation}
\begin{split}
w_{i j}^{(k)} &  = \left \{ \begin{split} &  \alpha_j^{(k)} f( \bs{W}_{i};\bs{\gamma}^{(k)},\bs{\theta}_j^{(k)}) / f_{M_0+1}(\bs{W}_{i} ;\bs{\vartheta}_{M_0 + 1}^{h(k)}(\tau_0)),j=1,\ldots,h-1,\\
& \alpha_{j-1}^{(k)} f( \bs{W}_{i};\bs{\gamma}^{(k)},\bs{\theta}_j^{(k)})/ f_{M_0+1}(\bs{W}_{i} ;\bs{\vartheta}_{M_0 + 1}^{h(k)}(\tau_0)) ,j=h+2,\ldots,M_0+1,\end{split}\right.\\
 w_{i h}^{(k)}  & =  \tau^{h(k)} \alpha_h^{(k)} f(\bs{W}_{i};\bs{\gamma}^{(k)},\bs{\theta}_h^{(k)}) /  f_{M_0+1}(\bs{W}_{i};\bs{\vartheta}_{M_0 + 1}^{h(k)}(\tau_0)), \\
 w_{i,h+1}^{(k)} & =  (1 - \tau^{h(k)}) \alpha_h^{(k)} f(\bs{W}_{i};\bs{\gamma}^{(k)},\bs{\theta}_{h+1}^{(k)})/ f_{M_0+1}(\bs{W}_{i};\bs{\vartheta}_{M_0 + 1}^{h(k)}(\tau_0)),
\end{split}
\end{equation}
where, for brevity, we drop the superscript $h$ and its dependency on $\tau_0$ from the notations, such as in $w_{ij}^{h(k)}(\tau_0)$.

In the M-step,  we update $\bs{\alpha}$ and $\tau$  by
\begin{align*}
 \alpha_j^{(k+1)}  & = \frac{1}{n} \sum_{i=1}^{n} {w_{ij}^{(k)}} \quad \text{for } j = 1,\ldots,M_0+ 1 \quad \text{and}\\
\tau^{h(k+1)} & = \arg\min_{\tau} \left\{\sum_{i=1}^{n} {w_{ih}^{(k)}}\log(\tau) +\sum_{i=1}^{n} {w_{i,h+1}^{(k)}}  \log(1-\tau) + p(\tau)
 \right\}.
\end{align*} 
We also update $\bs{\theta}_j$ and $\bs\gamma$ as 
\begin{align*}
(\sigma_j^{(k+1)})^2 & = \arg\min_{\sigma_j^2} \left\{ \sum_{i=1}^n w_{i}^{j(k)} \sum_{t=1}^T  (y_{it} - \mu_j^{(k+1)} -  \bs{z}_{it}^\top \bs{\gamma}^{(k+1)} - \bs{x}_{it}^\top \bs{\beta}_j^{(k+1)}  )^2 +  p_n(\sigma_j^2) \right\},\\
\bs{\gamma}^{(k+1)} & = \left(\sum_{i=1}^n \sum_{t=1}^T \bs{z}_{it} \bs{z}_{it}^\top \right)^{-1} \left(\sum_{i=1}^n \sum_{t=1}^T \bs{z}_{it} \left(y_{it} - \sum_{j=1}^{M_0+1} w_{ij}^{(k)} \tilde{\bs{x}}_{it}^\top \begin{pmatrix}
\mu_j^{(k)}\\
\bs{\beta}_j^{(k)}
\end{pmatrix} \right) \right) , \quad\text{and} \\
\begin{pmatrix}
\mu_j^{(k+1)}\\
\bs{\beta}_j^{(k+1)}
\end{pmatrix} & = \left(\sum_{i=1}^n w_{ij}^{(k)} \sum_{t=1}^T \tilde{\bs{x}}_{it} \tilde{\bs{x}}_{it}^\top \right)^{-1} \left(\sum_{i=1}^n w_{ij}^{(k)} \sum_{t=1}^T \tilde{\bs{x}}_{it} (y_{it} - \bs{z}_{it}^\top \bs{\gamma}^{(k+1)}) \right),
\end{align*}
where $\tilde{\bs{x}}_{it} = (1,\bs{x}_{it}^\top)^\top$.
In the updating procedure,  $\bs{\vartheta}_{M_0 + 1}^{h(k+1)}(\tau_0)$ is not restricted to be in  $\hat{\bs{\Psi}}_h^*$.

For each $\tau_0 \in \mathcal{T}$ and each step $k$, define
\begin{equation}
	M_n^{h(k)}(\tau_0) : = 2\left\{ PL_n(\bs{\vartheta}_{M_0 + 1}^{h(k)}(\tau_0),\tau^{h(k)}(\tau_0))  - L_{M_0,n}(\hat{\bs{\vartheta}}_{M_0}) \right\}.
\end{equation}

With a predetermined finite number  $K$, define the \textit{local}   EM test statistic by taking the maximum of $M_n^{h(k)}(\tau_0)$ across different values of $\tau_0$ as
\begin{equation}
EM_n^{h} : = \max \{ M^{h(K)}_n(\tau_0) : \tau_0 \in \mathcal{T} \}.
\end{equation}
The test statistic $EM_n^{h}$ tests $H_{0,1h}: \bs\theta_h=\bs\theta_{h+1}$ and has a power against the local alternative that splits the $h$-th component of  the null $M_0$-component model into two different components. To achieve power against a wide range of local alternatives,
we consider the   EM test statistic that takes the maximum of $M_0$ local   EM test statistics:
\begin{equation}
EM_n(M_0) := \max \{ EM_n^{1(K)},\ldots, EM_n^{M_0(K)} \}.
\end{equation}

\begin{proposition}\label{prop:KS18_prop4}
 Suppose that Assumptions \ref{assumption:1}--\ref{assumption:3} hold and $\{0.5\}\in \mathcal{T}$. Then, under the null hypothesis $H_0: M = M_0$,  for any finite $K$, $EM_n (M_0)\overset{d}{\to}  \max \{(\hat{\bs{t}}^1_{\bs\lambda})^\top \bs{\mathcal{I}}^1_{\bs\lambda,\bs\eta}  \hat{\bs{t}}^1_{\bs\lambda} ,   \ldots, (\hat{\bs{t}}^{M_0}_{\bs\lambda})^\top \bs{\mathcal{I}}^{M_0}_{\bs\lambda,\bs\eta}  \hat{\bs{t}}^{M_0}_{\bs\lambda}  \}$. 
\end{proposition} 

Therefore, the asymptotic null distribution of the EM test statistic $EM_n(M_0)$ is the same as that of the PLRTS.

%% file: sections/FM_6_simulated.tex
In this section, we examine the finite sample performance of the EM test and the PLRT by simulation. We  test $H_0 : M = M_0  $ against $H_1 : M = M_0 + 1$ for the model with $M_0=2$ and $3$. 

\subsection{Choice of penalty function}\label{sec:penalty}
%
 
 We develop a data-dependent empirical formula for $a_n$ by selecting a formula that ensures that the empirical rejection probabilities match the nominal size (5\%) across various null models and sample sizes, as reported in Table \ref{table:parameter} in Appendix D. Specifically, for the model without conditioning variables, we derive the following data-dependent empirical formula for testing the null hypotheses of $M_0=1, 2, 3, 4$: \begin{equation}\label{eq:a_n}
 a_n  = \begin{cases}
   \left({1 +\exp\left\{  \frac{\hat{\rho}_1^{M_0}}{ \hat{\rho}_4^{M_0}} + \frac{\hat{\rho}_2^{M_0}}{\hat{\rho}_4^{M_0}}  \frac{1}{T} + \frac{\hat{\rho}_3^{M_0}}{\hat{\rho}_4^{M_0}}  \frac{1}{n}  \right\} }\right)^{-1}, & M_0 = 1 \\
   \left({1 +\exp\left\{  \frac{\hat{\rho}_1^{M_0}}{ \hat{\rho}_4^{M_0}} + \frac{\hat{\rho}_2^{M_0}}{\hat{\rho}_4^{M_0}}  \frac{1}{T} + \frac{\hat{\rho}_3^{M_0}}{\hat{\rho}_4^{M_0}}  \frac{1}{n} + \frac{\hat{\rho}_5^{M_0}}{\hat{\rho}_4^{M_0}}  \log\left(\frac{\omega(\bs{\vartheta}_{M_0};M_0)}{1-\omega(\bs{\vartheta}_{M_0};M_0)}\right) \right\} }\right)^{-1}, &M_0 = 2,3,4,
\end{cases}
\end{equation}
where     
$\omega(\bs{\vartheta}_{M_0};M_0)$ is the misclassification probability as defined in \cite{Melnykov2010}  for each of the null models. The parameters 
$\hat{\rho}_1^{M_0}$, $\hat{\rho}_2^{M_0}$, $\hat{\rho}_3^{M_0}$, $\hat{\rho}_4^{M_0}$, and $\hat{\rho}_5^{M_0}$ are chosen as follows.   
Across different null models, sample sizes, and various candidate values of $a_n$, we estimate the empirical rejection probabilities at the $5\%$ significance level by simulations and denote them by $\hat s$. For example, when testing $H_0: M_0=2$, we repeatedly simulate the  500 datasets under each of  the $48$ null model parameters and sample sizes $(N,T,\alpha,\mu,\sigma) \in \{100,500\} \times \{ 2,5,10\} \times \{(0.5,0.5),(0.2,0.8)\} \times \{ (-1,1),(-0.5,0.5), (-0.5,0.8)\} \times \{ (1,1), (1.5,0.75), (0.8,1.2)\}$ and test the null hypothesis of $H_0: M_0=2$ by the EM test using one of the six values of $a_n\in \{0.01, 0.05,0.1,0.2, 0.3, 0.4\}$.  
For each of the $108 \times 6= 648$ combinations of the parameter values, sample sizes, and $a_n$ values, let $\hat s$ denote the fraction of simulated datasets that lead to the rejection of the null hypothesis at the $5\%$ significance level.  Using these $648$ observations of $\{\hat{s},N,T,\omega(\bs{\vartheta}_{2};2), a_n\}$, we run the following regression: 
\[\begin{split}
& \log\left(\frac{\hat{s}}{1- \hat{s}}\right) -  \log\left(\frac{ {0.05}}{1-  {0.05}}\right)  \\ & = \begin{cases}
	\rho_1^{M_0} + \rho_2^{M_0}\frac{1}{T} + \rho_3^{M_0}\frac{1}{n} + \rho_4^{M_0} \log\left(\frac{ {a_n}}{1 -  {a_n}}\right) , & M_0 = 1 \\
	 \rho_1^{M_0} + \rho_2^{M_0}\frac{1}{T} + \rho_3^{M_0}\frac{1}{n} + \rho_4^{M_0} \log\left(\frac{ {a_n}}{1 -  {a_n}}\right) + \rho_5^{M_0} \log\left(\frac{\omega(\bs{\vartheta}_{M_0};M_0)}{1 - \omega(\bs{\vartheta}_{M_0};M_0)}\right),  & M_0 = 2,3,4,
\end{cases}
\end{split}
\]  
where $\hat{\rho}_1^{M_0}$, $\hat{\rho}_2^{M_0}$, $\hat{\rho}_3^{M_0}$, $\hat{\rho}_4^{M_0}$, and $\hat{\rho}_5^{M_0}$ in (\ref{eq:a_n}) denote the corresponding estimates. Table \ref{tab:rho_estim} in the Appendix reports the estimates.
Note that the data-dependent formula (\ref{eq:a_n}) is obtained by setting $\hat s=0.05$ and solving for $a_n$ in the above equation.  

For the model with conditioning variables, we find that the value of $a_n$ that gives accurate Type I errors is sensitive to the dimension of covariates, and developing a data-dependent empirical formula for
$a_n$ is difficult. Consequently, we choose a constant value of $a_n$ that depends only on the number of components $M_0=1, 2, 3,$ and $4$ as follows:
$a_n=   0.1617 \text{ if } M_0 =1 ; a_n= 0.0025 \text{ if } M_0 = 2; a_n= 0.0567 \text{ if } M_0 = 3;  a_n= 0.4858 \text{ if } M_0 = 4; \text{ and }  a_n= 0.5  \text{ if } M_0 \geq 5.$  These penalty terms for the regression with covariates are chosen by averaging the predictions of the penalty function for the null parameters used in the simulations. For example, the penalty term for $M_0 = 2$ is chosen by generating $a_n$ using the formula for all of the combinations of $(N, T,\alpha, \mu,\sigma)$ in Table \ref{table:parameter} for $M_0 = 2$ and taking the average across the predicted values of $\hat a_n$. 
For $M_0 \ge 5$, we use the parametric bootstrap method to obtain the critical values for our empirical application, where we set $a_n = 0.5$.


\subsection{Simulation results}

Table \ref{table:size_test} displays the simulated Type I error rates for the EM test when we examine the null hypothesis $H_0: M=2$ against the alternative hypothesis $H_1: M=3$. A total of 2,000 repetitions are used for the asymptotic distribution, and 1,000 repetitions are used for the bootstrap distribution. Moreover, the PLRT with simulated critical values is considered.

The table presents the results for four distinct null models, as explained in the table's footnote. Utilizing the asymptotic distribution, the EM test sizes generally approximate the nominal 5\% level. Nonetheless, the test may be undersized in instances where $T \ge 5$. Furthermore, the test size is larger when the mixing proportions are equal ($\bs{\alpha} = (0.5,0.5)$) than when they are unequal ($\bs{\alpha} = (0.2, 0.8)$). The bootstrapped EM test demonstrates satisfactory performance.

For the PLRT, 2,000 repetitions are conducted, and the results are reported for cases where a constraint is applied to $\alpha_j \in [\epsilon, 1 - \epsilon]$ with $\epsilon=0.1$. The value of $a_n$ for the PLRT  is chosen to be 10 times larger than its value for the EM test. The findings suggest that the PLRT  is slightly oversized.

Table \ref{table:power_test} reports the rejection frequency of testing $H_0: M_0 = 2$ under 12 alternative three-component mixture models, as elaborated in the table's footnote. For both the EM test and the PLRT, the test power is greater when the distances between $\mu_j$s are larger and equal, such as $(\mu_1,\mu_2,\mu_3)=(-1,0,1)$ or $(-1.5,0,1.5)$, as opposed to unbalanced distances such as $(-1,0,2)$ or $(-0.5,0,1.5)$. The power is also improved when the mixture probabilities are equal ($\bs{\alpha} = (1/3,1/3,1/3)$) relative to when they are unequal ($\bs{\alpha} = (1/4,1/2,1/4)$). The power increases with both the time dimension $T$ and the cross-sectional sample size $N$. Reflecting a larger actual rejection frequency of the PLRT  under  $H_0: M_0 = 2$ in Table \ref{table:size_test}, the power  of the PLRT  is often higher than that of the EM test, although the EM test sometimes has higher power, especially when  the mixing probabilities are unequal. 

Table \ref{table:size_M3} displays the simulated Type I error rates of the EM test using the asymptotic distribution for testing $H_0: M_0=3$ against $H_1: M_0=4$. Six null models are considered with varying $(\alpha_1,\alpha_2, \alpha_3)$ and $(\mu_1,\mu_2, \mu_3)$ values. The EM test generally yields accurate Type I errors.

The Type I error rates of the EM test with conditioning variables under the null $M_0 = 2$ are examined using 500 repetitions. The results presented in Table \ref{table:size_regressor} indicate a slightly oversized test for small samples with $(N, T) = (200, 2)$, but overall, the finite sample properties are satisfactory.

In our empirical application examining production function heterogeneity in Japan and Chile, we find evidence that the number of components is frequently greater than 5 when we sequentially apply our EM test to estimate the number of components. We also investigate the performance of the SHT using the EM test in comparison with the AIC and the BIC when the data are generated from a five-component model in a realistic setting. Specifically, we simulate 100 datasets from the estimated five-component model of the Chilean textile industry in our empirical application and apply these three methods to select the number of components in each of the 100 datasets. Here, we apply the EM test at the 5\% significance level to sequentially test the null hypothesis $H_0: M=M_0$ for $M_0=1,2,...,7$, and we determine the number of components to be $M_0$ when we fail to reject $H_0: M=M_0$, as in (\ref{M-hat}).

Table \ref{table:simulated_BIC} presents the frequencies at which the three methods select the number of components in this simulation. The table demonstrates that the proposed SHT selects the correct number of components 72\% of the time, while it underestimates the true number of components 25\% of the time. Conversely, the AIC overestimates the number of components 86\% of the time, and the BIC underestimates the number of components by selecting a four-component model 41\% of the time and accurately estimates the number of components 58\% of the time. Overall, in this simulation, our proposed SHT approach outperforms both the AIC and the BIC.

\begin{landscape}
	
\begin{table}[H] \centering
	\begin{threeparttable}
	   \caption{Sizes (in \%) of the EM Test and the PLRT of $H_0 : M_0 = 2$ Against $H_A: M_0 = 3$ at the $5 \%$ Level}
	   \label{table:size_test}
	   \begin{tabular}{lllllllllllllllllll} 
		\toprule
        & \multicolumn{6}{c}{\textbf{EM Test}}       & \multicolumn{6}{c}{\textbf{EM Test}}       & \multicolumn{6}{c}{\textbf{PLRT}}      \\
		\midrule
        & \multicolumn{6}{c}{Asymptotic}       & \multicolumn{6}{c}{Parametric Bootstrap}              & \multicolumn{6}{c}{Asymptotic}              \\
$T$ & \multicolumn{2}{c}{3}     & \multicolumn{2}{c}{5}     & \multicolumn{2}{c}{8}   & \multicolumn{2}{c}{3}     & \multicolumn{2}{c}{5}     & \multicolumn{2}{c}{8}  & \multicolumn{2}{c}{3}     & \multicolumn{2}{c}{5}     & \multicolumn{2}{c}{8} \\
$N$       & 200 & 400 & 200 & 400 & 200 & 400  & 200 & 400 & 200 & 400  & 200 & 400  & 200 & 400  & 200 & 400  & 200 & 400 \\ 
\midrule
$(A,C)$ & 5.3 & 4.8 & 4   & 3.8 & 4   & 2.95   & 4.6 & 6.2 & 6.2 & 4.6 & 4.8 & 5.2 & 7.7 & 6.6 & 6.75   & 6.7 & 6.5 & 5.5 \\
$(A,D)$ & 5.9 & 4.9 & 5   & 5   & 4.45   & 4   & 5.4 & 4.8 & 5.2 & 5.6 & 5   & 5.8 & 5.4 & 5.15   & 5.1 & 6.1 & 5.55   & 6.2 \\
$(B,C)$ & 3.8 & 2.5 & 3.45   & 3.05   & 3.6 & 3.25   & 3.6 & 5.6 & 4.2 & 5.2 & 4   & 5.4 & 6.25   & 5.45   & 6.45   & 6.05   & 5.05   & 6   \\
$(B,D)$ & 4.8 & 4.6 & 3.5 & 3.15   & 3.55   & 3.95   & 3.6 & 3.6 & 5.8 & 4   & 6.2 & 4.6 & 2.35   & 4.4 & 3.9 & 4.85   & 4.95   & 5.2 \\ 
\bottomrule
\end{tabular}
\smallskip
\begin{tablenotes}
	\item[1]  $A$ and $B$ refer to, respectively, $(\alpha_1,\alpha_2) = (0.5,0.5)$ and $(0.2,0.8)$, while $C$ and $D$ refer to $(\mu_1,\mu_2) = (-1,1)$ and $(-0.5,0.5)$, respectively.\\
	\item[2] The variance is set to $(\sigma_1,\sigma_2) = (0.8,1.2)$.   The asymptotic simulations are based on 2,000 repetitions and the bootstrap simulation is based on 1,000 repetitions.
\end{tablenotes}
	\end{threeparttable}
\end{table}

\end{landscape}

\begin{table}[H]
	\centering
	\begin{threeparttable}
		\caption{Powers (in \%) of the EM Test and the PLRT of $H_0 : M_0 = 2$ Against $H_A: M_0 = 3$ at the $5\%$ Level }
		\label{table:power_test}
		\begin{tabularx}{\textwidth}{l @{\extracolsep{\fill}} l |rrrr|rrrr }
			\toprule
			&  & \multicolumn{4}{c|}{A}& \multicolumn{4}{c}{B}\\ \cmidrule{3-10}
			& N& \multicolumn{2}{c}{100} & \multicolumn{2}{c|}{500} & \multicolumn{2}{c}{100} & \multicolumn{2}{c}{500} \\ \midrule
			& T& 2  & 5  & 2  & 5  & 2  & 5  & 2  & 5  \\ \midrule
			 &&\multicolumn{8}{c}{\textbf{EM test}} \\ \midrule
			& $(C,G)$& 20.9 & 81.6 & 57.6& 100.0& 20.5& 82.7& 62.6 & 100.0\\
			& $(C,H)$& 49.2 & 99.9& 99.9 & 100.0& 38.4 & 98.7 & 98.8& 100.0\\
			& $(C,I)$& 12.1& 20.4& 18.0& 62.6& 10.6& 20.4 & 16.8& 65.8\\
			& $(D,G)$& 77.9& 100.0& 100.0& 100.0& 86.5 & 100.0& 100.0& 100.0\\
			& $(D,H)$& 57.4& 100.0& 100.0& 100.0& 42.8& 100.0& 100.0& 100.0\\
			& $(D,I)$& 16.0 & 59.5& 31.8& 99.9 & 13.8& 70.8 & 40.3 & 100.0\\
			& $(E,G)$& 93.0 & 100.0& 100.0& 100.0& 94.0& 100.0& 100.0& 100.0\\
			& $(E,H)$& 83.8 & 100.0& 100.0& 100.0& 70.7 & 100.0& 100.0& 100.0\\
			& $(E,I)$& 25.7& 97.0 & 80.2& 100.0& 30.7& 96.8& 83.1& 100.0\\
			& $(F,G)$& 99.9& 100.0& 100.0& 100.0& 100.0& 100.0& 100.0& 100.0\\
			& $(F,H)$& 93.5 & 100.0& 100.0& 100.0& 85.3 & 100.0& 100.0& 100.0\\
			& $(F,I)$& 40.8& 99.9& 98.2& 100.0& 52.1 & 100.0& 99.5& 100.0\\ \midrule \midrule
			 &&\multicolumn{8}{c}{\textbf{PLRT}} \\ \midrule  
& $(C,G)$&22.7&85.1&56.5&100.0&23.2&82.6&58.7&100.0\\
& $(C,H)$&57.1&100.0&99.8&100.0&43.2&99.7&99.1&100.0\\
& $(C,I)$&12.0&21.0&12.4&66.1&11.3&22.1&12.4&69.3\\
& $(D,G)$&79.9&100.0&100.0&100.0&87.6&100.0&100.0&100.0\\
& $(D,H)$&65.3&100.0&100.0&100.0&49.1&100.0&100.0&100.0\\
& $(D,I)$&14.8&63.6&28.6&100.0&13.6&75.2&36.7&100.0\\
& $(E,G)$&91.5&100.0&100.0&100.0&93.7&100.0&100.0&100.0\\
& $(E,H)$&86.8&100.0&100.0&100.0&75.9&100.0&100.0&100.0\\
& $(E,I)$&28.7&97.2&77.7&100.0&33.0&97.9&85.5&100.0\\
& $(F,G)$&99.9&100.0&100.0&100.0&100.0&100.0&100.0&100.0\\
& $(F,H)$&96.4&100.0&100.0&100.0&89.5&100.0&100.0&100.0\\
& $(F,I)$&45.7&100.0&98.4&100.0&57.2&100.0&99.8&100.0\\
			  \bottomrule
		\end{tabularx}
		\smallskip
  \begin{flushleft}
 Notes:  $A$ and $B$ refer to $(\alpha_1,\alpha_2,\alpha_3) = (1/3,1/3,1/3)$ and $(1/4,1/2,1/4)$, respectively; 
			 $C, D, E,$ and $F$ refer to $(\mu_1,\mu_2,\mu_3) = (-0.5,0,1.5), (-1,0,1),(-1,0,2),$ and  $(-1.5,0,1.5)$, respectively; and  $G, H,$ and $I$ refer to $(\sigma_1,\sigma_2,\sigma_3) = (0.6,0.6,1.2), (0.6,1.2,0.6),$ and $(1,1,1)$, respectively. 
\end{flushleft} 
	\end{threeparttable}
\end{table}

\begin{table}[H]
	\centering
	\caption{Sizes (in \%) of the EM Test of $H_0 : M_0 = 3$ Against $H_A: M_0 = 4$ at the $5 \%$ Level}
	\label{table:size_M3}
	\begin{threeparttable}
		\begin{tabularx}{\textwidth}{l @{\extracolsep{\fill}} rrrrrr}
			\toprule
			& (A,C) & (A,D) & (A,E) & (B,C) & (B,D) & (B,E) \\ 
			\midrule
			100,2 & 5.95& 5.15& 5.05& 5.05   & 5.85   & 4.40\\
			500,2 & 5.60& 5.55& 5.25& 5.10& 5.65   & 4.05   \\
			100,5 & 4.30& 6.00& 4.20& 5.15   & 5.10& 5.70\\
			500,5 & 4.20& 4.55& 3.95& 4.50& 4.15   & 4.15  \\ 
			\bottomrule
		\end{tabularx}\smallskip
  \begin{flushleft}
 Notes:  $A$ and $B$ refer to $(\alpha_1,\alpha_2, \alpha_3) = (1/3,1/3,1/3)$ and $(0.25, 0.5, 0.25)$, respectively, and $C,D,$ and $E$ refer to $(\mu_1,\mu_2, \mu_3) = (-4, 0, 4)$, $(-4, 0, 6),$ and $(-6, 0, 6)$, respectively.
	The variance is set to $(\sigma_1,\sigma_2,\sigma_3) = (0.75, 1.5, 0.75)$.   The asymptotic simulations are based on 2,000 repetitions and the bootstrap simulation is based on 1,000 repetitions.
\end{flushleft} 
	\end{threeparttable}
\end{table}

\begin{table}[H] \centering
	\begin{threeparttable}
		\caption{Sizes of the EM Test of $H_0 : M_0 = 2$ Against $H_A: M_0 = 3$ with Conditioning Variables}
		\label{table:size_regressor}
		\small
		\begin{tabularx}{\textwidth}{lrrrrrrrr}
			\toprule
			& $(A,C,E)$ & $(A,C,F)$ &  $(A,D,E)$ & $(A,D,F)$ &  $(B,C,E)$ & $(B,C,F)$ &  $(B,D,E)$ & $(B,D,F)$ \\ 
			$(N,T)$ & \\
			\midrule
			$(200,2)$ & 8.4   & 8.2& 7.4 & 8.8& 8.6   & 8.2& 7.4 & 3.6\\
			$(500,2)$ & 4.6   & 3.2& 3.2 & 2.2& 4.8   & 4.8& 3.6 & 3.6\\
			$(200,5)$ & 4.0  & 1.8& 3.0   & 2.6& 2.2   & 2.0 & 2.2 & 3.2\\
			$(500,5)$ & 2.2   & 1.2& 1.6 & 1.4& 3.0  & 2.0 & 1.8 & 2.0 \\
			\bottomrule
		\end{tabularx}
	\end{threeparttable}\smallskip
	  \begin{flushleft}
 Notes:  $A$ and $B$ refer to  $(\mu_1,\mu_2) = (-1,1)$ and $(-0.5,0.5)$, respectively, and $C$ and $D$ refer to  $(\beta_1,\beta_2) = (1,1)$ and $(-1,1)$, respectively.
	  $E$ and $F$ refer to  $(\sigma_1,\sigma_2) = (0.3,0.1)$ and $(0.1,0.1)$, respectively. The mixing proportion is set to $(\alpha_1,\alpha_2) = (0.2,0.8)$.  The asymptotic simulations are based on 500 repetitions.
\end{flushleft} 
\end{table}

\begin{table}[H]
	\centering
	\caption{Frequency of the Number of Components with the Simulated Data}
	\label{table:simulated_BIC} 
	\begin{threeparttable}
	\begin{tabularx}{\textwidth}{@{\extracolsep{\fill}}lccccccc}
	\toprule
	$M$ & 1 & 2 & 3 & 4 & 5 & 6 & 7   \\
	\midrule
	SHT with EM test & 0 & 0 & 0 & 0.26 & 0.72 & 0.02 & 0    \\
	AIC & 0 & 0 & 0 & 0.01 & 0.13 & 0.31 & 0.55 \\
	BIC & 0 & 0 & 0 & 0.41 & 0.58 & 0.01 & 0   \\
	\bottomrule
	\end{tabularx}
	\begin{tablenotes}
		\footnotesize
	\item[1] The data are generated using the estimated parameters based on the Chilean textile industry with five components and panel length $T=3$, where $(\alpha_1,\alpha_2, \alpha_3, \alpha_4,\alpha_5) = (0.16076522, 0.32454077, 0.09025875, 0.35478905, 0.06964622)$, $(\mu_1,\mu_2, \mu_3, \mu_4,\mu_5) = (-1.241241, -0.33803875,  0.4480291,  0.52379553, 1.4139465)$, $(\beta_1,\beta_2, \beta_3, \beta_4,\beta_5)= (0.451833, -0.05988709, -0.2453261, -0.03106076, 0.2053708)$, and $(\sigma_1,\sigma_2, \sigma_3, \sigma_4,\sigma_5) = (0.9933480, 0.4585760, 0.9954302, 0.4116855, 0.1863346)$. 
	We use the panel length and sample size that are equal to those in the dataset, i.e.,  $n = 196$ and $T = 3$.
	
	\item[2] The results are based on 100 repetitions. 
	
	\item[3] Each cell indicates the proportion of times that the model selection indicates an $M$-component model. 
	\end{tablenotes}
	\end{threeparttable}
	
\end{table}


%% file: sections/FM_7_application.tex
In this section, we conduct an empirical application of our proposed test for the number of components in a finite mixture production function model, the identification of which is analyzed in \cite{Kasahara2022esri}. Specifically, we estimate the number of types of input elasticities in production functions using panel data from Japanese publicly traded firms in the machinery industry and data from Chilean manufacturing firms.

\subsection{Production function and first-order condition} 
%

Consider the input and output panel data of $n$ firms over $T$  years, $\{\{Y_{it}, V_{it}, L_{it},$ and $K_{it}\}_{t=1}^T\}_{i=1}^N$, where $Y_{it}$, $V_{it}$,  $L_{it}$, and  $K_{it}$ represent the output, intermediate input, labor, and capital of firm $i$ in year $t$, respectively. We denote the logarithms of the corresponding variables by lowercase letters as $(y_{it}, v_{it}, l_{it}, k_{it})$, with, for example, $y_{it} = \log(Y_{it})$.

We use a finite mixture specification to capture the unobserved heterogeneity in a firm's input elasticities.  We are interested in testing the number of production technology types. 
Assume that there are $M$ discrete types of production technologies
and define the latent random variable $D_i \in \{ 1,2,\ldots, M\}$ to represent the production technology type of firm $i$. If $D_i  = j$, then firm $i$ is of type $j$.  The population proportion of type $j$ is denoted by $\alpha_j=\Pr(D=j)$. The production function for type $j$ is Cobb--Douglas, and 
the output is related to inputs as
\begin{equation}
Y_{it}=\exp{(\epsilon_{it})} F_t^j(V_{it}, L_{it},K_{it}, \omega_{it})  \label{prod},
\end{equation} 
with
\begin{align*}
F_t^j(V_{it},  L_{it},K_{it},\omega_{it})  :=  \exp(\gamma_t^j + \omega_{it}) V_{it}^{\delta_{v,j}} L_{it}^{\delta_{\ell,j}}K_{it}^{\delta_{k,j}},
\end{align*}
where $\gamma_t^j$ represents the aggregate productivity shock of type $j$ in year $t$, $\omega_{it}$ is the serially correlated productivity shock, and $\epsilon_{it}$ is the idiosyncratic productivity shock.

We assume that an intermediate input $V_{it}$ is flexibly chosen by firm $i$ after observing the aggregate shock $\gamma_t^j$ and the serially correlated productivity shock $\omega_{it}$. The variable $\epsilon_{it}$ represents a mean-zero i.i.d. random variable, the realization of which is unknown when the intermediate input $V$ is selected. Denote the information available to a firm for making decisions on $V_{it}$ by $\mathcal{I}_{it}$. Denote the information available to a firm for making decisions on $V_{it}$ by $\mathcal{I}_{it}$.

To identify the intermediate input elasticity of the production function, we introduce the following assumptions (cf. \cite{Kasahara2022esri}).
\begin{assumption}\label{app:A1}
	(a) Each firm belongs to one of $M$ types, and the probability of being type $j$ is given by $\alpha_j = P(D_i = j)$ with  $\sum_{j=1}^M \alpha_j = 1$.
(b) For the $j^{th}$ type of production technology at time $t$, the output is expressed in terms of input as in (\ref{prod}), where $\epsilon_{it} \sim N(0, \sigma_j^2)$ are i.i.d across values of $i$ and $t$.   $\omega_{it}$  follows an exogenous first-order stationary Markov process  given by 
$\omega_{it}=h^j(\omega_{it-1}) + \eta_{it} $ where,
conditional on $\mathcal{I}_{it-1}$, $\eta_{it}$ is  a mean-zero i.i.d. random variable. (c)   $(\gamma_t^j,\omega_{it}) \in\mathcal{I}_{it}$  and   $\epsilon_{it}\not\in\mathcal{I}_{it}$.  
\end{assumption}

\begin{assumption}\label{app:A2}
(a) Firms are price-takers in both output and input markets, where $P_{Y,t}$ and $P_{V,t}$ are, respectively, the prices of output and intermediate input in year $t$.
(b) $(P_{Y,t}, P_{V,t})$ are observed by firms at the beginning of the period before $V_{it}$ is chosen.
\end{assumption}

\begin{assumption}\label{app:A3}
 	Values of $V_{it}$ are chosen at time $t$ by maximizing the expected profit conditional on information $\mathcal{I}_{it}$ at time $t$ and conditional on the value of $(K_{it},L_{it})$. The profit maximization problem for firms with type $j$ technology is given by
	 \begin{equation}\label{eq:profit}
	V_{it} = \arg \max_{V} P_{Y,t}  
	\E[\exp(\epsilon_{it})|D_i=j] F^{j}_t(V,K_{it},L_{it},\omega_{it}) - P_{V,t} V.
	\end{equation} 
\end{assumption}

In Assumption \ref{app:A1}(a), each firm's production function belongs to one of the $M$ types. Assumption \ref{app:A1}(b) assumes that the idiosyncratic productivity shock follows a normal distribution. Assumption \ref{app:A1}(c) assumes that both the aggregate shock $\gamma_t^j$ and the serially correlated productivity shock $\omega_{it}$ are observed when intermediate inputs are chosen but idiosyncratic productivity shocks are unknown. 
Assumption \ref{app:A2} states that firms observe input and output prices when deciding on $V_{it}$.
Assumption \ref{app:A3} assumes that $V_{it}$ is chosen to maximize the current expected period profit conditional on the value of $(K_{it},L_{it})$.\footnote{We are agnostic about the timing of choosing $K_{it}$ and $L_{it}$ as long as they are either determined before $V_{it}$ or simultaneously chosen with $V_{it}$. It is reasonable to assume that capital input $K_{it}$ is determined before the value of $V_{it}$ is chosen. However, labor input $L_{it}$ may be flexibly chosen simultaneously with $V_{it}$ after $\gamma_t^j$ and $\omega_{it}$ are observed. Even when labor input is simultaneously chosen with intermediate input, equation (\ref{eq:profit}) and the corresponding first-order condition characterize the intermediate input choice once we interpret $L_{it}$ in (\ref{eq:profit}) as the optimal value chosen by firm $i$, as discussed in \cite{Ackerberg2015}.}

Given the above Assumptions \ref{app:A1}, \ref{app:A2}, and \ref{app:A3}, we derive an empirical specification based on the first-order condition of the profit maximization problem (\ref{eq:profit}), following the idea developed by \cite{Gandhi2020} and extending it to a finite mixture production function modeled by \cite{Kasahara2022esri}. Note that $E[\exp(\epsilon_{it})|D_i=j]=\exp( \sigma_j^2/2)$ for $\epsilon_{it}\sim N(0,\sigma_j^2)$. Then, because $\delta_{v,j}= \frac{\partial F_i^{j}(V_{it},K_{it},L_{it})/\partial V_{it}} {{F_i^{j}(V_{it},K_{it},L_{it})}/{V_{it}}}$ for the Cobb--Douglas production function, the first-order condition with respect to $V_{it}$ in (\ref{eq:profit}) together with the production function (\ref{prod}) implies that
\begin{equation}\label{eq:logs}
s_{it}
= \log \delta_{v,j} + \frac{1}{2}\sigma_{j}^2 -  \epsilon_{it}\quad \text{for $D_i=j$},
\end{equation}
where
\[
s_{it}:=\log\left(\frac{P_{V,t} V_{it}}{P_{Y,t} Y_{it}}\right)
\]
is the logarithm of the ratio of the intermediate input cost to revenue.

 
Collect the observed data as $\bs{W}_i = \{ s_{it},\log K_{it}\}_{t=1}^T$. 
Let $\mu_j = \log \delta_{v,j} + \frac{1}{2} \sigma_j^2$ and define a type-specific parameter to be $\bs{\theta}_j = (\mu_j,\sigma_j)$, where $\delta_{v,j}$ can be identified from $\bs{\theta}_j$ as $\delta_{v,j}=\exp(\mu_j-\sigma_j^2/2)$.
Collect the parameters of each type and the mixing probability as $\bs{\vartheta}_M = (\alpha_1,\ldots,\alpha_{M-1}, \bs{\theta}_1\t,\ldots,\bs{\theta}_M\t)\t$. Recall that $\epsilon_{it} \overset{iid}{\sim} N(0, \sigma_j^2)$ over $i$ and $t$ conditional on the technology type $D_i=j$. Then, from (\ref{eq:logs}), we can write the  density function of $s_{i1},..., s_{iT}$ as a mixture of type-specific likelihood density similar to the density function in equation (\ref{eq:fm}):
\begin{equation}\label{model-1}
f_{M}(\bs{W}_i;\bs{\vartheta}_{M}) = \sum_{j=1}^M \alpha_j \prod_{t=1}^T \frac{1}{\sigma_{j}} \phi\left(\frac{s_{it} - \mu_j}{\sigma_{j}} \right).
\end{equation}
The PMLE is defined as
\[
\hat{\bs{\vartheta}}_{M} = \arg\max_{\bs{\vartheta}_M} \sum_{i=1}^{n} \log f_{M}(\bs{W}_i;\bs{\vartheta}_{M}) +\tilde p_n(\bs{\vartheta}_{M}).
\]
 
As an alternative specification, we allow the elasticity of output for intermediate input to be a function of $\log K_{it}$  as 
$\log\delta_{v,j}=\beta_{0,j} + \beta_{k,j} \log K_{it}$.
This results in the logarithm of the ratio of intermediate input cost to revenue being linearly related to $\log K_{it}$   as
$
s_{it}
= \mu_j +\beta_{k,j} \log K_{it} -  \epsilon_{it} $ for $D_i=j$
with $\mu_j  = \beta_{0,j}+ \frac{1}{2}\sigma_{j}^2$. In this case, the conditional density function of $\{ s_{it}\}_{t=1}^T$ given $\{\log K_{it}\}_{t=1}^T$ is 
\begin{equation}\label{model-2}
f_{M}(\bs{W}_i;\bs{\vartheta}_{M}) = \sum_{j=1}^M \alpha_j \prod_{t=1}^T \frac{1}{\sigma_{j}} \phi\left(\frac{s_{it} - \mu_j - \beta_{k,j} \log K_{it}}{\sigma_{j}} \right).
\end{equation}
In addition, we consider a specification in which we include not only $\log K_{it}$ but also $\log L_{it}$ as a regressor:
\begin{equation}\label{model-3}
f_{M}(\bs{W}_i;\bs{\vartheta}_{M}) = \sum_{j=1}^M \alpha_j \prod_{t=1}^T \frac{1}{\sigma_{j}} \phi\left(\frac{s_{it} - \mu_j - \beta_{k,j} \log K_{it}- \beta_{\ell,j} \log L_{it}}{\sigma_{j}} \right).
\end{equation}


%

\subsection{Empirical results}
We apply the EM test to two producer-level datasets to determine the number of production technology types. 
We use the production data from Japanese publicly traded firms from 2003 to 2007 and Chilean manufacturing plants from 1992 to 1996.\footnote{Please refer to \cite{kasahara21} and \cite{KASAHARA2008}  for the details of the datasets of the Japanese publicly traded firms and the Chilean manufacturing plants, respectively. }
We clean the data and use the firms/plants with continuous data entry for five years to ensure that we have balanced panel data. We focus on the three largest industries
in terms of the number of firms and plants for each country (chemical, machine, and electronics for Japan and food products, fabricated metal products, and textiles for Chile). Table \ref{tab:sum}  presents the summary statistics for the revenue share of intermediate materials and the log of gross output in these industries. The within-industry standard deviations of the revenue share of intermediate materials are substantial across all industries, suggesting that the intermediate input elasticities differ across firms within the narrowly defined industries.


\begin{table}[H]
	\centering  
	\caption{Descriptive Statistics for the Revenue Share of Intermediate Materials and the Log of Gross Output for the Japanese Firms and Chilean Plants} \label{tab:sum}
	\begin{threeparttable} 
				\begin{tabularx}{\textwidth}{@{\extracolsep{\fill}} p{4.1cm}cccccc} 

			\toprule
		\multicolumn{7}{c}{\textbf{Panel A: \ Japanese publicly traded firms}}\\ 
		\midrule
		   &   &  &\multicolumn{2}{c}{$\frac{P_{V,t} V_{it}}{P_{Y,t} Y_{it}}$} & \multicolumn{2}{c}{$\log (Y_{it})$}\\  
		 Industry & NObs & n & $mean $ & $sd  $ & $mean  $ & $sd  $ \\  
			\midrule \noalign{\vskip 0.5mm}
			 Chemical & 805 & 161 & 0.34 & 0.15 & 17.52 & 1.24 \\ 
			 Machine & 790 & 158 & 0.50 & 0.16 & 17.31 & 1.35 \\ 
		 Electronics & 775 & 155 & 0.45 & 0.18 & 17.54 & 1.27 \\ 
		 \toprule
		\multicolumn{7}{c}{\textbf{Panel B: \ Chilean  plants}}\\
		\midrule
		  &   &  & \multicolumn{2}{c}{$\frac{P_{V,t} V_{it}}{P_{Y,t} Y_{it}}$} & \multicolumn{2}{c}{$\log (Y_{it})$}\\  
		Industry & NObs & n & $mean $ & $sd  $ & $mean  $ & $sd  $ \\ 
			\midrule \\[-1.8ex] 
		Food products & 4645 & 929 & 0.65 & 0.15 & 10.62 & 1.66 \\ 
		 Fabricated metal products & 1260 & 252 & 0.53 & 0.18 & 11.00 & 1.37 \\ 
		Textiles & 1130 & 226 & 0.58 & 0.19 & 11.01 & 1.32 \\ 
			\bottomrule \\[-1.8ex] 		
			\end{tabularx} 
		
		\begin{tablenotes}
		\item [1] The summary statistics are based on Japanese firm-level data from 2003 to 2007 and Chilean plant-level data from 1992 to 1996.
	 All observations with $\log (V_{it} / Y_{it}) \le -3$ and $\log (V_{it} / Y_{it}) > \log(2)$ are removed.  The dataset is a balanced panel; i.e., we keep firms/plants that are continuously observed for these five years.

	 \item[2] The variable $\frac{P_{V,t} V_{it} }{P_{Y,t} Y_{it} }$ is defined as the revenue share of the intermediate input, where $P_{V,t}$ is the average price of the intermediate input at time $t$, $P_{Y,t}$ is the average price of the output, $V_{it}$ is the quantity of the intermediate input, and $Y_{it}$ is the quantity of the output. 
	 
		\end{tablenotes}
	\end{threeparttable}	
\end{table}

\begin{table}[H] \centering 
	\caption{The EM Test for Japanese Producers Without Conditioning Variables}
	\small
	\label{table:Japan} 
	\begin{threeparttable}
		\begin{tabularx}{\textwidth}{@{}lc@{\extracolsep{\fill}}rrrrr@{}}
			\toprule
			& & M=1            & M=2            & M=3            & M=4           & M=5           \\
			\cmidrule{3-7}
			& &  \multicolumn{5}{c}{$T = 3$} \\
			\cmidrule{3-7}
			Chemical      &\textit{EM}      & $436.37^{***}$ & $239.83^{***}$ & $130.1^{***}$  & $126.4^{***}$ & $63.24^{***}$ \\
			&\textit{BIC}    & 805.55         & 383.43         & 157.5          & 41.62         & -70.46        \\ \noalign{\vskip 0.5mm}
			Electronics    &\textit{EM}& $563.94^{***}$ & $186.67^{***}$ & $115.82^{***}$ & $81.06^{***}$ & $47.76^{***}$ \\
			&\textit{BIC}  & 814.01         & 264.27         & 91.67          & -10.39        & -77.2         \\ \noalign{\vskip 0.5mm}
			Machine        &\textit{EM} & $434.91^{***}$ & $194.48^{***}$ & $72.83^{***}$  & $56.94^{***}$ & $54.77^{***}$ \\
			&\textit{BIC}  & 458.72         & 37.85          & -142.28        & -200.74       & -242.71          \\ \noalign{\vskip 0.5mm}

			\cmidrule{3-7}
			& & \multicolumn{5}{c}{$T = 4$} \\
			\cmidrule{3-7}
			Chemical      &\textit{EM}& $629.22^{***}$ & $308.6^{***}$  & $181.39^{***}$ & $177.38^{***}$ & $96.35^{***}$ \\
			&\textit{BIC}  & 1071.45        & 456.54         & 162.15         & -4.99          & -168.01       \\ \noalign{\vskip 0.5mm}
			Electronics &\textit{EM}  & $803.15^{***}$ & $282.32^{***}$ & $167.83^{***}$ & $106.43^{***}$ & $89.93^{***}$ \\
			&\textit{BIC}  & 1081.48 & 292.68 & 24.54 &	-484.46  \\ \noalign{\vskip 0.5mm}
			Machine    &\textit{EM}    & $620.95^{***}$ & $292.52^{***}$ & $118.37^{***}$ & $102.57^{***}$ & $75.32^{***}$ \\
			&\textit{BIC}  & 609.1          & 2.14           & -276.04        & -380.16        & -467.96  \\\noalign{\vskip 0.5mm}
			\cmidrule{3-7}
			& & \multicolumn{5}{c}{$T = 5$} \\
			\cmidrule{3-7}
			Chemical     &\textit{EM}  & $818.38^{***}$  & $386.08^{***}$ & $219.13^{***}$ & $209.42^{***}$ & $118.25^{***}$ \\
			&\textit{BIC}  & 1331.53         & 527.48         & 155.86         & -48.53         & -243.73        \\ \noalign{\vskip 0.5mm}
			Electronics   &\textit{EM}& $1024.86^{***}$ & $375.29^{***}$ & $226.01^{***}$ & $134.53^{***}$ & $126.36^{***}$ \\
			&\textit{BIC}  & 1343.12         & 332.61         & -28.32         & -239.31        & -359.17        \\ \noalign{\vskip 0.5mm}
			Machine       &\textit{EM}  & $819.98^{***}$  & $389.69^{***}$ & $156.44^{***}$ & $149.98^{***}$ & $96.32^{***}$  \\
			&\textit{BIC}  & 775.75          & -30.17         & -406.59        & -548.81        & -683.96        \\ \noalign{\vskip 0.5mm}
			\bottomrule
		\end{tabularx}
		\begin{tablenotes}
			\item[1] The estimation is based on the revenue share of intermediate materials.
			\item[2] $~^{*}$, $~^{**}$, and $~^{***}$ indicate significance at the $10 \%$, $ 5\%$, and $1\%$ levels, respectively.
		\end{tablenotes}
	\end{threeparttable}
\end{table}

\begin{table}[H] \centering 
	\caption{The EM Test for Chilean Producers Without Conditioning Variables}
	\small
	\label{table:Chile} 
	\begin{threeparttable}
		\begin{tabularx}{\textwidth}{@{}lc@{\extracolsep{\fill}}rrrrr@{}}
			\toprule
			& & M=1            & M=2            & M=3            & M=4           & M=5           \\
			\cmidrule{3-7}
			& & \multicolumn{5}{c}{$T = 3$} \\
			\cmidrule{3-7}
            Food products     &\textit{EM}    & $805.51^{***}$ & $637.77^{***}$ & $204.92^{***}$ & $80.54^{***}$ & $72.41^{***}$ \\
			&\textit{BIC}  &  422.55& 	-371.13& 	-991.96& 	-1176.61& 	-1236.82 \\
			\noalign{\vskip 0.5mm}    
			Fabricated metal products  &\textit{EM}   & $238.84^{***}$ & $68.91^{***}$  & $26.24^{***}$  & $24.42^{***}$ & $21.82^{***}$ \\
			&\textit{BIC}  &719.74&	496.49&	444.02&	433.01&	425
			\\  \noalign{\vskip 0.5mm}
			Textiles   &\textit{EM} & $229.87^{***}$ & $146.17^{***}$ & $64.76^{***}$  & $27.06^{***}$ & $29.98^{**}$  \\
			&\textit{BIC}  &635.37&	418.28&	288.34&	236.9&	223.34\\ \noalign{\vskip 0.5mm}
			\cmidrule{3-7}
			& & \multicolumn{5}{c}{$T = 4$} \\
			\cmidrule{3-7}
			Food products     &\textit{EM}   & $1165.08^{***}$ & $874.27^{***}$ & $257.49^{***}$ & $130.61^{***}$ & $139.59^{***}$ \\
			&\textit{BIC}  &419.47&	-730.83&	-1586.11&	-1825.87&	-1938.03
			\\ \noalign{\vskip 0.5mm}
			Fabricated metal products    &\textit{EM} & $362.1^{***}$   & $120.7^{***}$  & $41.6^{***}$   & $43.68^{***}$  & $20.95^{***}$  \\
			&\textit{BIC}  &905.9&	559.3&	453.41&	427.34&	399.82
			\\ \noalign{\vskip 0.5mm}
			Textiles   &\textit{EM}  & $325.17^{***}$  & $222.28^{***}$ & $74.19^{***}$  & $47.58^{***}$  & $51.65^{***}$  \\
			&\textit{BIC}  & 821.73&	510.98&	303.8&	243.51&	210.77
			\\\noalign{\vskip 0.5mm}
			\cmidrule{3-7}
			& & \multicolumn{5}{c}{$T = 5$} \\
			\cmidrule{3-7}
			Food products     &\textit{EM}     & $1553.9^{***}$ & $1010.31^{***}$ & $290.02^{***}$ & $172.46^{***}$ & $155.25^{***}$ \\
			&\textit{BIC}  &471.66&	-1066.71&	-2057.71&	-2329.38&	-2484.82
			\\ \noalign{\vskip 0.5mm}
			Fabricated metal products     &\textit{EM}   & $478.94^{***}$ & $176.5^{***}$   & $58.96^{***}$  & $59.37^{***}$  & $33.19^{***}$  \\
			&\textit{BIC}  &1101.11&	637.21&	477.1&	433.62&	389.54 \\ \noalign{\vskip 0.5mm}
			Textiles   &\textit{EM} & $428.29^{***}$ & $280.46^{***}$  & $103.41^{***}$ & $56.63^{***}$  & $53.57^{***}$  \\
			&\textit{BIC}  &968.16&	556.01&	289.55&	201.41&	160\\ \noalign{\vskip 0.5mm}
			\bottomrule
		\end{tabularx}
		\begin{tablenotes}
			\item[1] The estimation is based on the revenue share of intermediate materials.
			\item[2] $~^{*}$, $~^{**}$, and $~^{***}$ indicate significance at the $10 \%$, $ 5\%$, and $1\%$  levels, respectively.			
		\end{tablenotes}
	\end{threeparttable}
\end{table}

To determine the number of components, we test the null hypothesis $H_0: M=M_0$ against $H_1: M=M_0+1$ by applying the EM test at the 5\% significance level sequentially for $M_0 = 1,\ldots,5$.   If we fail to reject the null hypothesis at a certain $M_0 = M$, then we conclude that there are $M$ types of intermediate input elasticities. We consider both the models without conditioning variables (\ref{model-1}) and the models with conditioning variables (\ref{model-2})--(\ref{model-3}).

Tables \ref{table:Japan}  and \ref{table:Chile}  report the results of the EM test for the model without conditioning variables  (\ref{model-1}) for the Japanese and the Chilean industries with a panel length of $T=3,4,5$ and a null model of $M=1,...,5$. For all industries in both countries and all panel lengths, we reject the null hypothesis of $H_0: M=M_0$ for all $M_0=1,2,3,4,$ and $5$ at the 5\% significance level, which indicates that there are at least five types of intermediate input elasticities. This result reflects the considerable and persistent heterogeneity in the revenue share of intermediate materials across firms or plants, providing strong evidence for substantial heterogeneity in intermediate input elasticities across firms' production functions among Japanese and Chilean producers. Our findings serve as a caution against the conventional empirical practice of estimating the Cobb--Douglas production function, which assumes that elasticity parameters are common across firms. Given the strong evidence of heterogeneity in the production function coefficients, incorporating heterogeneity in production function coefficients in empirical applications is warranted and should be encouraged.

One possible reason for the estimated number of technology types being greater than 5 is that the assumption of the Cobb--Douglas production function may be too restrictive. When the production function is not Cobb--Douglas, the revenue share of intermediate materials generally depends on the value of production inputs \citep{Gandhi2020}. For this reason, we test the number of technology types when the revenue share of intermediate materials depends on the values of capital input and labor input by estimating models (\ref{model-2})--(\ref{model-3}).

Table \ref{table:regressor_1} presents the results of the SHT and the BIC when we estimate the mixture regression model with $\log K_{it}$ in (\ref{model-2}) using data with a panel length of $T=3$. For the Japanese chemical, electronics, and machinery industries, the SHT suggests that the data are generated from seven- to nine-component models; concurrently, the BIC selects models with at least 10 components. For the Chilean food industry, the SHT indicates a 10-component model, while the BIC chooses an eight-component model. In contrast, the SHT and the BIC respectively select models with seven and six components for the Chilean fabricated metal products industry and the Chilean textile industry.

Table \ref{table:regressor_2} reports the results for the model that includes both $\log K_{it}$ and $\log L_{it}$ as regressors. Across six industries, the SHT and the BIC in Table \ref{table:regressor_2} both select models with at least five components, providing evidence for substantial heterogeneity in production technology across firms and plants. Comparing the results of Table \ref{table:regressor_2} with those of Table \ref{table:regressor_1}, the selected number of components for the model with $\log K_{it}$ and $\log L_{it}$ is smaller than that for the model with only $\log K$. This suggests that the number of components may be overestimated if we do not consider a sufficiently flexible production function specification by excluding some regressors.

\begin{landscape}
	\begin{table}[H]
		\caption{The EM Test and the BIC  (Dependent Variable:  $\log \frac{P_{V,t} V_{it}}{P_{Y,t} Y_{it}}$,  Regressor: $\log K_{it}$)}
		\label{table:regressor_1}
	\small
	\begin{tabular}{l rrrrrrrrrr}
	\toprule
	 $M_0$   & 1    & 2    & 3   & 4 & 5  & 6 & 7& 8& 9 & 10 \\
	\midrule
	\multicolumn{11}{l}{ Japanese Chemical } \\
	\cmidrule{1-11} 
	\textit{EM} & $459.4^{***}$  & $236.36^{***}$ & $125.42^{***}$ & $118.36^{***}$ & $87.63^{***}$ & $53.72^{***}$ & $38.69^{***}$& $34.07^{**}$ & \cellcolor{yellow} $36.47$ & -\\
	\textit{BIC} & 1384.76   & 943.61    & 726.53    & 620.32   & 518.86   & 449.92   & 413.49  & 394.46   & 381.48 & \cellcolor{yellow}366.09  \\
 	\bottomrule
	
	\multicolumn{11}{l}{Japanese Electronics} \\
	\cmidrule{1-11}  
	\textit{EM} & $560.06^{***}$ & $213.82^{***}$ & $116.29^{***}$ & $78.81^{***}$ & $47.05^{***}$  & $40.77^{***}$ & $27.4^{**}$  & \cellcolor{yellow}$29.02$ & - & - \\
	\textit{BIC} & 1332.14   & 788.19    & 593.44    & 495.74   & 434.15    & 406.77   & 385.45   & 372.63& 367.31  & \cellcolor{yellow}351.17 \\
 	\bottomrule

	\multicolumn{11}{l}{ Japanese Machine } \\
	\cmidrule{1-11} 
	\textit{EM} & $433.19^{***}$ & $202.92^{***}$ & $80.42^{***}$  & $76.82^{***}$  & $53.83^{***}$ & $34.62^{**}$ & \cellcolor{yellow}$55.65$ & - & - & -\\
	\textit{BIC} & 1355.6    & 940.49    & 757  & 696.06    & 638.48   & 617.4    & 588.94   & 568.71  & 555.15  & \cellcolor{yellow}544.51 \\
 	\bottomrule

	\multicolumn{11}{l}{Chilean Food Products } \\
	\cmidrule{1-11}  
	\textit{EM} & $816.06^{***}$ & $489.37^{***}$    & $169.14^{***}$    & $80.88^{***}$ & $80.63^{***}$ & $52.67^{***}$   & $31.29^{***}$  & $17.16^{**}$ & $20.55^{***}$ & \cellcolor{yellow} $-60.46^{}$	\\
		\textit{BIC} & 6759.39& 5962.74    & 5499.3& 5356.47    & 5301.31    & 5241.91& 5210.27 & \cellcolor{yellow} 5200.71 & 5210.77 & 5222.29 \\ 
 	\bottomrule
	
	\multicolumn{11}{l}{  Chilean Fabricated Metal Products} \\
	\cmidrule{1-11} 
		\textit{EM} & $199.35^{***}$ & $63.25^{***}$  & $49.24^{***}$ & $30.27^{***}$ &  $15.73^{**}$ & $18.25^{**}$ & \cellcolor{yellow} $10.88^{}$	& - & - & - \\
		\textit{BIC}  & 1923.64    & 1744.72    & 1699.97   & 1670.93   & 1661.03&  \cellcolor{yellow} 1659.54	& 1665.08	& 1669.02	& 1680.96&	1695.54\\
 	\bottomrule
	 	
	\multicolumn{11}{l}{ Chilean Textile } \\
	\cmidrule{1-11} 
	\textit{EM} & $201.86^{***}$ & $95.17^{***}$  & $61.43^{***}$ & $31.17^{***}$   & $14.12^{*}$  &  $17.45^{**}$ & \cellcolor{yellow} $7.94^{}$ & - & - & - \\
	\textit{BIC} & 1681.91    & 1499.99    & 1424.93   & 1380.93& 1368.65& \cellcolor{yellow}1364.94 & 1365.72    & 1370.72 &	1382.83 &	1392.24 \\
 	\bottomrule
	\end{tabular}
		\begin{tablenotes}
			\item[1] The estimation is based on the revenue share of intermediate materials using panel data of length $T=3$.	
			\item[2] $~^{*}$, $~^{**}$, and $~^{***}$ indicate significance at the $10 \%$, $ 5\%$, and $1\%$  levels, respectively.	\end{tablenotes}
	\end{table}

	\begin{table}[H]
		\caption{The EM Test and the BIC  (Dependent Variable:  $\log \frac{P_{V,t} V_{it}}{P_{Y,t} Y_{it}}$,  Regressors: $\log K_{it}$ and $\log L_{it}$)}
		\label{table:regressor_2}
		\small
	\begin{tabular}{l rrrrrrrrrr}
	\toprule
	 $M_0$   & 1    & 2    & 3   & 4 & 5  & 6 & 7& 8& 9 & 10 \\
	\midrule
	\multicolumn{11}{l}{ Japanese Chemical } \\
	\cmidrule{1-11}   
	\textit{EM} & $412.35^{***}$ & $224.09^{***}$ & $141.59^{***}$ & $132.24^{**}$ & \cellcolor{yellow}$121.56$ & - & - & - & - & - \\
	\textit{BIC} & 1294.05   & 905.44    & 705.72    & 587.3& 490.07   & 479.74   & \cellcolor{yellow}389.05 & 390.29   & 382.28 & 372.69  \\ 
 	\bottomrule
	
	\multicolumn{11}{l}{Japanese Electronics} \\
	\cmidrule{1-11}   
	\textit{EM} & $573.11^{***}$ & $218.38^{***}$ & $116.07^{***}$ &  $94.76^{**}$ &  \cellcolor{yellow}$47.73$ & - & - & -  & - & -  \\
	\textit{BIC} & 1336.69   & 784.95    & 590.73    & 498.55   & 426.23    & 389.91   & 372.05   & 371.64& \cellcolor{yellow}359.15 & 368.25  \\
 	\bottomrule

	\multicolumn{11}{l}{ Japanese Machine } \\
	\cmidrule{1-11}  
	\textit{EM} & $468.06^{***}$ & $204.01^{***}$ & $93.35^{***}$  & $81.62^{***}$ & $62.00^{***}$    & $37.04^{***}$ & \cellcolor{yellow}$14.21^{}$  & - & - & - \\
	\textit{BIC} & 1360.56   & 915.69    & 736.2& 676.26    & 625.66   & 596.45   & 564.34   & 548.7   & \cellcolor{yellow}536.78 & 539.64  \\
 	\bottomrule

	\multicolumn{11}{l}{Chilean Food Products } \\
	\cmidrule{1-11}   
	\textit{EM}   & $805.09^{***}$  & $478.64^{***}$ & $177.08^{***}$ & $84.13^{***}$  & $80.96^{***}$  & $51.97^{***}$    & $32.3^{**}$ &  \cellcolor{yellow}$19.50$ & -  & -   \\
	\textit{BIC}  & 6732.11& 5952.7& 5506.55    & 5362.27    & 5309.37    & 5257.78& 5233.37 & \cellcolor{yellow} 5229.9 & 5242.9  & 5258.41  \\
 	\bottomrule
	
	\multicolumn{11}{l}{  Chilean Fabricated Metal Products} \\
	\cmidrule{1-11}  
	\textit{EM} & $204.45^{***}$ & $63.57^{***}$  & $49.42^{***}$ & $28.61^{***}$ & \cellcolor{yellow}$18.32^{}$   & - & - & - & - & -  \\
	\textit{BIC}  & 1926.06    & 1747.29    & 1709.44   & 1685.39   & \cellcolor{yellow}1678.71  & 1680.54    & 1685.02 & 1696.19 & 1703.21    & 1723.56  \\
 	\bottomrule
	 	
	\multicolumn{11}{l}{ Chilean Textile } \\
	\cmidrule{1-11}  
	\textit{EM} & $203.69^{***}$ & $90.69^{***}$  & $58.4^{***}$  & $32.55^{***}$ & \cellcolor{yellow}$16.19^{}$   & - & - & - & - & - \\
	\textit{BIC} & 1673.99    & 1495.55    & 1431.18   & 1394.54& 1382.59& 1373.03 & \cellcolor{yellow}1368.8 & 1382.42    & 1394.3  & 1394.09  \\
 	\bottomrule
	\end{tabular}
		\begin{tablenotes}
			\item[1] The estimation is based on the revenue share of intermediate materials using panel data of length $T=3$.	
			\item[2] $~^{*}$, $~^{**}$, and $~^{***}$ indicate significance at the $10 \%$, $ 5\%$, and $1\%$  levels, respectively.	\end{tablenotes}
	\end{table}
\end{landscape}

%% file: sections/FM_8_conclusion.tex

The selection of the number of components in a finite normal mixture panel regression model is a crucial practical issue that must be addressed with care. Arbitrary choice of the number of components can result in biased estimates and invalid inferences, and can reduce the credibility of the final outcomes. To tackle this issue, this study proposes the PLRT and an EM test and derives their asymptotic distribution for the null hypothesis of a model with $M_0$ components against the alternative hypothesis with $(M_0 + 1)$ components. We also develop a procedure to consistently select the number of components by sequentially applying the PLRT and EM tests. Through a simulation exercise, we demonstrate that the proposed SHT procedure exhibits good performance in finite samples.

As an empirical application, we estimate the number of production technology types using producer-level panel data from Japan and Chile. We find that most industries in our dataset exhibit a level of heterogeneity that requires a five-or-more-component mixture model when using the Cobb--Douglas production specification or a specification in which the elasticity of inputs depends on capital and labor input linearly. This provides strong evidence of the presence of unobserved heterogeneity in technology types. One important caveat of our empirical exercise is that the class of production functions that we investigate may be restrictive. Investigating production function heterogeneity with more flexible function forms is an important future research topic.

%% file: sections/FM_appendixa.tex
\begin{proof}[Proof of Proposition \ref{prop:unbounded_likelihood}]

	We first consider a model with an intercept parameter and a variance parameter but without covariates with  $\mathbf W_i = \{ y_{it} \}_{t=1}^T$. 

		Define
	\[
	s_i^2 = \frac{1}{T-1}\sum_{t=1}^T (Y_{it}-\bar Y_i)^2 \quad\text{with}\quad \bar Y_i = \frac{1}{T} \sum_{t=1}^T Y_{it},
	\]
	where $s_i^2$ follows a chi-square distribution with $T-1$ degrees of freedom.
Let $i^* = \arg\min_{i=1,\ldots,n} \{s_i^2\}$ so that  $s_{i^*}^2  = \min\{s_1^2,\ldots,s_n^2\}$ is the minimum of $s_i^2$ across all values of $i$.
We consider a sequence of parameters $\bs\vartheta_{2,n}=(\alpha_n,\bs\theta_{1,n}\t,\bs{\theta}_{2,n}\t)\t$ with $\alpha_n = 1/n$,  $\bs\theta_{1,n}=(\mu_{1,n},\sigma_{1,n}^2)\t=(\bar Y_{i^*},s_{i^*}^2)\t$, and  $\bs\theta_{2,n}=\bs\theta^*=(\mu^*,\sigma^*)\t$ for all $n$. Because $LR_n^*(\bs\vartheta_{2,n})\leq LR_n^*(\tilde{\bs\vartheta}_{2,n})$, it suffices to show that $LR_n^*(\bs\vartheta_{2,n})$ is unbounded in probability.

	Define \[
	\ell(\bs{W}_i;\bs\theta) : = \log f(\bs{W}_i;\bs\theta) = -\frac{T}{2}\log \sigma^2 - \frac{T}{2}\log(2\pi) - \frac{1}{2}\sum_{t=1}^T \left(\frac{Y_{it}-\mu}{\sigma}\right)^2.
	\]
	Then, the LRT statistic for a two-component mixture is written as
		\begin{align}
		LR_n^*(\bs\vartheta_{2,n}) & = 2 \left\{\sum_{i=1}^n \log\left(\alpha_n \prod_{t=1}^T \frac{1}{\sigma_{1,n}}\phi\left(   \frac{Y_{it}-\mu_{1,n}}{\sigma_{1,n}}\right)
	+(1-\alpha_n) \prod_{t=1}^T \frac{1}{\sigma^*}\phi\left(  \frac{Y_{it}-\mu^*}{\sigma^*}\right)\right)-  \sum_{i=1}^n\ell(\bs{W}_i;\bs\theta^*) \right\}\nonumber\\
	& =
	2\sum_{i\neq i^*} \left\{\log\left(   \exp(\log\alpha_n+\ell(\bs{W}_i;{\bs\theta}_{1,n})) + \exp(\log (1-\alpha_n)+\ell(\bs{W}_i;{\bs\theta}^*))\right)-\ell(\bs{W}_i;\bs\theta^*)\right\} \nonumber \\
	&\quad +2 \left\{\log\left(   \exp(\log\alpha_n+\ell(\bs{W}_{i^*};{\bs\theta}_{1,n})) + \exp(\log (1-\alpha_n)+\ell(\bs{W}_{i^*};{\bs\theta}^*))\right)-\ell(\bs{W}_{i^*};\bs\theta^*)\right\} \label{lr-1}.
	\end{align}

	The first term on the right-hand side of (\ref{lr-1}) can be rewritten as
	\begin{align*}
	 = 2(n-1) \log \left(\frac{n-1}{n}\right)
	+ 2	\sum_{i\neq i^*}\log\left(1+\frac{1}{n-1} \exp(\ell(\bs{W}_i;{\bs\theta}_{1,n})-\ell(\bs{W}_i;{\bs\theta}^*))\right),
	\end{align*}
	which is bounded from below by $-1$ as $n\rightarrow \infty$ because  $\lim_{n\rightarrow \infty} 2(n-1) \log \left(\frac{n-1}{n}\right)  =-1$ and $\log\left(1+\frac{1}{n-1} \exp(\ell(\bs{W}_i;{\bs\theta}_{1,n})-\ell(\bs{W}_i;{\bs\theta}^*))\right)\geq 0$ for all $n$.

The second term  on the right-hand side of (\ref{lr-1}) is written as
\begin{align}
2\{ - \log n + \ell(\bs{W}_{i^*};{\bs\theta}_{1,n})\} + 2 \log\left(1 + (n-1)  \exp(\ell(\bs{W}_{i^*};{\bs\theta}^*)-\ell(\bs{W}_{i^*};{\bs\theta}_{1,n}))\right) - 2 \ell(\bs{W}_{i^*};\bs\theta^*),\label{lr-2}
\end{align}
where  $2\{ - \log n + \ell(\bs{W}_{i^*};{\bs\theta}_{1,n})\}$ diverges to infinity as $n\rightarrow
\infty$ by   Lemma \ref{lemma:unbounded_likelihood}, the second term in (\ref{lr-2}) is bounded below from zero, and the third term is bounded in probability because $\ell(\bs{W}_{i^*};\bs\theta^*)=O_p(1)$. Therefore,    for any $M<\infty$, we have
	$\Pr \Big(LR_n^*(\bs\vartheta_{2,n}) \le M \Big) \to 0$ as $n \to \infty$. The stated result follows from $LR_n^*(\bs\vartheta_{2,n})\leq LR_n^*(\tilde{\bs\vartheta}_{2,n})$ for all $n$.

For a model with covariates, we can consider a sequence of parameters $\bs\vartheta_{2,n}=(\alpha_n,\bs\theta_{1,n}\t,\bs{\theta}_{2,n}\t,\bs\gamma_n\t)\t$ with $\alpha_n = 1/n$,  $\bs\theta_{1,n}=(\mu_{1,n},\sigma_{1,n}^2,\bs\beta_{1,n}\t)\t=(\bar Y_{i^*}-\bar {\bs Z}_{i^*}\t\bs\gamma^*, s_{i^*}^2,\bs 0\t)\t$ with $\bar {\bs Z}_{i^*}=(1/T)\sum_{t=1}^T \bs Z_{it}$,    $\bs\theta_{2,n}=\bs\theta^*=(\mu^*,\sigma^*,(\bs\beta^*)\t)\t$, and $\bs\gamma_n=\bs\gamma^*$. Then, repeating the above argument, the stated result follows. 
\end{proof}

\begin{proof}[Proof of Proposition \ref{prop:vartheta_convergence}]
The stated result follows from repeating the proof of Proposition \ref{prop:vartheta_convergence_M}.
\end{proof}

\begin{proof}[Proof of Proposition \ref{prop:expansion}]
The proof follows that of Proposition 2 in \cite{Kasahara2012}.
For a vector $\bs{x}$ and a function $f(\bs{x})$, let $\nabla_{\bs{x}^k}f(\bs{x})$ denote its $k$-th derivative with respect to $\bs{x}$, which can be a multidimensional array. Observe that for any finite $k$ and for a neighborhood $\mathcal{N}$ of $\bs{\psi}^*$, we obtain
\begin{equation}
\begin{aligned}
&E|| \nabla_{\bs{\psi}^k} g(\bs{W}_i;\bs{\psi}^*,\alpha)/ g(\bs{W}_i;\bs{\psi}^*,\alpha)||^2<\infty, \\
&E||\sup_{\bs{\psi}\in\Theta_{\bs{\psi}} \cap \mathcal{N}}\nabla_{\bs{\psi}^k} \log g(\bs{W}_i;\bs{\psi},\alpha)||^2<\infty
\end{aligned} \label{nabla_0}
\end{equation}
because each element of $\nabla_{\bs{\psi}^k} \log g(y|\bs{x},\bs{z};\bs{\psi},\alpha)$ is written as a sum of products of Hermite polynomials. 
 Note also that  the following holds:
\begin{align}
& \nabla_{\eta \lambda_j }L_n(\bs\psi^*,\alpha)=0, \quad \nabla_{\lambda_i \lambda_j\lambda_k } L_n(\bs\psi^*,\alpha)=O_p(n^{1/2}), \label{nabla_1} \\
& \nabla_{\eta\eta \lambda_i }L_n(\bs\psi^*,\alpha) =O_p(n), \quad \nabla_{\eta\eta\eta}L_n(\bs\psi^*,\alpha)=O_p(n), \label{nabla_2}
\end{align} 
where equation (\ref{nabla_1}) follows from Proposition \ref{lemma:KS2018_lemma7}{(a) and (c)} and (\ref{nabla_0}) and equation (\ref{nabla_2}) is a simple consequence of (\ref{nabla_0}).
Furthermore, for a neighborhood $\mathcal{N}$ of $\bs\psi^*$,
\begin{align}
& \sup_{\bs\psi\in\Theta_{\bs\psi} \cap \mathcal{N}} \left| n^{-1}\nabla^{(4)} L_n(\bs\psi ,\alpha)- E\nabla^{(4)}  \log g(\bs W_i;\bs\psi ,\alpha) \right| = o_p(1), \label{nabla_3}\\
& E\nabla^{(4)} g(\bs W_i;\bs\psi,\alpha) \text{ is continuous in }\psi  {\in \Theta_{\bs\psi} \cap \mathcal{N}}. \label{nabla_4}
\end{align} 
  Equations (\ref{nabla_3}) and (\ref{nabla_4})  
follow  from Lemma 2.4 of \cite{Newey1994} and the fact that $\nabla_{\bs{\psi}^k} \log g(\bs{w};\bs{\psi},\alpha)$ is written as a sum of products of Hermite polynomials. 

Taking a fourth-order Taylor expansion of $L_n(\bs{\psi},\alpha)$ around $\bs{\psi}^*$ and using (\ref{nabla_0}) and (\ref{nabla_1}), we can write $L_n(\bs{\psi},{\alpha}) -L_n(\bs{\psi}^*,\alpha)$ as the sum of the relevant terms and the remainder term as follows:
\begin{align}
\lefteqn{ L_n(\bs{\psi},{\alpha}) -L_n(\bs{\psi}^*,\alpha) = } \nonumber\\
& \quad \nabla_{\bs{\eta}}L_n^* (\bs{\eta} - \bs{\eta}^*) + \frac{1}{2!} (\bs{\eta} - \bs{\eta}^*)^{\top}\nabla_{\bs{\eta} \bs{\eta}^{\top}}L_n^*(\bs{\eta} - \bs{\eta}^*) +\frac{1}{2!}  \sum_{i=1}^{q+2}\sum_{j=1}^{q+2} \nabla_{\lambda_i\lambda_j }L_n^*\lambda_i\lambda_j \label{LR0-score0} \\ 
&\quad + \frac{3}{3!} \sum_{i=1}^{q+2} \sum_{j=1}^{q+2} (\bs{\eta} - \bs{\eta}^*)\t\nabla_{\bs\eta \lambda_i\lambda_j} L_n^* \lambda_i\lambda_j    \label{LR0-hessian0} \\
 &\quad + \frac{1}{4!}  \sum_{i=1}^{q+2}\sum_{j=1}^{q+2}\sum_{k=1}^{q+2} \sum_{\ell=1}^{q+2} \nabla_{\lambda_i\lambda_j\lambda_k\lambda_\ell}L_n^*\lambda_i\lambda_j\lambda_k\lambda_\ell + R_{n}(\bs{\psi},\alpha), \label{LR0-hessian1}   \end{align}
where $\nabla L_n^*$ denotes the derivative of $L_n(\bs{\psi},\alpha)$ evaluated at $(\bs{\psi}^*,\alpha)$. In view of (\ref{nabla_1}) and (\ref{nabla_2}), the remainder term is written as
\begin{align}
\lefteqn{ R_n(\bs\psi,\alpha) =  O_p(n^{1/2}) \sum_{i=1}^{q+2}\sum_{j=1}^{q+2}\sum_{k=1}^{q+2} \lambda_i\lambda_j \lambda_k + O_p(n) \left(  \sum_{i=1}^{q+2} ||\bs{\eta}-\bs{\eta}^*||^2 \lambda_i +  ||\bs{\eta}-\bs{\eta}^*||^3 \right)} \label{Rn_2} \\
& +O_p(n)\sum_{i=1}^{q+2}\sum_{j=1}^{q+2}\sum_{k=1}^{q+2}\left( ||\bs{\eta}-\bs{\eta}^*||^4 + ||\bs{\eta}-\bs{\eta}^*||^3 |\lambda_i|+ ||\bs{\eta}-\bs{\eta}^*||^2 |\lambda_i \lambda_j| + ||\bs{\eta}-\bs{\eta}^*|| |\lambda_i \lambda_j \lambda_k| \right) \qquad \label{Rn_3} \\
& + \frac{1}{4!}  \sum_{i=1}^{q+2}\sum_{j=1}^{q+2}\sum_{k=1}^{q+2}\sum_{\ell=1}^{q+2}  \{\nabla_{\lambda_i \lambda_j \lambda_k \lambda_\ell} L_n(\bs{\psi}^\dag,\alpha) - \nabla_{\lambda_i \lambda_j \lambda_k \lambda_\ell} L_n(\bs{\psi}^*,\alpha) \} \lambda_i \lambda_j \lambda_k \lambda_\ell
\label{Rn_4}
\end{align}
with $\bs{\psi}^\dag$ being between $\bs{\psi}$ and $\bs{\psi}^*$. Because $||\sqrt{n}t(\bs{\psi},\alpha)||^2 = n||\bs{\eta} - \bs{\eta}^*||^2 + n \sum_{i=1}^{q+2}\sum_{j=1}^i \alpha^2(1-\alpha)^2 |\lambda_i\lambda_j|^2$, the right-hand side of (\ref{Rn_2}) and the terms in (\ref{Rn_3}) are bounded by $O_p(1)(||\sqrt{n}t(\bs{\psi},\alpha)||+||\sqrt{n}t(\bs{\psi},\alpha)||^2)(||\bs{\eta}-\bs\eta^*||+||\lambda||)$. In view of (\ref{nabla_3}) and (\ref{nabla_4}), (\ref{Rn_4}) is bounded by  $||\sqrt{n}t(\bs{\psi},\alpha)||^2[d(\bs{\psi}^\dagger)+o_p(1)]$ with $d(\bs{\psi}^\dagger)\rightarrow 0$ as $\bs{\psi}^\dagger \rightarrow \bs{\psi}^*$, where a function $d(\bs{\psi}^\dagger)$ corresponds to $n^{-1}\E[\nabla_{\lambda_i \lambda_j \lambda_k \lambda_\ell}L_n(\bs{\psi}^\dagger,\alpha) - \nabla_{\lambda_i \lambda_j \lambda_k \lambda_\ell}L_n(\bs{\psi}^*,\alpha)]$. Therefore, $R_n(\bs{\psi},\alpha)=(1+||\sqrt{n}t(\bs{\psi},\alpha)||)^2[d(\bs{\psi}^\dagger)+o_p(1)  + O_p(||\bs{\psi}-\bs{\psi}^*||)]$, and part (a) follows.

Part (b) follows from  Lemma \ref{lemma:KS2018_lemma7}(c) and (d), the Lindeberg--Levy central limit theorem, and the finiteness of $\bs{\mathcal{I}}$ in part (c).

For part (c), we first provide the formula of $\bs{\mathcal{I}}_n$. Partition $\bs{\mathcal{I}}_n$ as
\[
\bs{\mathcal{I}}_n = \left( \begin{array}{cc}
\bs{\mathcal{I}}_{\bs\eta n} & \bs{\mathcal{I}}_{\bs\eta\bs\lambda n}\\
\bs{\mathcal{I}}_{\bs\eta\bs\lambda n}\t & \bs{\mathcal{I}}_{\bs\lambda n}
\end{array}\right), \quad \bs{\mathcal{I}}_{\bs\eta n}: (p+q+2)\times (p+q+2), \quad \bs{\mathcal{I}}_{\bs\eta\bs\lambda n}: (p+q+2)\times q_\lambda, \quad \bs{\mathcal{I}}_{\bs\lambda n}: q_\lambda \times q_\lambda,
\]
where $q_\lambda$ represents the number of unique terms in $ \sum_{i=1}^{q+2}\sum_{j=1}^{q+2}\sum_{k=1}^{q+2} \sum_{\ell=1}^{q+2} \lambda_i\lambda_j\lambda_k\lambda_\ell$.
$\bs{\mathcal{I}}_{\bs\eta n}$ is given by $\bs{\mathcal{I}}_{\bs\eta n} = - n^{-1}\nabla_{\bs\eta\bs\eta\t}L_n(\psi^*,\alpha)$. For $\bs{\mathcal{I}}_{\bs\eta\bs\lambda n}$, let $A_{ij} = n^{-1}\nabla_{\bs\eta \lambda_i \lambda_j} L_n(\psi^*,\alpha)$  and write the  term in (\ref{LR0-hessian0})   as   $(n/2) \sum_{i=1}^{q+2} \sum_{j=1}^{q+2} (\bs\eta -\bs\eta^*)\t A_{ij} \lambda_i \lambda_j= n \sum_{i=1}^{q+2} \sum_{j=1}^i c_{ij}  (\bs\eta -\bs\eta^*)\t A_{ij} \lambda_i \lambda_j$, where the values of $c_{ij}$ are defined when we introduce $\widetilde{\nabla}_{\bs\theta\bs\theta\t} f^*$ after (\ref{eq:s_1}). Then, by defining  $\bs{\mathcal{I}}_{\bs\eta\bs\lambda n} = -(c_{11}A_{11}, \ldots,  c_{qq}A_{qq}, c_{12} A_{12}, \ldots, c_{q-1,q} A_{q-1,q}) /\alpha(1-\alpha)$, the  term in (\ref{LR0-hessian0}) equals $- n(\bs\eta -\bs\eta^*)\t \bs{\mathcal{I}}_{\bs\eta\bs\lambda n}[\alpha(1-\alpha)v(\bs\lambda)]$. For $\bs{\mathcal{I}}_{\bs\lambda n}$, define $B_{ijk\ell} = n^{-1}(8/4!)\nabla_{\lambda_i \lambda_j\lambda_k \lambda_\ell} L_n(\bs\psi^*,\alpha)$ so that the first term in (\ref{LR0-hessian1}) is written as $(n/8)\sum_{i=1}^{q+2} \sum_{j=1}^{q+2} \sum_{k=1}^{q+2} \sum_{\ell=1}^{q+2} B_{ijk\ell} \lambda_i \lambda_j \lambda_k \lambda_\ell= (n/2) \sum_{i=1}^{q+2} \sum_{j=1}^i \sum_{k=1}^{q+2} \sum_{\ell=1}^k c_{ij}c_{k\ell}B_{ijk\ell}\lambda_i \lambda_j \lambda_k \lambda_\ell$. Define $\bs{\mathcal{I}}_{\bs\lambda n}$ such that the $(ij,k\ell)$ element of $\bs{\mathcal{I}}_{\bs\lambda n}$ is $-c_{ij}c_{k\ell} B_{ijk\ell}/\alpha^2(1-\alpha)^2$, where the values of $ij$ run over $\{(1,1),\ldots,(q,q),(1,2),\ldots,(q-1,q)\}$. Then,  the first term in (\ref{LR0-hessian1}) equals $-(n/2) [\alpha(1-\alpha)v(\bs\lambda)]'\bs{\mathcal{I}}_{\bs\lambda n}[\alpha(1-\alpha)v(\bs\lambda)]$. With this definition of $\bs{\mathcal{I}}_n$, the expansion (\ref{LR0-score0})-(\ref{LR0-hessian1}) is written as (\ref{eq:LR0}) in terms of $\sqrt{n}t(\psi,\alpha)$.

We now show that $\bs{\mathcal{I}}_n \rightarrow_p \bs{\mathcal{I}}$. $\bs{\mathcal{I}}_{\bs\eta n} \rightarrow_p \bs{\mathcal{I}}_{\eta}$ holds trivially. For $\bs{\mathcal{I}}_{\bs\eta\bs\lambda n}$, it follows from Lemma \ref{lemma:KS2018_lemma7}(c) and the law of large numbers that $A_{ij} \rightarrow_p -\E[\nabla_{\bs\eta} l(\bs W;\psi^*,\alpha) \nabla_{\lambda_i\lambda_j}l(\bs W;\psi^*,\alpha)]$, giving $\bs{\mathcal{I}}_{\bs\eta\bs\lambda n} \rightarrow_p E\left[\bs s_{\bs\eta} \bs s_{\bs\lambda\bs\lambda}\t/\alpha(1-\alpha)\right]=\bs{\mathcal{I}}_{\bs\eta \bs \lambda}$. For $\bs{\mathcal{I}}_{\bs\lambda n}$, Lemma \ref{lemma:KS2018_lemma7}(d) and the law of large numbers imply that  $\sum_{i=1}^{q+2} \sum_{j=1}^{q+2} \sum_{k=1}^{q+2} \sum_{\ell=1}^{q+2} B_{ijk\ell} \lambda_i \lambda_j \lambda_k \lambda_\ell \rightarrow_p \\ - \sum_{i=1}^{q+2} \sum_{j=1}^{q+2} \sum_{k=1}^{q+2} \sum_{\ell=1}^{q+2} E[\nabla_{\lambda_i \lambda_j}l(\bs W;\psi^*,\alpha)\nabla_{\lambda_k \lambda_\ell}l(\bs W;\psi^*,\alpha)]\lambda_i \lambda_j \lambda_k \lambda_\ell$, where the factor $(8/4!)=1/3$ in $B_{ijk\ell}$ and the three derivatives on the right-hand side of Lemma \ref{lemma:KS2018_lemma7}(d) cancel each other out. Therefore, we have  $\bs{\mathcal{I}}_{\bs\lambda n} \rightarrow_p E\left[\bs s_{\bs\lambda\bs\lambda}\bs s_{\bs\lambda\bs\lambda}\t  /\alpha^2(1-\alpha)^2\right]=\bs{\mathcal{I}}_{\bs\lambda}$, and $\bs{\mathcal{I}}_n \rightarrow_p \bs{\mathcal{I}}$ follows.

We complete the proof of part (c) by showing that $\bs{\mathcal{I}} = E[\bs{s}(\bs{W}) \bs{s}(\bs{W})\t]$ is finite and non-singular. 
Note that $\bs s(\bs W) $ can be expressed in Hermite polynomials as in (\ref{eq:s_1_hermite}). Then, the finiteness of $\bs{\mathcal{I}}$ follows from  Assumption \ref{assumption:2}(a) and the definition of Hermite polynomials.

To show that $\bs{\mathcal{I}}$ is positive definite, it suffices to show that there exists no multicollinearity in $\bs s(\bs w)$. Suppose, to the contrary, that $\bs {s}(\bs w)$ is multicollinear and that there exists a non-zero vector $\bs a$ that solves the  equation $\bs a\t \bs s(\bs w) = 0$ for all values of $\bs w $. Partition $ s(\bs w)$ as $\bs s(\bs w)=(\bs s_{(\mu)} \t, \bs s_{(\beta)}\t )\t$  with $\bs s_{(\mu)} = (s_{\mu},s_{\sigma},s_{\lambda_{\mu \mu }} ,s_{\lambda_{\mu \sigma }} ,s_{\lambda_{\sigma \sigma }})\t$  and $\bs s_{(\beta)} = (\bs s_{\bs\beta }\t,\bs s_{\bs \gamma }\t, \bs s_{\lambda_{\mu\bs \beta }}\t,\bs s_{\lambda_{\sigma\bs \beta }}\t,\bs s_{\lambda_{\bs \beta \bs \beta }}\t)\t$,
where $\bs s(\bs w) $ is defined in (\ref{eq:s_1}) and  (\ref{eq:s_1_hermite}).
Similarly, partition $\bs a$ as $\bs a = (\bs a_{(\mu)}\t, \bs a_{(\beta)}\t )\t$ so that 
\begin{equation}\label{eq:as}
\bs a\t \bs s(\bs w) = \bs a_{(\mu)}\t \bs s_{(\mu)} + \bs a_{(\beta)}\t \bs s_{(\beta)}.
\end{equation}
By Assumption \ref{assumption:2}(b) and the property of Hermite polynomials, if $\bs a\t \bs s(\bs w) = \bs 0$ for all $\bs w$, then $\bs a_{(\beta)} = \bs 0$.  

Then, in view of (\ref{eq:as}), the stated result follows if we can show that $\bs a_{(\mu)}\t  \bs s_{(\mu)} = \bs 0$ for all $\bs w$ implies $\bs a_{(\mu)}=0$.
Suppose that
\begin{align*}
	\bs a_{(\mu)}\t  \bs s_{(\mu)} & = a_{\mu } \sum_{t=1}^T  H^{1*}_{i,t} +  (a_{\sigma } + a_{\lambda_{\mu \mu }}) \sum_{t=1}^T  H^{2*}_{i,t} +    \frac{a_{\lambda_{\mu \mu }}}{2} \sum_{t=1}^T \sum_{s \neq t} H^{1*}_{i,t} H^{1*}_{i,s}  \\
	& + a_{\lambda_{\mu \sigma }} \sum_{t=1}^T H^{3*}_{i,t} + a_{\lambda_{\mu \sigma }} \sum_{t=1}^T \sum_{s \neq t} H^{1*}_{i,t} H^{2*}_{i,s} +
	3  a_{\lambda_{\sigma \sigma }}\sum_{t=1}^T H^{4*}_{i,t} +  \frac{a_{\lambda_{\sigma \sigma }}}{2} \sum_{t=1}^T \sum_{s \neq t} H^{2*}_{i,t} H^{2*}_{i,s}  = 0
\end{align*}
for all $\bs w$,
where $H^{j*}_{i,t}$ for $j=1,2,3$ is defined in (\ref{eq:Hermite_polynomial})  in Appendix \ref{sec:appendix_score_1}.
 
Because the above equation holds for all values of $\bs w$, with the property of the Hermite polynomials, we have $a_{\mu } = 0, (a_{\sigma } + a_{\lambda_{\mu \mu }}) = 0, a_{\lambda_{\mu \mu }} = 0, a_{\lambda_{\mu \sigma }} = 0, a_{\lambda_{\sigma \sigma }} = 0$. This implies that $\bs a_{(\mu)}=0$.
Therefore,  no multicollinearity exists in $\bs s(\bs w)$ and $\bs{\mathcal{I}}$ is non-singular, proving part (c). 
\end{proof}

\begin{proof}[Proof of Proposition \ref{prop:t2_distribution}]


The proof  is similar to that of Proposition 3 in \cite{Kasahara2015a}.  

The proof of part (a) closely follows the proof of Theorem 1 of \cite{Andrews1999}. Let $\boldsymbol{T}_{n}: = \boldsymbol{\mathcal{I}}_{n}^{1/2}\sqrt{n}\bs{t}(\hat{\bs\psi}_\alpha,\alpha)$. Then, in view of (\ref{eq:LR0}), we have
\begin{align*}
o_p(1)&\leq L_n(\hat{\boldsymbol{\psi}}_\alpha,\alpha) - L_n(\boldsymbol{\psi}^*,\alpha)\\
 &= \boldsymbol{T}_{n}' \boldsymbol{\mathcal{I}}_{n}^{-1/2} \boldsymbol{S}_{n} - \frac{1}{2} ||\boldsymbol{T}_{n}||^2 + R_n(\hat{\boldsymbol{\psi}}_\alpha,\alpha)\\
 &= O_p(||\boldsymbol{T}_{n}||) - \frac{1}{2} ||\boldsymbol{T}_{n}||^2 + (1 + || \boldsymbol{\mathcal{I}}_{n}^{-1/2} \boldsymbol{T}_{n}||)^2 o_p(1)\\
 &= ||\boldsymbol{T}_{n}||O_p(1) - \frac{1}{2} ||\boldsymbol{T}_{n}||^2 + o_p(||\boldsymbol{T}_{n}||) + o_p(||\boldsymbol{T}_{n}||^2) + o_p(1),
\end{align*}
where the third equality holds because $\boldsymbol{\mathcal{I}}_{n}^{-1/2} \boldsymbol{S}_{n} = O_p(1)$ and $R_n(\hat{\boldsymbol{\psi}}_\alpha,\alpha) = o_p((1 + ||\boldsymbol{\mathcal{I}}_{n}^{-1/2}\boldsymbol{T}_{n}||)^2)$ from Propositions \ref {prop:vartheta_convergence} and \ref{prop:expansion}. Rearranging this equation yields $||\boldsymbol{T}_n||^2 \leq 2 ||\boldsymbol{T}_n|| O_p(1)+o_p(1)$. Denote the $O_p(1)$ term by $\varsigma_{n}$. Then, $(||\boldsymbol{T}_{n}||-\varsigma_{n})^2\leq \varsigma_{n}^2 + o_p(1) = O_p(1)$; taking its square root gives $||\boldsymbol{T}_{n}|| \leq O_p(1)$. In conjunction with $\boldsymbol{\mathcal{I}}_{n} \rightarrow_p \boldsymbol{\mathcal{I}}$, we obtain $\sqrt{n}\bs{t}(\hat{\bs\psi}_\alpha,\alpha) = O_p(1)$, and part (a) follows.

For part (b), noting that $L_n(\bs{{\psi}}^*,\alpha)= L_{0,n}(\bs{\gamma}^*_0,\bs{\theta}^*_0)$, 
write
\begin{align}
LR_n & =   \max_{\alpha \in [\epsilon, 1- \epsilon]} 2 \{  L_n(\bs{\hat{\psi}}_\alpha,\alpha)  -   L_n(\bs{{\psi}}^*,\alpha)  \} -  2\{ L_{0,n}(\hat{\bs{\gamma}}_0,\hat{\bs{\theta}}_0)   -   L_{0,n}(\bs{\gamma}^*_0,\bs{\theta}^*_0)   \}. \label{eq:split}
\end{align}

Define \begin{equation*}
\bs{S}_n = \begin{pmatrix}
\bs{S}_{\bs{\eta} n } \\
\bs{S}_{\bs{\lambda} n }
\end{pmatrix}:=  \begin{pmatrix}
n^{-1/2} \sum_{i=1}^n{\bs s}_{\bs\eta}(\bs W_i)\\
n^{-1/2} \sum_{i=1}^n{\bs s}_{\bs\lambda\bs\lambda}(\bs W_i)
\end{pmatrix}, \quad \begin{matrix}
\bs{S}_{\bs{\lambda}, \bs{\eta} n }  := \bs{S}_{\bs{\lambda} n } - \bs{\mathcal{I}}_{\bs{\lambda} \bs{\eta}}  \bs{\mathcal{I}}_{\bs{\eta}}^{-1} \bs{S}_{\bs{\eta} n } , \quad  \bs{G}_{\bs{\lambda}, \bs{\eta} n } := \bs{\mathcal{I}}_{\bs{\lambda},\bs{\eta}}^{-1} \bs{S}_{\bs{\lambda}, \bs{\eta} n }, \\
\bs{t}_{\bs{\eta}, \bs{\lambda} }   := \bs{t}_{\bs{\eta} } -
\bs{\mathcal{I}}_{\bs{\eta} } \bs{\mathcal{I}}_{\bs{\eta} \bs{\lambda} }^{-1} \bs{t}_{ \bs{\lambda} }(\bs{\lambda},\alpha),
\end{matrix}
\end{equation*} 
and split the quadratic form in (\ref{eq:LR0}) to obtain
\begin{align}\label{eq:LR2}
2\{L_n(\bs{\psi},\alpha)-  L_n(\bs{\psi}^*,\alpha)\} =  B_n (\sqrt{n} \bs{t}_{\bs{\eta},\bs{\lambda}} )  + C_n (\sqrt{n} \bs{t}_{\bs{\lambda}}(\bs{\lambda},\alpha) ) +R_n(\bs\psi,\alpha),
\end{align}
where \begin{equation} \label{eq:split2}
\begin{split}
B_n ( \bs{t}_{\bs{\eta},\bs{\lambda}} )  & = 2 \bs{t}_{\bs{\eta},\bs{\lambda}}^\top \bs{S}_{\bs{\eta} n} - \bs{t}_{\bs{\eta},\bs{\lambda}}^\top \bs{\mathcal{I}}_{\bs{\eta}} \bs{t}_{\bs{\eta},\bs{\lambda}},\\
C_n ( \bs{t}_{\bs{\lambda}} ) & = 2 \bs{t}_{\bs{\lambda}}^\top \bs{S}_{\bs{\lambda},\bs{\eta} n} - \bs{t}_{\bs{\lambda}} ^\top \bs{\mathcal{I}}_{\bs{\lambda},\bs{\eta}} \bs{t}_{\bs{\lambda}}\\
& = \bs{G}_{\bs{\lambda}, \bs{\eta} n}^\top \bs{\mathcal{I}}_{\bs{\lambda},\bs{\eta}} \bs{G}_{\bs{\lambda}, \bs{\eta} n} - (\bs{t}_{\bs{\lambda}} - \bs{G}_{\bs{\lambda}, \bs{\eta} n})^\top \bs{\mathcal{I}}_{\bs{\lambda},\bs{\eta}}  (\bs{t}_{\bs{\lambda}}  - \bs{G}_{\bs{\lambda}, \bs{\eta} n}),
\end{split}
\end{equation}
with
$\bs{G}_{\bs{\lambda},\bs{\eta}n}\overset{d}{\rightarrow} \bs{G}_{\bs{\lambda},\bs{\eta}}=(\bs{\mathcal{I}}_{ \bs{\lambda},\bs{\eta}} )^{-1} \bs{S}_{\bs{\lambda},\bs{\eta}  } $ and $\bs{S}_{\bs{\lambda}, \bs{\eta} n } \overset{d}{\to} \bs{S}_{\bs{\lambda},\bs{\eta}  } \sim N(\bs 0, \bs{\mathcal{I}}_{ \bs{\lambda},\bs{\eta}} )$. 
In addition, $R_n(\hat{\boldsymbol{\psi}}_\alpha,\alpha) = o_p(1)$ holds from  Proposition \ref{prop:expansion}(a) and $\sqrt{n}\bs{t}(\hat{\bs\psi}_\alpha,\alpha) = O_p(1)$.

Because $\Delta_{(\gamma,\theta)}f(x; \hat{\bs{\gamma}}_0^*,\hat{\bs{\theta}}_0^*)$ is identical to $\Delta_{\bs\eta} f(x; \bs\psi^*,\alpha)$, a standard analysis gives $2 [L_{0,n}(\hat{\bs{\gamma}}_0,\hat{\bs{\theta}}_0)   -   L_{0,n}(\bs{\gamma}^*_0,\bs{\theta}^*_0)]  = \max_{\bs{t}_{\bs{\eta} }}  B_n (\sqrt{n} \bs{t}_{\bs{\eta} } ) + o_p(1)$.  
Note that the possible values of both $\sqrt{n}\bs{t}_{\bs{\eta}}$ and $\sqrt{n}\bs{t}_{\bs{\eta},\bs{\lambda}}$ approach $\R^{q+2}$. Therefore,  in view of (\ref{eq:LR2}) and (\ref{eq:split2}), we can write equation (\ref{eq:split}) as
\begin{equation}\label{eq:LRTS_expansion_homo_repar}
LR_n  = C_n (\sqrt{n} \bs{t}_{\bs{\lambda}}(\hat{\bs{\lambda}}, \alpha))  + o_p(1),
\end{equation}
where $\hat{\bs{\lambda}}$ is as defined in (\ref{eq:pmle}).

%
The asymptotic distribution of $LR_n$ follows from applying Theorem 3(c)  of \cite[][p. 1362]{Andrews1999} to (\ref{eq:LR2}) and (\ref{eq:LRTS_expansion_homo_repar}).
First,   Assumption 2 of \cite{Andrews1999} holds because Assumption 2* of \cite{Andrews1999} holds because of Proposition 2(a). 
Second, Assumption 3 of \cite{Andrews1999} holds with $B_T = n^{1/2}$  and $T=n$ because $\bs{S}_{\bs{\lambda}, \bs{\eta} n } \overset{d}{\to} \bs{S}_{\bs{\lambda},\bs{\eta}  } \sim N(\bs 0, \bs{\mathcal{I}}_{ \bs{\lambda},\bs{\eta}} )$
 and $\bs{\mathcal{I}}_{ \bs{\lambda},\bs{\eta}} $ is non-singular.
Assumption 4 of \cite{Andrews1999} holds from part (a).
Assumption 5 of \cite{Andrews1999} follows from Assumption 5* and Lemma 3 of \cite{Andrews1999} with $b_T=n^{1/2}$ because $\alpha(1-\alpha)v(\Theta_{\bs\lambda})$ is locally equal to $\Lambda_{\bs{\lambda}}$.  
Therefore, it follows from Theorem 3(c) of \cite{Andrews1999} that $C_n (\sqrt{n} \bs{t}_{\bs{\lambda}}(\hat{\bs{\lambda}},\alpha) )   \overset{d}{\to} (\hat{\bs{t}}_{\bs\lambda} )^\top \bs{\mathcal{I}}_{\bs\lambda,\bs\eta} \hat{\bs{t}}_{\bs\lambda}$, 
where $\hat{\bs{t}}_{\bs\lambda}$ is defined by (\ref{eq:t_lambda_def}).

\end{proof}

\begin{proof}[Proof of Proposition \ref{prop:infinite_fisher}]
	Under $H_{2,0}$,  we obtain $\vartheta_{M_0+1} \in \Upsilon_{2h}^*$, 
	\begin{equation}
		\begin{split}
		& \E [\{ \nabla_{\alpha_h}  \log f_{M_0+1}  (\bs W_i, \vartheta_{M_0+1} )\}^2  ]  \\
		& = \int  \frac{ \{ f( \bs w; \bs \mu_{h}, \bs \sigma_{h} )   - f( \bs w; \bs \mu_{M_0}^*, \bs \sigma_{M_0}^* )  \}^2 }{\sum_{j=1}^{M_0} \alpha_j^* f( \bs w; \bs \mu_j^*, \bs \sigma_j^* ) } d \bs w \\
		& = \int  \frac{ \{ f( \bs w; \bs \mu_{h}, \bs \sigma_{h} ) \}^2 }{\sum_{j=1}^{M_0} \alpha_j^* f( \bs w; \bs \mu_j^*, \bs \sigma_j^* ) } d \bs w  + \int  \frac{ \{  f( \bs w; \bs \mu_{M_0}^*, \bs \sigma_{M_0}^* )  \}^2 }{\sum_{j=1}^{M_0} \alpha_j^* f( \bs w; \bs \mu_j^*, \bs \sigma_j^* ) } d \bs w -   2  \int  \frac{  f( \bs w; \bs \mu_{h}, \bs \sigma_{h} )  f( \bs w; \bs \mu_{M_0}^*, \bs \sigma_{M_0}^* )   }{\sum_{j=1}^{M_0} \alpha_j^* f( \bs w; \bs \mu_j^*, \bs \sigma_j^* ) } d \bs w .
	\end{split} \label{eq:infinite_fisher}
 \end{equation}
The latter two terms on the right-hand side of (\ref{eq:infinite_fisher}) are bounded because 
$f( \bs w; \bs \mu_{M_0}^*, \bs \sigma_{M_0}^* )  / \sum_{j=1}^{M_0} \alpha_j^* f( \bs w; \bs \mu_j^*, \bs \sigma_j^* )  \le (1 / \alpha_{M_0}^*)$ for any $\bs w$ and $f(\bs w;\mu,\sigma)$ integrates to one. 
Therefore, the left-hand side of (\ref{eq:infinite_fisher}) goes to infinity if and only if the first term on the right-hand side of (\ref{eq:infinite_fisher}) goes to infinity. 

Because $\max_j \alpha_j \le \sum_{j}^{M_0} \alpha_j \le M_0 \max_j \alpha_j$, we obtain 
\[  \frac{1}{M_0}  \frac{ \{  f( \bs w; \bs \mu_{h}, \bs \sigma_{h} )  \}^2 }{ \max_j  \{ \alpha_j^* f( \bs w; \bs \mu_j^*, \bs \sigma_j^* )  \} }  \le   \frac{ \{  f( \bs w; \bs \mu_{h}, \bs \sigma_{h} )  \}^2 }{ \sum_{j=1}^{M_0} \alpha_j^* f( \bs w; \bs \mu_j^*, \bs \sigma_j^* ) }  \le   \frac{ \{  f( \bs w; \bs \mu_{h}, \bs \sigma_{h} )  \}^2 }{ \max_j  \{ \alpha_j^* f( \bs w; \bs \mu_j^*, \bs \sigma_j^* )  \} }. \]

Without loss of generality, we assume that $\sigma_{M_0}^* = \max \{ \sigma_1^*,\ldots, \sigma_{M_0}^2 \}$ and that the maximum is unique. Then, there exists $M \in (0, \infty )$, such that $\max_j \{ \alpha_j^* f (\bs w, \bs \mu_j^*, \bs \sigma_j^2) \} = \alpha_{M_0}^* f (\bs w, \bs \mu_{M_0}^*, \bs \sigma_{M_0}^2)$ when $| y_t| > M  ~ \forall t = 1,\ldots, T$.
Note that 
\begin{equation}\label{bound}
\begin{split}
\frac{ \{  f( \bs w; \bs \mu_{h}, \bs \sigma_{h} )  \}^2 }{ f( \bs w; \bs \mu_{M_0}^*, \bs \sigma_{M_0}^* )  \} } & = \prod_{t=1}^T \frac{ (\sigma_{M_0}^{*})^2 }{(2 \pi )^{1/2} \sigma_h^2 } \exp \left\{ -\frac{1}{\sigma_h^2} (y_t - \mu_h)^2 + \frac{1}{2 (\sigma_{M_0}^{*})^2 } (y_t - \mu_{M_0}^{*})^2 \right\} \\
& = \frac{(\sigma_{M_0}^*)^{2T}}{(2 \pi )^{1/2} \sigma_h^{2T} } \exp \left\{ -  \frac{1}{\sigma_h^2}\sum_{t=1}^T( y_t - \mu_h)^2 + \frac{1}{2 (\sigma_{M_0}^{*})^2} \sum_{t=1}^T(y_t - \mu_{M_0}^{*})^2 \right\}.
\end{split}
\end{equation}

The stated result follows because the integral of the right-hand side of (\ref{bound})  over $|y| \ge M$ is finite if $\sigma_h^2/ \sigma_{M_0}^{2*} < 2$ and infinite if  $\sigma_h^2/ \sigma_{M_0}^{2*} > 2$. When $\sigma_h^2/ \sigma_{M_0}^{2*} = 2$, it is finite if $\mu_h = \mu_{M_0}^*$ and infinite if $\mu_h \neq \mu_{M_0}^*$.
\end{proof}
\begin{proof}[Proof of Proposition \ref{prop:vartheta_convergence_M}]
Our panel data model can be viewed as a special case of the $T$-dimensional multivariate normal mixture models, where the variance-covariance matrix for each component is given by a $T\times T$ diagonal matrix, $\Sigma_j:=\text{diag}(\sigma_j^2,\ldots,\sigma_2^2)$. 
\cite{Chen2009a} provide the consistency proof for the PMLE for a multivariate normal mixture under their conditions C1--C3 for the penalty function. However, \cite{Alexandrovich2014} identifies a weakness in the proof of  \cite{Chen2009a} and provides an alternative consistency proof by strengthening condition C3 of  \cite{Chen2009a}. Their $p_n(G)$ and $\tilde p_n(G)$ correspond to our  $\tilde p_n(\bs\vartheta_2)$ and $p_n(\bs\vartheta_2)$, respectively; consequently,
the conditions C1 and C2 in \cite{Chen2009a} and a version of condition C3 strengthened by \cite{Alexandrovich2014} can be stated in our notation as follows:
\begin{description}
\item[C1.] 
 The penalty function is written as $\tilde p_n(\bs\vartheta_M)=\sum_{j=1}^M p_n(\sigma_j^2)$.
 \item[C2.] For any  fixed $\bs\vartheta_M$ with $\sigma_j^2>0$ for $j=1,2,...,M$, we have $\tilde p_n(\bs\vartheta_M)=o(n)$ and $\sup_{\bs\vartheta_M\in\bs\Theta_M} \max\{0,\tilde p_n(\bs\vartheta_M)\}=o(n)$. In addition, $\tilde p_n(\bs\vartheta_M)$ is differentiable with respect to $\bs\vartheta_M$ and as $n\rightarrow \infty$, $\nabla_{\bs\vartheta_M} \tilde p_n(\bs\vartheta_M)=o(n^{1/2})$ at any fixed $\bs\vartheta$ such that $\sigma_j^2>0$ for $j=1,2,...,M.$
\item[A version of C3 by \cite{Alexandrovich2014}.] For a sufficiently large $n$, $p_n(\sigma_j^2) \leq \left(\frac{3}{4}\sqrt{n\log\log n}\right)\log(\sigma_j^2)$, when  $\sigma^2_j<c n^{-2}$ for some $c>0$.
\end{description} 
The consistency of the PMLE, $\hat {\bs\vartheta}_{M_0}$, and  $\hat {\bs\vartheta}_{M_0+1}$, follows from Theorems 1 and 3 of \cite{Chen2009a} and Corollary 3 of \cite{Alexandrovich2014} if we can show that the above three conditions hold for our penalty function  (\ref{eq:penalty}).
Given (\ref{eq:penalty}), C1 trivially holds. Under Assumption \ref{assumption:1}(b), C2 also holds because $a_n=O(n^{1/4-\zeta})$ with $\zeta>0$ implying $a_n=o(n)$ or $o(n^{1/2})$, and $\Delta_{\sigma_j^2} \tilde p_n(\bs\vartheta) = - a_n ( - \sigma_0^2/(\sigma_j^2) + 1/\sigma_j^2) = o(n^{1/4})$ if $\sigma_j^2>0$. For C3, suppose that $\sigma_j^2< n^{-2}$. Then, because $a_n=o(n^{1/4})$  and $a_n>0$, 
  $p_n(\sigma_j) = - a_n \left(  \sigma_j^{-2}\sigma_0^2  + 2\log (\sigma_j/\sigma_0) -1\right)  < - c_n\left( n^{9/4} \sigma_0^2  - 2n^{5/4}\log (n\sigma_0) -n^{1/4} \right) <  \left(\frac{3}{4}\sqrt{n\log\log n}\right)2\log(n) $ when $n$ is sufficiently large, where $c_n$ is a sequence of positive numbers that are bounded.   Therefore,  $\tilde p_n(\bs\vartheta_M)$ satisfies the above three conditions, and  the stated result follows from Theorems 1 and 3 of \cite{Chen2009a} and Corollary 3 of \cite{Alexandrovich2014}.
\end{proof}

\begin{proof}[Proof of Proposition \ref{prop:tm0_distribution}] 
    For $h = 1,\ldots, M_0$,  let $\mathcal{N}_h^* \subset \Theta_{\bs{\vartheta}_{M_0 + 1}}(\epsilon) $ be a sufficiently small closed neighborhood of $\Upsilon^*_{1h}$ such that $\alpha_h,\alpha_{h+1} > 0$ holds and $\Upsilon^*_{1k} \not\subset \mathcal{N}_h^*$ if $k \neq h$. Consider the following one-to-one reparameterization from the $(M_0 + 1)$-component model parameter $\bs{\vartheta}_{M_0 + 1} = (\alpha_1,\ldots,\alpha_{M_0},\bs{\theta}^\top_1,\ldots,\bs{\theta}^\top_h,\bs{\theta}^\top_{h+1},\ldots,\bs{\theta}^\top_{M_0 + 1},\bs{\gamma}^\top)^\top$.
     Similar to (\ref{eq:m0_repar2}), the one-to-one reparameterization  for testing the null hypothesis $H_{0,1h}$ is given by
    \begin{equation*}
    \begin{pmatrix}
    \bs{\lambda}_h \\
    \bs{\nu}_h
    \end{pmatrix} := \begin{pmatrix}
    \bs{\theta}_{h} - \bs{\theta}_{h+1} \\
    \tau \bs{\theta}_{h} + (1 - \tau) \bs{\theta}_{h+1}
    \end{pmatrix} \text{ so that }
    \begin{pmatrix}
    \bs{\theta}_{h} \\
    \bs{\theta}_{h+1}
    \end{pmatrix} = \begin{pmatrix}
    \bs{\nu} + ( 1- \tau) \bs{\lambda} \\
    \bs{\nu} - \tau \bs{\lambda}
    \end{pmatrix},
    \end{equation*} and  $\alpha_j$ is reparameterized for $j=1,2,...,M_0$ as
    \begin{align*}
     (\pi_1,\ldots,\pi_{h-1},\pi_h,\pi_{h+1},\ldots , \pi_{M_0 - 1}) &= (\alpha_1,\ldots,\alpha_{h-1}, (\alpha_h + \alpha_{h+1}),\alpha_{h+2},\ldots,\alpha_{M_0})\\
 \tau  &=  {\alpha_h}/({\alpha_h + \alpha_{h+1}})
      \end{align*}
    so that $\pi_h  = \alpha_h + \alpha_{h+1}$ and  $\pi_{M_0} = 1 - \sum_{j=1}^{M_0 - 1} \pi_j$. 
    
    Collect the reparameterized parameters except $\tau$ as 
\begin{equation*}
 \bs{\psi}_{h,\tau}  = (\bs\eta\t,\bs{\lambda}^\top_h)^\top\quad\text{with}\quad\bs \eta= (\pi_1,\ldots,\pi_{M_0-1},\bs{\theta}^\top_1,\ldots,\bs{\theta}^\top_{h-1},\bs{\nu}^\top_h,\bs{\theta}^\top_{h+2},\ldots, \bs{\theta}^\top_{M_0+1}, \bs{\gamma})^\top.
\end{equation*}  
 In the reparameterized model, the null restriction $\bs{\theta}_h  = \bs{\theta}_{h+1}$ implied by $H_{0,1h}$ holds if and only if $\bs{\lambda}_h = 0$. Under $H_{0,1h}$, we have $\lambda_h^*=0$ and 
 $\eta^*=(\alpha_1^*,...,\alpha_{M_0-1}^*, (\bs{\theta}^*_1)^\top,\ldots, (\bs{\theta}^*_{M_0})^\top, (\bs{\gamma}^*)^\top)^\top$.  
  Define the log-likelihood under the reparameterized parameters as
    \begin{equation*}
    f_{M_0+1}^{h}(\bs{w};\bs{\psi}_{h,\tau},\tau )  = \pi_h g^h(\bs{w},\bs{\psi}_{h,\tau}, \tau) +  \sum_{j=1}^{h-1} \pi_j f(\bs{w};\bs{\gamma},\bs{\theta}_j) +  \sum_{j=h}^{M_0} \pi_{j+1} f(\bs{w};\bs{\gamma},\bs{\theta}_{j+1}),
    \end{equation*}
    where $g^h(\bs{w},\bs{\psi}_{h,\tau}, \tau)$ is defined similarly to (\ref{eq:repar}) as
    \begin{equation}
    g^h(\bs{w},\bs{\psi}_{h,\tau}, \tau) = \tau f(\bs{w};\bs{\gamma},\bs{\nu}_h + (1 - \tau) \bs{\lambda}_h) + (1 - \tau) f(\bs{w};\bs{\gamma},\bs{\nu}_h - \tau \bs{\lambda}_h).
    \end{equation}

 Define the local PMLE of $\bs{\psi}_{h,\tau}$ by
 \begin{equation}\label{eq:psi_pmle}
 \bs{\hat{\psi}}_{h,\tau} := \argmax_{\bs{\psi}_{h,\tau}  \in \mathcal{N}_h^*}  L^h_n(\bs{\psi}_{h,\tau}, \tau)  + \sum_{j=1}^{M_0} p_n(\sigma_j^2(\bs{\psi}_{h,\tau}, \tau) ),
 \end{equation}
 where $L_n^h(\bs{\psi}_{h,\tau}, \tau) := \sum_{i=1}^N \log g^h(\bs{W}_i;\bs{\psi}_{h,\tau}, \tau) $ and $\sigma_j^2(\bs{\psi}_{h,\tau}, \tau) $ is the value of $\sigma_j^2$ implied by the values of $\bs{\psi}_{h,\tau}$ and $\tau$.
 Because $\bs{\psi}_{h,\tau}^*$ is the only parameter value in $\mathcal{N}_h^*$ that generates the true density, $\bs{\hat{\psi}}_{h,\tau} - \bs{\psi}_{h,\tau}^* = o_p(1)$ follows Proposition \ref{prop:t2_distribution}.

 For $\epsilon \in (0, 1/2)$, define the LRTS for testing $H_{0,1h}$ as $LR_{n}^{M_0,h}  := \max_{\tau \in [\epsilon,1 -\epsilon ]} 2 ( L_n^h(\hat{\bs{\psi}}_{h,\tau}, \tau)  - L_{0,n}( \hat{\bs{\vartheta}}_{M_0}))$. Because $\hat{\sigma}_j^2 - \sigma_{0,j}^2 = O_p(n^{-1/4})$ under the null hypothesis (cf. Proposition \ref{prop:t2_distribution}(a)),  we have $ \tilde p_n(\bs{\vartheta}_{M_0+1 })=o_p(1)$ by Assumption \ref{assumption:3}(c), and $PLR_{n}^{M_0,h}-LR_{n}^{M_0,h}=o_p(1)$ follows for $h=1,...,M_0$.

Then, in view of (\ref{eq:LR_M0_max}), the stated result holds if 
 \begin{equation}
(LR_{n}^{M_0,1},\ldots,LR_{n}^{M_0,M_0})^\top \overset{d}{\to} (\hat{\bs{t}}^1_{\lambda})^\top \bs{\mathcal{I}}^1_{\eta,\lambda} (\hat{\bs{t}}^1_{\lambda}) ,   \ldots, (\hat{\bs{t}}^{M_0}_{\lambda})^\top \bs{\mathcal{I}}^{M_0}_{\eta,\lambda} (\hat{\bs{t}}^{M_0}_{\lambda}) )^\top.
\end{equation}
Observe that $L^h_n(\bs{\psi}_{h,\tau}, \tau)  -  L^h_n(\bs{\psi}^*_{h,\tau}, \tau) $ admits the same expansion as
 $L_n(\bs{\hat{\psi}},\alpha)  -   L_n(\bs{{\psi}}^*,\alpha)  $ in (\ref{eq:LR1}) and (\ref{eq:LR2}) when
  $(\alpha,\bs{t}(\bs{\psi}, \alpha), \bs{t}_{\bs\lambda}(\bs{\lambda}, \alpha), \bs S_n, \bs G_n, \bs{\mathcal{I}}_n,R_n(\bs\psi,\alpha))$ is replaced with   $(\tau,\bs{t}^h(\bs{\psi}^h, \tau), \bs{t}_{\bs\lambda}^h(\bs{\lambda}^h, \tau),  \bs S_n^h, \bs G_n^h,\bs{\mathcal{I}}_n^h,R_n^h(\bs\psi^h,\tau))$, where $(\bs S_n^h, \bs{\mathcal{I}}_n^h)$ is defined similarly to $(\bs S_n, \bs{\mathcal{I}}_n)$ but  
  $(\bs{s}_{\bs{\eta}},\bs{s}_{\bs{\lambda\lambda}})$ is replaced with $(\tilde{\bs{s}}_{\bs{\eta}}, {\bs{s}}^h_{\bs{\lambda\lambda}})$ and $\bs G_n^h := (\bs{\mathcal{I}}_n^h)^{-1} \bs S_n^h$.  Applying the proof of Proposition \ref{prop:expansion}, we have ${\bs S}_n^h\overset{d}{\to} {\bs S}^h \sim N(\bs 0, \bs{\mathcal{I}}^h)$ and $\bs{\mathcal{I}}_n^h\overset{p}{\to} \bs{\mathcal{I}}^h$. Then, (\ref{eq:psi_pmle}) follows from the proofs of  Propositions \ref{prop:expansion} and \ref{prop:t2_distribution} for each local PMLE when $(\bs G_n, \hat{\bs t}_{\bs\lambda},\bs{\mathcal{I}}_{\bs\lambda,\bs\eta})$ is replaced with  $(\bs G_n^h, \hat{\bs t}_{\bs\lambda}^h,\bs{\mathcal{I}}_{\bs\lambda,\bs\eta}^h)$ and the results are collected; note that $(\bs S_n^1,...,\bs S_n^{M_0})\overset{d}{\to} (\bs S^1,...,\bs S^{M_0})$. 
\end{proof}

\begin{proof}[Proof of Proposition \ref{prop:KS18_prop4}]  

The proof  is similar to that of Proposition 7 in \cite{Kasahara2015a}.  
    Let $\omega_{n}^h$ denote the sample counterpart of $(\hat{\bs t}^h_{\bs\lambda})^\top \bs{\mathcal{I}}^h_{\bs\lambda,\bs\eta} \hat{\bs t}^h_{\bs\lambda}$ in Proposition \ref{prop:tm0_distribution} such that the LRTS satisfies $2 \{  L_n^h(\hat{\bs{\psi}}_{h,\tau},\tau) -  L_{0,n}(\hat{\bs{\vartheta}}_{M_0})\}  = \omega_{n}^h + o_p(1)$,
    where $\hat{\bs{\psi}}_{h,\tau} $ is the local PMLE as defined in (\ref{eq:psi_pmle}) and $\omega_{n}^h$ is defined similarly to $C_n(\sqrt{n} \bs t_{\bs\lambda}(\hat{\bs\lambda},\alpha))$ in (\ref{eq:LRTS_expansion_homo_repar}) but with $(\bs t_{\bs\lambda}(\hat{\bs\lambda},\alpha),\bs S_{n},\bs G_n,\bs{\mathcal{I}}_{\bs\lambda,\bs\eta})$ replaced with $(\bs t_{\bs\lambda}^h(\hat{\bs\lambda}^h,\tau),\bs S_n, \bs G^h_{n},\bs{\mathcal{I}}^h_{\bs\lambda,\bs\eta})$ in the proof of Proposition \ref{prop:tm0_distribution}.
     
  First, we show that $M_n^{h(1)}(\tau_0) =2 \{  PL_n^h( \bs{\vartheta}_{M_0 + 1}^{h(1)}(\tau_0),\tau_0 ) -  L_{0,n}(\hat{\bs{\vartheta}}_{M_0})\}= \omega_{n}^h + p(\tau_0) + o_p(1)$.  Define $\bs{\vartheta}_{M_0 + 1}^{h *}(\tau_0)$ by the value of $\bs{\vartheta}_{M_0 + 1}$ in  $\Theta_{\bs{\vartheta}_{M_0 +1}}^h(\tau_0) := \{ \bs\vartheta \in  {\bs{\Psi}}_h^* : \alpha_h / (\alpha_h + \alpha_{h+1} ) =  \tau_0 \}$. Because $\bs{\vartheta}_{M_0 + 1}^{h *}(\tau_0)$ is the only value of $\bs{\vartheta}_{M_0 + 1}$  that yields the true density if $\bs{\vartheta}_{M_0 + 1}  \in  {\bs{\Psi}}_h^*$ in (\ref{Omega}) and $\alpha_h / (\alpha_h + \alpha_{h+1} ) =  \tau_0$, $\bs{\vartheta}_{M_0 + 1}^{h(1)}(\tau_0) $ equals a reparameterized local PMLE in the neighborhood of $\bs{\vartheta}_{M_0 + 1}^{h *}(\tau_0)$, and $\bs{\vartheta}_{M_0 + 1}^{h(1)}(\tau_0) -\bs{\vartheta}_{M_0 + 1}^{h *}(\tau_0)=o_p(1)$ holds in view of Proposition \ref{prop:vartheta_convergence_M}. Furthermore, by the consistency of $\sigma_j^{h(1)}$ and  $a_n=O(1)$, we have $\tilde p( \bs{\vartheta}_{M_0 + 1}^{h(1)}(\tau_0)) \overset{p}{\rightarrow} 0$. 
   Therefore, $M_n^{h(1)}(\tau_0) = \omega_{n}^h + p(\tau_0) + o_p(1)$ follows from repeating the proof of Proposition \ref{prop:tm0_distribution}. Finally, $EM_n^{(1)}\overset{d}{\rightarrow} \max_{h=1}^{M_0} \{\omega_{n}^h\} $ holds because $\{0.5\}\in \mathcal{T}$ and $p(0.5)=0$.
    
    We proceed to show that $M_n^{h(K)}(\tau_0) = \omega_{n}^h +p(\tau_0) +  o_p(1)$ for any finite $K$. Because a generalized EM step never decreases likelihood \citep{Dempster1977}, we have 
    \begin{equation}\label{eq:PL-K}
      PL_n( \bs{\vartheta}_{M_0 + 1}^{h(K)}(\tau_0),\tau^{h(K)}(\tau_0) )  > PL_n( \bs{\vartheta}_{M_0 + 1}^{h(1)}(\tau_0),\tau^{h(1)}(\tau_0) ).
     \end{equation} 
    Therefore, it follows from Theorem 1 of \cite{Chen2009a}, Lemma \ref{tau_update} in Appendix \ref{sec:appendixc}, and induction that $\bs{\vartheta}_{M_0 + 1}^{h(K)}(\tau_0)  - \bs{\vartheta}_{M_0 + 1}^{h*} = o_p(1)$ for any finite $K$.
    Let $\tilde{\bs{\vartheta}}_{M_0 + 1}^h$ be the maximizer of $PL_{M_0 +1}(\bs{\vartheta}_{M_0 +1},\tau^{h(K)}(\tau_0))$ under the constraint of $\alpha_h / ( \alpha_h + \alpha_{h+1}) = \tau^{h(K)}(\tau_0)$ in an arbitrary small neighborhood of $\bs{\vartheta}^{h *}_{M_0+1}(\tau^{(K)})$. Then, $2 \{  PL_n^h( \tilde{\bs{\vartheta}}_{M_0 + 1}^h,\tau^{h(K)}(\tau_0)) -  L_{0,n}(\hat{\bs{\vartheta}}_{M_0})\}  = \omega_{n}^h  + p(\tau_0) + o_p(1)$ holds from the definition of $\tilde{\bs{\vartheta}}_{M_0 + 1}^h$ and $\tilde p( \tilde{\bs{\vartheta}}_{M_0 + 1}^h) \overset{p}{\rightarrow} 0$   by repeating the proof of Proposition \ref{prop:tm0_distribution}.  It also follows from  the consistency of $\bs{\vartheta}_{M_0 + 1}^{h(K)}(\tau_0)$ that $PL_{n}(\tilde{\bs{\vartheta}}_{n}^h,\tau^{h(K)}(\tau_0)) \ge PL_{n}(\bs{\vartheta}_{M_0 + 1}^{h(K)}(\tau_0),\tau^{h(K)}(\tau_0)) + o_p(1)$.  Therefore, in view of (\ref{eq:PL-K}), we have
        \begin{equation} \label{pl-ineq}
    PL_{n}(\tilde{\bs{\vartheta}}_{n}^h,\tau^{h(K)}(\tau_0)) \ge PL_{n}(\bs{\vartheta}_{M_0 + 1}^{h(K)}(\tau_0),\tau^{h(K)}(\tau_0)) + o_p(1)\ge PL_n( \bs{\vartheta}_{M_0 + 1}^{h(1)}(\tau_0),\tau^{h(1)}(\tau_0)).
    \end{equation}
 Finally,  because $2 \{  PL_n( \tilde{\bs{\vartheta}}_{M_0 + 1}^h,\tau^{h(K)}(\tau_0)) -  L_{0,n}(\hat{\bs{\vartheta}}_{M_0})\}  = \omega_{n}^h +p(\tau_0) + o_p(1)$ and $2 \{  PL_n( \bs{\vartheta}_{M_0 + 1}^{h(1)}(\tau_0),\tau^{h(1)}(\tau_0) ) -  L_{0,n}(\hat{\bs{\vartheta}}_{M_0})\}  = \omega_{n}^h +p(\tau_0)  + o_p(1)$, it follows from (\ref{pl-ineq}) that $M_n^{h(K)}(\tau_0) =2 \{  PL_n( \bs{\vartheta}_{M_0 + 1}^{h(K)}(\tau_0) ) -  L_{0,n}(\hat{\bs{\vartheta}}_{M_0})\}  = \omega_{n}^h +p(\tau_0)+ o_p(1)$ holds for all $h$. The stated result then follows from the definition of $EM_n^{(K)}$  and  $\{0.5\}\in \mathcal{T}$. 
 \end{proof}

\begin{proof}[Proof of Proposition \ref{prop:local-power}]
Let $\bs\psi_n=((\bs\nu^*)\t,\bs\lambda_n\t)\t$ be the value of $\bs\psi$ under $H_{1,n}: \bs\vartheta=  \bs\vartheta_{2,n}$ and let $\bs{h}=(\bs{0}\t,\bs{h}_{\bs\lambda}\t)\t$, where $\bs{h}_{\bs\lambda}$ is defined by \eqref{eq:h}.
Let $\mathbb{P}_{\bs\vartheta}$ be the probability measure on $\{\bs W_i\}_{i=1}^n$ under $\bs\vartheta$.  
Denote the log-likelihood ratio of $\mathbb{P}_{\bs\vartheta_{2,n}}$  to $\mathbb{P}_{\bs\vartheta_{2}^*}$ by $\log \left( \frac{d\mathbb{P}_{\bs\vartheta_n}}{d\mathbb{P}_{\bs\vartheta^*}}\right)=L_n(\bs\psi_n,\alpha^*) - L_n(\bs\psi^*,\alpha^*)$. Then,  it follows from (\ref{eq:LR0}) and Proposition \ref{prop:expansion} that
\begin{equation}\label{eq:likelihoodratio}
\log \frac{d\mathbb{P}_{\bs\vartheta_{2,n}}}{d\mathbb{P}_{\bs\vartheta^*_2}} = \bs{h} \bs{S}_n - \bs{h}\t \bs{\mathcal{I}}\bs{h}/2 + o_{p}(1)\quad\text{under $\mathbb{P}_{\bs\vartheta_{2}^*}$}.
\end{equation}
Furthermore, because $ \bs{S}_n \overset{d}{\rightarrow} N(\bs{0},\bs{\mathcal{I}})$ under $\mathbb{P}_{\bs\vartheta_{2}^*}$, $\frac{d\mathbb{P}_{\bs\vartheta_n}}{d\mathbb{P}_{\bs\vartheta^*}}$ converges in distribution under $\mathbb{P}_{\bs\vartheta_{2}^*}$ to $\exp(N(\mu,\sigma^2))$ with $\mu = -(1/2) \bs{h}\t \bs{\mathcal{I}}\bs{h}$ and $\sigma^2= \bs{h}\t \bs{\mathcal{I}}\bs{h}$ so that $E(\exp(N(\mu,\sigma^2))=1$. Consequently, $\mathbb{P}_{\bs\vartheta_{2,n}}$ is mutually contiguous with respect to $\mathbb{P}_{\bs\vartheta_{2}^*}$ from Le Cam's First Lemma \citep[see, e.g., Corollary 12.3.1  of][]{lehmannromano05book}, and in view of (\ref{eq:likelihoodratio}), we have
\[
\begin{pmatrix}
\bs{S}_n\\
\log \frac{d\mathbb{P}_{\bs\vartheta_{2,n}}}{d\mathbb{P}_{\bs\vartheta_2^*}}
\end{pmatrix}
 \overset{d}{\rightarrow} N\left( \begin{pmatrix}
 \bs{0}\\
 -(1/2) \bs{h}\t \bs{\mathcal{I}}\bs{h}
 \end{pmatrix},
 \begin{pmatrix}
 \bs{\mathcal{I}}& \bs{\mathcal{I}}\bs{h}\\
   \bs{h}\t \bs{\mathcal{I}} &
  \bs{h}\t \bs{\mathcal{I}}\bs{h}
 \end{pmatrix}\right)
 \quad\text{under $\mathbb{P}_{\bs\vartheta_{2}^*}$}
 \]
and 
\[
\bs{S}_n \overset{d}{\rightarrow} N(\bs{\mathcal{I}}\bs{h}, \bs{\mathcal{I}})\quad\text{under $\mathbb{P}_{\bs\vartheta_{2,n}}$}
\]
 from  Le Cam's Third Lemma \citep[see, e.g.,  12.3.2  of][]{lehmannromano05book}.
 Therefore, the proof of Proposition \ref{prop:t2_distribution} goes through under $\mathbb{P}_{\bs\vartheta_{2,n}}$ if we replace $\bs{S}_{\bs{\lambda},\bs{\eta}n}\overset{d}{\rightarrow} \bs{S}_{\bs{\lambda},\bs{\eta}}$ with 
$\bs{S}_{\bs{\lambda},\bs{\eta}n}\overset{d}{\rightarrow} \bs{S}_{\bs{\lambda},\bs{\eta}} + (\bs{\mathcal{I}}_{\bs\lambda }-\bs{\mathcal{I}}_{\bs\eta\bs\lambda }\bs{\mathcal{I}}_{\bs\eta }^{-1}\bs{\mathcal{I}}_{\bs\eta\bs\lambda }) \bs{h}_{\bs\lambda} =  \bs{S}_{\bs{\lambda},\bs{\eta}} +\bs{\mathcal{I}}_{\bs\lambda,\bs\eta} \bs{h}_{\bs\lambda}$, and the stated result follows.
\end{proof}

 \begin{proof}[Proof of Proposition \ref{prop:sht}]
 We provide a proof for $\hat M_{\text{ PLR}}$. The consistency proof for $\hat M_{\text{EM}}$ is similar.  We first prove that when $M<M_0$, $\Pr( PLR_n(M)> \hat c^M_{1-q_n})\rightarrow 1$ as $n\rightarrow \infty$.  Let $\tilde Q_n^M(\bs{\vartheta}_M):= Q_n^M( \bs{\vartheta}_M) + n^{-1} \tilde p_n(\bs{\vartheta}_{M})$ for $M\geq 2$. By Assumptions  \ref{assumption:3}(c) and \ref{assumption:4}(b), $n^{-1} \tilde p_n(\bs{\vartheta}_{M})\overset{p}{\rightarrow}0$ uniformly over $\Theta_{\bs{\vartheta}_M}$. Then, it follows from Lemma 2.4 of \cite{Newey1994} that
\begin{equation}\label{eq:uniform-conv}
\sup_{\bs{\vartheta}_M\in\Theta_{\bs{\vartheta}_M}} |\tilde Q_n^M(\bs{\vartheta}_M) -Q^M( \bs{\vartheta}_M)| = o_p(1),
\end{equation}
and  Assumption \ref{assumption:4}(a) and (b) and
 the standard consistency proof \citep[e.g., Theorem 2.1 of][]{Newey1994} give $\hat{\bs{\vartheta}}_M\overset{p}{\rightarrow} \bs{\vartheta}_M^*$ for $M<M_0$. Furthermore, by Assumption \ref{assumption:3}(c), $n^{-1/4}\nabla \tilde p_n(\bs{\vartheta}_{M})=o_p(1)$ uniformly over $\Theta_{\bs{\vartheta}_M}$, and it follows from the argument in Theorem 3.2 of \cite{white82em} that 
 \begin{equation}\label{eq:asynormal}
 \sqrt{n} (\hat{\bs{\vartheta}}_M- \bs{\vartheta}_M^*) \overset{p}{\rightarrow} N(0, A^M(\bs{\vartheta}_M^*)^{-1} B^M(\bs{\vartheta}_M^*)A^M(\bs{\vartheta}_M^*)^{-1}).
 \end{equation}
 
Then, from (\ref{eq:uniform-conv}), (\ref{eq:asynormal}), and the mean value expansion,  we have $\tilde Q_n^M(\hat{\bs{\vartheta}}_M)- Q^M({\bs{\vartheta}}_M^*) = O_p(n^{-1/2})$, and
\begin{align*}
\frac{ PLR_n(M)}{n} &:= \tilde Q_n^{M+1}(\hat{\bs{\vartheta}}_{M+1}) - \tilde Q_n^M(\hat{\bs{\vartheta}}_M) = Q^{M+1}({\bs{\vartheta}}_{M+1}^*) -Q^M({\bs{\vartheta}}_M^*)  + o_p(1).
\end{align*}
Because $Q^{M+1}({\bs{\vartheta}}_{M+1}^*) -Q^M({\bs{\vartheta}}_M^*)>0$ by Assumption \ref{assumption:4}(f), ${ PLR_n(M)}/{n}\rightarrow \infty$ as $n\rightarrow \infty$. By Lemma \ref{lemma: sht}, $- n^{-1} \ln q_n=o(1)$ and $\hat c^M_{1-q_n}-c^M_{1-q_n}=o_p(1)$ implies that $n^{-1} \hat c^M_{1-q_n}=o_p(1)$. Therefore,  when $M<M_0$,
we have $\Pr(  PLR_n(M)>\hat  c^M_{1-q_n})=\Pr(  PLR_n(M)/n>\hat c^M_{1-q_n}/n) \rightarrow 1$ as $n\rightarrow \infty$.

When $M=M_0$, because $  PLR_n(M_0) =O_p(1)$ by Proposition \ref{prop:tm0_distribution} and $\hat c^{M}_{1-q_n}\rightarrow\infty$ by $q_n=o(1)$, $ \Pr(  PLR_n(M_0)>\hat c^{M_0}_{1-q_n}) \rightarrow 0$ as $n\rightarrow \infty$.
 \end{proof}

%% file: sections/FM_appendixb.tex
 \subsection{Lemmas}

\begin{lemma}\label{lemma:unbounded_likelihood}
For any $M<\infty$, $ \Pr\left( -\log n + \ell(\bs W_{i^*};\bar Y_{i^*},s_{i^*}^2) < M\right) \rightarrow 0$ as $n\rightarrow \infty$.
\end{lemma}
\begin{proof}[Proof of Lemma \ref{lemma:unbounded_likelihood}]
Because $\sum_{t=1}^T \frac{\left(Y_{it}-\mu \right)^2}{s_{i^*}^2}= T-1$ when $i=i^*$, we have
\begin{align}
-\log n+\ell(\bs W_{i^*};\bar Y_{i^*},s_{i^*}^2)
&= -\log n - \frac{T}{2} \log s_{i^*}^2 -\frac{T}{2}\log(2\pi) - \frac{T-1}{2} \nonumber\\
&=-\log\left( C n  (s_{i^*}^2 )^{T/2} \right) , 
\end{align}
for some positive constant $C$.

Therefore, to prove the stated result, it suffices to show that for any $\epsilon>0$, $\Pr(n  (s_{i^*}^2 )^{T/2}>\epsilon)\rightarrow 0$ as $n\rightarrow\infty$. Given the property of the first-order statistic, the distribution of $s_{i^*}^2$ is given by $1-[1-F_{{T-1}}(s)]^n$, where $F_{{T-1}}(s)$ is the cumulative distribution function for chi-squared variables of degree $T-1$. It follows that
\[
\Pr(n  (s_{i^*}^2 )^{T/2}>\epsilon) = [1-F_{{T-1}}( (\epsilon/n)^{2/T} )]^n.
\]
When $T=3$, $1-F_{T-1}(s)= e^{-s/2}$ and $\Pr(n  (s_{i^*}^2 )^{T/2}>\epsilon) = e^{-C n^{1/3}}$ for some positive constant $C$, and therefore, $\Pr(n  (s_{i^*}^2 )^{T/2}>\epsilon)\rightarrow 0$ as $n\rightarrow \infty$, and the stated result follows.

For general $T\geq 2$, write
\begin{equation}\label{eq-lemma1}
\Big[1-F_{{T-1}}\big( (\epsilon/n)^{2/T} \big) \Big]^n =  \Big\{ \Big[1-F_{{T-1}}\big( (\epsilon/n)^{2/T} \big) \Big]^{\frac{1}{  F_{{T-1}}( (\epsilon/n)^{2/ T} )} }\Big\}^{ n F_{{T-1}}( (\epsilon/n)^{2/T} )  }.
\end{equation}
Then, because  $(1 - F)^{\frac{1}{F}} \to \frac{1}{e}$ when $ F \to 0 $, the stated result follows from (\ref{eq-lemma1}) if we can show  $$\frac{  F_{{T-1}}( (\epsilon x)^{2/T} )}{x} \to \infty \text{ as }  x \to 0$$ for $x=1/n$.
By applying L'Hôpital's rule, we have
\[\begin{split}
\lim_{x \to 0} \frac{  F_{{T-1}}( (\epsilon x )^{2T} )}{x } = \lim_{x \to 0}  f_{T-1}( (\epsilon x )^{2/T} ) \epsilon^{2T}  x^{2/T-1},
\end{split}\]
where $f_{k}$ is the PDF of the $\chi$-square distribution with $k$ degrees of freedom.
Note that $f_{T-1}((\epsilon x )^{2/T} ) = \frac{1}{2^{(T-1)/2} \Gamma((T-1)/2)} ((\epsilon x )^{2/T} )^{(T-1)/2-1} e^{- ((\epsilon x )^{2/T} )/2}$; then, $$ \lim_{x \to 0}  f_{T-1}( (\epsilon x )^{2/T} )  x^{2/T-1} = \lim_{x \to 0} C_{T,\epsilon} e^{- ((\epsilon x )^{2/T} )/2} x^{ -\frac{1}{T}} = \infty, $$
where $C_{T,\epsilon} = \frac{\epsilon^{(T-1)/T}}{2^{(T-1)/2} \Gamma((T-1)/2)} $ because $e^{- ((\epsilon x )^{2/T} )/2}\to 1$  and $x^{ -\frac{1}{T}} \to \infty$ as $x \to 0$ for any finite $T\geq 2$.
Therefore, $\lim_{x \to 0} \frac{  F_{{T-1}}( (\epsilon x )^{2T} )}{x } =\infty$, and the stated result for $T\geq 2$ follows from (\ref{eq-lemma1}).

\end{proof}

\begin{lemma}\label{lemma: sht}
Suppose that the assumptions in Proposition \ref{prop:sht} hold. If  $-n^{-1}\ln q_n=o(1)$, then  $n^{-1} c^M_{1-q_n}=o(1)$.

\end{lemma}
\begin{proof}
For brevity of notation,  
write $c_n=c^{M}_{1-q_n}$. 
By Theorem 2.1 of \cite{Foutz77as}, $PLR_n(M)\overset{d}{\rightarrow} \sum_{j=1}^K b_j \chi_j^2$ for $0<b_j<\infty$ and $K$ is finite, where $\chi_1^2$, ..., $\chi_K^2$ are independent chi-square random variables with one degree of freedom.
Then, we have 
\begin{align*}\label{eq:q_n}
q_n&= \Pr\left( \sum_{j=1}^K b_j \chi_j^2 \geq c_n\right)\leq \sum_{j=1}^K  \Pr\left(  \chi_j^2 \geq \frac{c_n}{ b_j}\right)\leq  \frac{K}{\sqrt{1-2t}}\exp\left(-t \frac{c_n}{ b^*}\right)\quad\text{for $0<t<\frac{1}{2}$}
\end{align*}
with $b^*=\arg\max\{b_1,...,b_K\}$, where the last inequality follows from a Chernoff bound: $\Pr\left(  \chi_j^2 \geq \frac{c_n}{ b^*}\right)\leq \frac{\E[\exp(t(\chi_j^2-1))]}{\exp(t(\frac{c_n}{ b^* }-1))}=\frac{1}{\sqrt{1-2t}}\exp\left(-t \frac{c_n}{ b^*}\right)$ for $0<t<\frac{1}{2}$. Therefore, 
$-\frac{\ln q_n}{n} \geq - \frac{1}{n}\ln\left(\frac{K}{\sqrt{1-2t}}\right) +\frac{1}{2b_j^*} \frac{c_n}{n}$,
and the stated result follows.

\end{proof}

\begin{lemma}\label{lemma:KS2018_lemma7}
    Suppose that $g(\bs{w};\bs{\psi},\alpha)$ is defined as (\ref{eq:repar}), where $\bs{\psi} = (\bs{\eta}^\top,\bs{\lambda}^\top)^\top$. Let $g^*$, $\nabla g^*$, and $\nabla\log g^*$ denote $g(\bs{W};\bs{\psi},\alpha)$, $\nabla g(\bs{W};\bs{\psi},\alpha)$,  and $\nabla \log g(\bs{W};\bs{\psi},\alpha)$ evaluated at $(\bs{\psi},\alpha)$, respectively. Let $\nabla f^*$ denote $\nabla f(\bs{W};\bs{\gamma}^*,\bs{\theta}^*)$.
    The following statements hold.
    \begin{enumerate}[label=(\alph*)]
        \item  For $l = 0,1,\ldots, \nabla_{ (\bs{\lambda} \otimes \bs{\eta}^{\otimes l})\t}g^* = 0$;
        \item $\nabla_{(\bs{\lambda}^{\otimes 2})\t } g^*  = \alpha (1 - \alpha) \nabla_{(\bs{\theta}^{\otimes 2})\t } f^*$;
        \item $\E[\nabla_{\lambda_i\lambda_j}\log g^*]=0$, $\E[\nabla_{\lambda_i\lambda_j\lambda_k}\log g^*]=0$, 
and  $\E[\nabla_{\eta\lambda_i\lambda_j }\log g^*]=-\E[\nabla_{\eta}\log g^*\nabla_{\lambda_i\lambda_j}\log g^*]$;  
 \item $\E[\nabla_{\lambda_i\lambda_j\lambda_k\lambda_\ell}\log g^*]=
-\E[\nabla_{\lambda_i\lambda_j}\log g^*\nabla_{\lambda_k\lambda_\ell}\log g^*+\nabla_{\lambda_i\lambda_k}\log g^*\nabla_{\lambda_j\lambda_\ell}\log g^*+\nabla_{\lambda_i\lambda_\ell}\log g^*\nabla_{\lambda_j\lambda_k}\log g^*]$.
    \end{enumerate}
\end{lemma}

\begin{proof}[Proof of Lemma \ref{lemma:KS2018_lemma7}]
    Recall that $$g(\bs{w};\bs{\psi},\alpha)  = \alpha f(\bs{w};\bs{\gamma},\bs{\nu}  + (1 - \alpha) \bs{\lambda}) + (1 - \alpha)  f(\bs{w};\bs{\gamma},\bs{\nu} - \alpha \bs{\lambda}). $$
    First, we show that for $l=0$ holds for (a),
    $\nabla_{\bs{\lambda}} g^* = \alpha ( 1 - \alpha) \nabla_{\bs{\theta}} f^* -  \alpha  ( 1 - \alpha) \nabla_{\bs{\theta}} f^*  = 0 $.
    For $l > 0$, by Fubini's theorem, we have
    \[
    \begin{split}
    \nabla_{ (\bs{\lambda}^{\otimes 2})\t } g & = \left. \nabla_{\bs{\lambda}}\Big( \alpha \nabla_{ (\bs{\gamma}, \bs{\theta})^{\otimes l}} f(\bs{w};\bs{\gamma},\bs{\nu}  + (1 - \alpha)  \bs{\lambda}) + (1 - \alpha) \nabla_{ (\bs{\gamma}, \bs{\theta})^{\otimes l}} f(\bs{w};\bs{\gamma},\bs{\nu} - \alpha \bs{\lambda}) \right|_{\bs{\nu} = \bs{\theta}^*, \bs{\lambda} = \bs{0} } \Big) \\
    & = \left. \Big( \alpha ( 1- \alpha) \nabla_{ (\bs{\gamma}^{\otimes l}, \bs{\theta}^{\otimes l+1})} f(\bs{w};\bs{\gamma},\bs{\nu}  + (1 - \alpha)  \bs{\lambda}) - \alpha (1 - \alpha) \nabla_{ (\bs{\gamma}^{\otimes l}, \bs{\theta}^{\otimes l+1})} f(\bs{w};\bs{\gamma},\bs{\nu} - \alpha \bs{\lambda}) \right|_{\bs{\nu} = \bs{\theta}^*, \bs{\lambda} = \bs{0} } \Big) \\
    & = 0.
    \end{split}
    \]
    To show part $(b)$, note that \begin{equation*}
    \begin{split}
    \nabla_{(\bs{\lambda}^{\otimes 2})\t } g & =  \nabla_{\bs{\lambda}} \Big( \alpha ( 1 - \alpha) \nabla_{\bs{\lambda}^\top } f(\bs{w};\bs{\gamma},\bs{\nu}  - \alpha (1 - \alpha)  \bs{\lambda})
    + (1 - \alpha) \nabla_{\bs{\lambda}^\top} f(\bs{w};\bs{\gamma},\bs{\nu} - \alpha \bs{\lambda}) \Big) \\
    & = \left.   \alpha ( 1 - \alpha)^2 \nabla_{(\bs{\lambda}^{\otimes 2})\t } f(\bs{w};\bs{\gamma},\bs{\nu}  + \alpha^2 (1 - \alpha)  \bs{\lambda})
    + (1 - \alpha) \nabla_{(\bs{\lambda}^{\otimes 2})\t } f(\bs{w};\bs{\gamma},\bs{\nu} - \alpha \bs{\lambda}) \right|_{\bs{\nu} = \bs{\theta}^*, \bs{\lambda} = \bs{0} } \\
    & = \nabla_{(\bs{\lambda}^{\otimes 2})\t } f^*.
    \end{split}
    \end{equation*}
 For  parts (c) and (d), observe that $\int  \nabla_{\lambda_i}  \log g(\bs w;\psi,\alpha)g(\bs w;\psi,\alpha)dx=0$ holds for any $\psi$ in the interior of $\Theta_\psi$, and differentiating this equation w.r.t.\ $\lambda_j$ gives
\begin{equation}
\int \{ \nabla_{\lambda_i\lambda_j} \log g(\bs w;\psi,\alpha)  +  \nabla_{\lambda_i} \log g(\bs w;\psi,\alpha)  \nabla_{\lambda_j} \log g(\bs w;\psi,\alpha) \} g(\bs w;\psi,\alpha) dx=0.\label{IME}
\end{equation}
Evaluating (\ref{IME}) at $\psi=\psi^*$ in conjunction with  part (a) gives the first equation in  part (c). Differentiating (\ref{IME}) w.r.t.\ $\lambda_k$ or $\eta$ and evaluating  at $\psi=\psi^*$ gives the latter two equations in part (c).  Part (d) follows from differentiating (\ref{IME}) w.r.t.\ $\lambda_k$ and $\lambda_\ell$ and evaluating at $\psi=\psi^*$ in conjunction with  parts (a) and (c).

\end{proof}

\begin{lemma} \label{tau_update} Suppose that the assumptions of Proposition \ref{prop:KS18_prop4} hold. If $\bs{\vartheta}_{M_0 + 1}^{h (K)}(\tau_0) - \bs{\vartheta}_{M_0 + 1}^{h *}(\tau_0)= o_p(1)$ and $\tau^{(K)}- \tau_0 = o_p(1)$, then (a) $\alpha_m^{(K+1)}/[\alpha_h^{(K+1)}+\alpha_{h+1}^{(K+1)}] - \tau_0 = o_p(1)$  and (b)  $\tau^{(K+1)} - \tau_0 = o_p(1)$.
\end{lemma}
 
\begin{proof}  
The proof is similar to the proof of Lemma 3 of \cite{chenli09as} and Lemma 10 in Appendix D of \cite{Kasahara2019}. We suppress $(\tau_0)$ from $\bs{\vartheta}_{M_0+1}^{h(K)}(\tau_0)$ and $\bs{\vartheta}_{M_0+1}^{h*}(\tau_0)$.  We suppress $\bs{Z}$ for brevity. Let $f_i(\bs\gamma,\bs{\theta}_j)$ and $f_i(\bs{\vartheta}_{M_0+1})$ denote $f(\bs{W}_i; \bs{\gamma}, \bs{\theta}_j)$ in (\ref{eq:f1}) and $f_{M_0+1}(\bs{W}_i; \bs{\vartheta}_{M_0+1})$ in (\ref{eq:fm0_1}), respectively. Applying a Taylor expansion to $\alpha_h^{(K+1)}= n^{-1}\sum_{i = 1}^n w_{ih}^{(K)}$ and using $\bs{\vartheta}_{M_0+1}^{h(K)} - \bs{\vartheta}_{M_0+1}^{h*} = o_p(1)$, we obtain
\[
\begin{aligned}
\alpha_m^{(K+1)}& = \frac{1}{n} \sum_{i = 1}^n \frac{\tau^{(K)}(\alpha_{h}^{(K)}+\alpha_{h+1}^{(K)})f_i(\bs\gamma^{(K)},\bs{\theta}_h^{(K)})}{f_i(\bs{\vartheta}_{M_0+1}^{h(K)})} \\
& = \frac{1}{n} \sum_{i = 1}^n \frac{\tau_0 \alpha_{h}^*f_i(\bs{\gamma}^*,\bs{\theta}_h^*)}{f_i(\bs{\vartheta}_{M_0+1}^{h*})} + o_p(1)
= \tau_0 \alpha_{h}^*+ o_p(1),
\end{aligned}
\]
where  the last equality follows from $\E[f_i(\bs{\gamma}^*,\bs{\theta}_h^*)/f_i(\bs{\vartheta}_{M_0+1}^{h*}) ] = 1$ and the law of large numbers. A similar argument gives $\alpha_{h+1}^{(K+1)} = (1 - \tau_0)\alpha_{h}^* + o_p(1)$, and part (a) follows. 

For part (b), define $H(\tau):=\sum_{i=1}^n w_{ih}^{(K)} \log(\tau) + \sum_{i=1}^n w_{i,h+1}^{(K)} \log(1-\tau) = n\alpha_h^{(K+1)} \log(\tau) + n\alpha_{h+1}^{(K+1)}$; then, $\tau^{(K+1)}$ maximizes $H(\tau)+p(\tau)$. $H(\tau)$ is maximized at $\tilde\tau:=  {\alpha_h^{(K+1)}}/(\alpha_h^{(K+1)}+\alpha_{h+1}^{(K+1)})= ({\tau_0\alpha_h^*+o_p(1)})/({ \tau_0\alpha_h^* + (1-\tau_0)\alpha_h^* + o_p(1)}) = \tau_0 + o_p(1)$. Observe that with $\bar\tau$ between $\tau^{(K+1)}$ and $\tilde \tau$, 
\begin{equation}\label{eq:p-H}
p(\tilde\tau)\leq p(\tilde\tau)-p(\tau^{(K+1)})\leq H(\tau^{(K+1)}) -H(\tilde\tau) = H''(\bar\tau)(\tau^{(K+1)}-\tilde\tau)^2,
\end{equation}
where the first inequality follows from $p(\tau)\leq 0$, the second inequality holds because $\tau^{(K+1)}$ maximizes $H(\tau)+p(\tau)$, and the last equality follows from expanding $H(\tau^{(K+1)})$ twice around $\tilde\tau$ and noting that $H'(\tilde\tau)=0$ because $\tilde\tau$ maximizes $H(\tau)$.  
Note that $H''(\tau) = -n \times\left\{\frac{\alpha_h^{(K+1)}}{\tau^2}  + \frac{\alpha_{h+1}^{(K+1)}}{(1-\tau)^2}\right\}<0$ and $\inf_\tau H''(\tau) \geq  -n  (\alpha_h^{(K+1)}+\alpha_{h+1}^{(K+1)})$. Therefore, in view of $\tilde\tau-\tau_0 =o_p(1)$ and (\ref{eq:p-H}), we have 
$(\tau^{(K+1)}-\tilde\tau)^2\leq p(\tilde\tau)/H''(\bar\tau) = O_p(n^{-1})$, and part (b) holds.
 
\end{proof}

\subsection{Score function for testing $H_0:m = 1$ against $H_A:m=2$}\label{sec:appendix_score_1}
$H^j(\cdot)$ is defined as the $j$-th order Hermite polynomial. $H^1(t) = t$, $H^2(t) = t^2 - 1$ , $H^3(t) = t^3 - 3t$, and $H^4(t) = t^4 - 6t^2 + 3$.
As shown in the supplementary material of \cite{Kasahara2015a}, the derivative of $\{ \frac{1}{\sigma} \phi(\frac{t}{\sigma}) \}$ is  $$\frac{\nabla_{\mu^m}\nabla_{ (\sigma^2)^{\ell}} \{ \frac{1}{\sigma} \phi(\frac{t}{\sigma}) \} }{\{ \frac{1}{\sigma} \phi(\frac{t}{\sigma}) \}}
 = \left(\frac{1}{2}\right)^\ell \left(\frac{1}{\sigma}\right)^{m+2\ell} H^{m+2\ell}\left(\frac{t}{\sigma}\right). $$

Let
\begin{equation}\label{eq:Hermite_polynomial}
\begin{split}
    f^* =  f(\bs{w};\gamma^*,\theta^* ), \nabla f^* =  \nabla f(\bs{w};\gamma^*,\theta^* ) , H^{j*}_{i,t} = \frac{1}{\sigma^* j!} H^{j}\left(\frac{y_{it} - \bs{x}_{it}^\top \bs{\beta}^* - \bs{z}_{it}^\top \bs{\gamma}^* - \mu^*  }{\sigma^*} \right);
\end{split}
\end{equation}
then, the first-order derivatives of the density functions are
\begin{align*}
    \nabla_{\mu} f^* & =   f^* \sum_{t=1}^T \frac{1}{\sigma} H^{1*}_{i,t}  ;    \nabla_{\sigma^2} f^*  =  f^*  \sum_{t=1}^T \frac{1}{2} \frac{1}{\sigma^2} H^{2*}_{it} ;\\
    \nabla_{\bs{\beta}} f^*  & =   f^*  \sum_{t=1}^T \frac{1}{\sigma} H^{1*}_{it} \bs{x}_{it}; \nabla_{\bs{\gamma}} f^*   =   f^* \sum_{t=1}^T \frac{1}{\sigma} H^{1*}_{it} \bs{z}_{it} .
\end{align*}

The score function defined in (\ref{eq:s_1}) is then written in terms of the Hermite polynomials:
\begin{equation}\label{eq:s_1_hermite}
\begin{split}
\bs s_{\bs \eta }  = \begin{pmatrix}
s_{\mu } \\
s_{\sigma } \\
\bs s_{\bs\beta } \\
\bs s_{\bs \gamma }
\end{pmatrix}= \begin{pmatrix}
\sum_{t=1}^T  H^{1*}_{i,t} \\
\sum_{t=1}^T  H^{2*}_{i,t} \\
\sum_{t=1}^T  H^{1*}_{i,t}\bs x_{it}\\
\sum_{t=1}^T  H^{1*}_{i,t}\bs z_{it}\\
\end{pmatrix},\qquad
\bs{s}_{\bs{\lambda} \bs\lambda} = \begin{pmatrix}
s_{\lambda_\mu\lambda_\mu}\\
s_{\lambda_\mu\lambda_\sigma}\\
s_{\lambda_\sigma\lambda_\sigma}\\
\bs s_{\lambda_\mu\lambda_{\bs\beta}}\\
\bs s_{\lambda_\sigma\lambda_{\bs\beta}}\\
\bs s_{\lambda_{\bs\beta}\lambda_{\bs\beta}}
  \end{pmatrix},
\end{split}
\end{equation}
where
\begin{equation}\label{eq:s-vector}
\begin{split}
 \begin{pmatrix}
s_{\lambda_{\mu \mu } }\\
s_{\lambda_{\mu \sigma }}\\
s_{\lambda_{\sigma \sigma }}\\
\bs s_{\lambda_{\mu\bs \beta }}\\
\bs s_{\lambda_{\sigma\bs \beta }}\\
\end{pmatrix} &=
\begin{pmatrix}
\sum_{t=1}^T H^{2*}_{i,t} +  \frac{1}{2} \sum_{t=1}^T \sum_{s \neq t} H^{1*}_{1,i,t} H^{1*}_{i,s} \\
3 \sum_{t=1}^T H^{3*}_{i,t} +  \sum_{t=1}^T \sum_{s \neq t} H^{1*}_{i,t} H^{2*}_{i,s} \\
3 \sum_{t=1}^T H^{4*}_{i,t} + \frac{1}{2} \sum_{t=1}^T \sum_{s \neq t} H^{2*}_{i,t} H^{2*}_{i,t}  \\
2 \sum_{t=1}^T H^{2*}_{i,t}\bs x_{it} +  \sum_{t=1}^T \sum_{s \neq t} H^{1*}_{i,t}H^{1*}_{i,s}\bs x_{it}  \\
3  \sum_{t=1}^T H^{3*}_{i,t}\bs x_{it} + 2 \sum_{t=1}^T\sum_{s \neq t} H^{1*}_{i,t}H^{2*}_{i,s}\bs x_{it}
\end{pmatrix},\quad\text{and} \\
\bs s_{\lambda_{\bs\beta\bs\beta}} & =
 \begin{pmatrix}
\sum_{t=1}^T H^{2*}_{i,t}x^2_{it,1} +  \frac{1}{2} \sum_{t=1}^T \sum_{s \neq t} H^{1*}_{i,t}x_{it,1} H^{1*}_{i,s} x_{is,1}\\
\vdots \\
\sum_{t=1}^T H^{2*}_{i,t}x^2_{it,q} +  \frac{1}{2} \sum_{t=1}^T \sum_{s \neq t} H^{1*}_{i,t}x_{it,q} H^{1*}_{i,s} x_{is,q} \\
2 \sum_{t=1}^T H^{2*}_{i,t}x_{it,1} x_{it,2} +  \sum_{t=1}^T \sum_{s \neq t} H^{1*}_{i,t}x_{it,1} H^{1*}_{i,s} x_{is,2}\\
\vdots \\
2 \sum_{t=1}^T H^{2*}_{i,t}x_{it,1}x_{it,q} + \sum_{t=1}^T \sum_{s \neq t} H^{1*}_{i,t}x_{it,1} H^{1*}_{i,s} x_{is,q} \\
2 \sum_{t=1}^T H^{2*}_{i,t}x_{it,2} x_{it,3} +  \sum_{t=1}^T \sum_{s \neq t} H^{1*}_{i,t}x_{it,2} H^{1*}_{i,s} x_{is,3}\\
\vdots \\
2 \sum_{t=1}^T H^{2*}_{i,t}x_{it,q-1}x_{it,q} + \sum_{t=1}^T \sum_{s \neq t} H^{1*}_{i,t}x_{it,q-1} H^{1*}_{i,s} x_{is,q}
\end{pmatrix}.
\end{split}
\end{equation}
When $T = 1$, the score functions are as follows:
\begin{equation} 
\bs s_{\bs\eta } = \begin{pmatrix}
s_{\mu } \\
s_{\sigma } \\
\bs s_{\bs\beta } \\
\bs s_{\bs \gamma }
\end{pmatrix}  = \begin{pmatrix}
  H^{1*}_{i} \\
  H^{2*}_{i} \\
  H^{1*}_{i}\bs x_{i}\\
  H^{1*}_{i}\bs z_{i}\\
\end{pmatrix},\quad
\begin{pmatrix}
s_{\lambda_{\mu \mu } }\\
s_{\lambda_{\mu \sigma }}\\
s_{\lambda_{\sigma \sigma }}\\
\bs s_{\lambda_{\mu\bs \beta }}\\
\bs s_{\lambda_{\sigma\bs \beta }}\\
\end{pmatrix} =
\begin{pmatrix}
 H^{2*}_{i}  \\
3  H^{3*}_{i} \\
3  H^{4*}_{i}   \\
2  H^{2*}_{i}\bs x_{i} \\
3   H^{3*}_{i}\bs x_{i}
\end{pmatrix},\quad \text{and }
\bs s_{\lambda_{\bs\beta\bs\beta}} =  \begin{pmatrix}
 H^{2*}_{i}x^2_{i,1} \\
\vdots \\
 H^{2*}_{i}x^2_{i,q} \\
2  H^{2*}_{i}x_{i,1} x_{i,2} \\
\vdots \\
2  H^{2*}_{i}x_{i,1}x_{i,q}  \\
2  H^{2*}_{i}x_{i,2} x_{i,3}\\
\vdots \\
2  H^{2*}_{i}x_{i,q-1}x_{i,q} \\
\end{pmatrix}.
\end{equation}
Note that $s_{\sigma  }$ and $ s_{\lambda_{\mu \mu }} $ are perfectly collinear, and therefore, the Fisher information matrix associated with the proposed score function is singular under this reparameterization for data with $T=1$.

\subsection{Score function for testing $H_0:m = M_0$ against $H_A:m=M_0 + 1$}\label{sec:appendixb:score_m}
The derivative of the reparameterized density w.r.t. $\lambda$ at $\psi^{h*}_{\tau}$ is identically zero similarly to the test of the homogeneity case.
The values of the score function $s_{\eta i}$ contain the first-order derivatives w.r.t. the $\pi$s $\gamma$ and $\nu$ at $\psi^{h*}_{\tau}$:
\begin{equation}
\begin{split}
\nabla_{\pi^j} l^h(\bs{w};\psi^{h*}_{\tau},\tau) & = \frac{f(\bs{w};\gamma^*,\theta_0^{j*}) - f(\bs{w};\gamma^*,\theta_0^{M_0 *}) }{\sum_{j=1}^{M_0} \alpha_0^{j*}  f(\bs{w};\gamma^*,\theta_0^{j*})};\\
\nabla_{\gamma} l^h(\bs{w};\psi^{h*}_{\tau},\tau) & = \frac{  \sum_{j=1}^{M_0} \alpha_0^{j*} \nabla_{\gamma} f(\bs{w};\gamma^*,\theta_0^{j*})}{\sum_{j=1}^{M_0} \alpha_0^{j*}  f(\bs{w};\gamma^*,\theta_0^{j*})};\\
\nabla_{\nu} l^h(\bs{w};
\psi^{h*}_{\tau},\tau) & = \frac{\nabla_{\theta} f(\bs{w};\gamma^*,\theta_0^{h*})}{\sum_{j=1}^{M_0} \alpha_0^{j*}  f(\bs{w};\gamma^*,\theta_0^{j*})}.
\end{split}
\end{equation}

Define $H^{b*}_{j,i,t}$ as an abridged expression for $\frac{1}{b!} \frac{1}{\sigma_0^* } H^{b}(\frac{y_{it} - \mu_0^{j*} - x_{it}' \beta_0^{j*} - z_{it}' \gamma^* }{\sigma_0^{j*}}) $.  Define the weight $w_{i}^{j*}$ as
\[
w_{i}^{j*}   =  \frac{\alpha_0^{j*} f(\{\bs{W}_{it}\}^T_{t=1};\gamma^*,\theta_0^{j*})}{ f_{M_0}(\{\bs{W}_{it}\}^T_{t=1};\vartheta_{M_0}^*)},j=1,\ldots,M_0,
\]
where $f_{M_0}(\{\bs{W}_{it}\}^T_{t=1};\vartheta_{M_0}^*)$ is defined by equation (\ref{eq:fm0}).

As shown in section \ref{sec:appendixb:score_m}, the score functions are
\[
\begin{split}
\bs s_{\bs\alpha}(\bs w_i)  = \begin{pmatrix}
\frac{f(\bs{w} | \theta_0^{1*}) - f(\bs{w} | \theta^{M_0*}_0) }{\sum_{l} \alpha_0^{l*} f(\bs{w} | \theta^{l*}_0)}\\
\vdots \\
\frac{f(\bs{w} | \theta_0^{M_0-1 *}) - f(\bs{w} | \theta_0^{M_0*}) }{\sum_{l} \alpha_0^{l*} f(\{\bs{W}_{it}^*\}_{t=1}^T | \theta_0^{l*})}
\end{pmatrix},
\bs s_{\mu }(\bs w_i)  = \begin{pmatrix}
w_{i}^{1*} \sum_{t=1}^T  H^{1*}_{1,i,t} \\
\vdots \\
w_{i}^{M_0*}  \sum_{t=1}^T  H^{1*}_{M_0,i,t}\end{pmatrix}, 
\bs s_{\bs\beta}(\bs w_i)  = \begin{pmatrix}
w_{i}^{1 *} \sum_{t=1}^T  H^{1*}_{1,i,t}x_{it}\\
\vdots \\
w_{i}^{M_0 *} \sum_{t=1}^T H^{1*}_{M_0,i,t}x_{it} \end{pmatrix},\\
\bs s_{\sigma }(\bs w_i) = \begin{pmatrix}
w_{i}^{1 *} \sum_{t=1}^T H^{2*}_{1,i,t} \\
\vdots \\
w_{i}^{M_0 *} \sum_{t=1}^T H^{2*}_{M_0,i,t} \end{pmatrix},
\bs s_{\bs\gamma }(\bs w_i) = \begin{pmatrix}
w_i^{1 *} \sum_{t=1}^T  H^{1*}_{1,i,t}z_{it}\\
\vdots \\
w_i^{M_0 *} \sum_{t=1}^T H^{1*}_{M_0,i,t}z_{it}
\end{pmatrix}.
\end{split}
\]
The score function for $\bs s_{\bs\lambda\bs\lambda}^h$ is obtained analogously to  $\bs s_{\bs\lambda\bs\lambda}$ by replacing 
$H^{b*}_{i,t}$ with $H^{b*}_{h,i,t}$ for $b=1,..,4$ so that
\[
\begin{split}
\bs s^h_{\bs\lambda_{\mu \sigma} }(\bs w_i) &  =  w_i^{h *}
\begin{pmatrix}
\sum_{t=1}^T H^{2*}_{h,i,t} +  \frac{1}{2} \sum_{t=1}^T \sum_{s \neq t} H^{1*}_{h,i,t} H^{1*}_{h,i,s} \\
3 \sum_{t=1}^T H^{4*}_{h,i,t} + \frac{1}{2} \sum_{t=1}^T \sum_{s \neq t} H^{2*}_{h,i,t} H^{2*}_{h,i,t}  \\
3 \sum_{t=1}^T H^{3*}_{h,i,t} +  \sum_{t=1}^T \sum_{s \neq t} H^{1*}_{h,i,t} H^{2*}_{h,i,s} \\
2 \sum_{t=1}^T H^{2*}_{h,i,t}x_{it} +  \sum_{t=1}^T \sum_{s \neq t} H^{1*}_{h,i,t}x_{it} H^{1*}_{h,i,s} \\
3  \sum_{t=1}^T H^{3*}_{h,i,t}x_{it} + 2 \sum_{t=1}^T\sum_{s \neq t} H^{1*}_{h,i,t}x_{it} H^{2*}_{h,i,s}
\end{pmatrix} , \\
\bs s^h_{\bs\lambda_{\beta}}(\bs w_i)  & =  w_i^{h *}  \begin{pmatrix}
\sum_{t=1}^T H^{2*}_{h,i,t}x^2_{it,1} +  \frac{1}{2} \sum_{t=1}^T \sum_{s \neq t} H^{1*}_{h,i,t}x_{it,1} H^{1*}_{h,i,s} x_{is,1}\\
\vdots \\
\sum_{t=1}^T H^{2*}_{h,i,t}x^2_{it,q} +  \frac{1}{2} \sum_{t=1}^T \sum_{s \neq t} H^{1*}_{h,i,t}x_{it,q} H^{1*}_{h,i,s} x_{is,q} \\
2 \sum_{t=1}^T H^{2*}_{h,i,t}x_{it,1} x_{it,2} +  \sum_{t=1}^T \sum_{s \neq t} H^{1*}_{h,i,t}x_{it,1} H^{1*}_{h,i,s} x_{is,2}\\
\vdots \\
2 \sum_{t=1}^T H^{2*}_{h,i,t}x_{it,1}x_{it,q} + \sum_{t=1}^T \sum_{s \neq t} H^{1*}_{h,i,t}x_{it,1} H^{1*}_{h,i,s} x_{is,q} \\
2 \sum_{t=1}^T H^{2*}_{h,i,t}x_{it,2} x_{it,3} +  \sum_{t=1}^T \sum_{s \neq t} H^{1*}_{h,i,t}x_{it,2} H^{1*}_{h,i,s} x_{is,3}\\
\vdots \\
2 \sum_{t=1}^T H^{2*}_{h,i,t}x_{it,q-1}x_{it,q} + \sum_{t=1}^T \sum_{s \neq t} H^{1*}_{h,i,t} x_{it,q-1} H^{1*}_{h,i,s} x_{is,q}
\end{pmatrix}.
\end{split}
\]

\subsection{How to simulate the asymptotic distribution}\label{sec:appendixb:simulate}

%% file: sections/FM_appendixd.tex
\begin{table}[H] \centering
	\caption{Parameter Specification for Null Models with $M_0 = 1,2,3,4$}
	\label{table:parameter}
	\begin{tabular}{l|l }
		\toprule
		~ & {$M_0 = 1$} \\
		\midrule
		$N$ & $\{ 100,500 \}$  \\
		$T$ & $\{ 2,5,10 \}$ \\
		$a_n$ & $(0.001,0.005, 0.01, 0.05,0.1,0.2, 0.3, 0.4)$ \\
		\midrule
		~ & {$M_0 = 2$} \\
		\midrule
		$N$ & $\{ 100,500 \}$  \\
		$T$ & $\{ 2,5,10 \}$ \\
		$\alpha$ & $\{ (0.5,0.5);(0.2,0.8) \}$  \\
		$\mu$ & $\{ (-1,1), (-0.5,0.5), (-0.8,0.8) \}$ \\
		$\sigma$ & $\{ (1, 1), (1.5, 0.75), (0.8,1.2)\}$ \\
		$a_n$ & $(0.01, 0.05,0.1,0.2, 0.3, 0.4)$ \\
		\midrule
		~ & {$M_0 = 3$} \\
		\midrule
		$N$ &  $\{ 100,500 \}$ \\
		$T$ &  $\{ 2,10 \}$\\
		$\alpha$ &$\{ (1/3,1/3,1/3);(0.25,0.5,0.25) \}$ \\
		$\mu$ &  $\{(-4, 0, 4);(-4, 0, 5);(-5, 0, 5);(-4, 0, 6);(-5, 0, 6);(-6, 0, 6)\}$\\
		$\sigma$ & $\{(1, 1, 1);(0.75, 1.5, 0.75)\}$ \\
		$a_n$ & $(0.01, 0.05,0.1,0.2, 0.3, 0.4)$ \\
		\midrule
		~ & {$M_0 = 4$} \\
		\midrule
		$N$ &  $\{ 100,500 \}$ \\
		$T$ &  $\{ 2,10 \}$\\
		$\alpha$ &$\{ (0.25, 0.25, 0.25, 0.25)\}$ \\
		$\mu$ &  $\{(-4,-1,1,4);(-5,-1,1,5);(-6,-2,2,6);(-6,-1,2,5);(-5,0,2,4);(-6,0,2,4)\}$\\
		$\sigma$ & $\{(1, 1, 1, 1);(1, 0.75, 0.5, 0.25)\}$ \\
		$a_n$ & $(0.01, 0.05,0.1,0.2, 0.3, 0.4)$ \\
		\bottomrule
	\end{tabular}
\end{table}

\begin{table}[H] \centering 
	\caption{The Estimated $a_n$-Function Based on the Simulated Nominal Size} 
	\label{tab:rho_estim} 
	\small
	\begin{tabularx}{\textwidth}{l @{\extracolsep{\fill}} cccc} 
		\\[-1.8ex]\hline 
		\hline \\[-1.8ex] 
		& \multicolumn{4}{c}{\textit{Dependent variable}} \\ 
		\cline{2-5} 
		\\[-1.8ex] & \multicolumn{4}{c}{$\log\left(\frac{\hat{s}}{1- \hat{s}}\right) -  \log\left(\frac{ {0.05}}{1-  {0.05}}\right)$} \\ 
		\\[-1.8ex] & (1) & (2) & (3) & (4)\\ 
		\hline \\[-1.8ex] 
		$1/T$ & 0.776$^{***}$ & $-$0.288$^{***}$ & 0.611$^{***}$ & 0.258$^{***}$ \\ 
		& (0.238) & (0.074) & (0.050) & (0.087) \\ 
		& & & & \\ 
		$1/N$ & 28.143$^{***}$ & 4.637 & 21.156$^{***}$ & 8.585$^{**}$ \\ 
		& (10.127) & (3.124) & (2.524) & (4.334) \\ 
		& & & & \\ 
		$\log\left(\frac{ {a_n}}{1 -  {a_n}}\right)$ & $-$0.016 & $-$0.101$^{***}$ & $-$0.111$^{***}$ & $-$0.128$^{***}$ \\ 
		& (0.019) & (0.009) & (0.007) & (0.030) \\ 
		& & & & \\ 
		$\log\left(\frac{\omega(\bs{\vartheta}_{M_0};M_0)}{1 - \omega(\bs{\vartheta}_{M_0};M_0)}\right)$ &  & $-$0.197$^{***}$ & 0.002 & $-$0.013$^{***}$ \\ 
		&  & (0.029) & (0.006) & (0.003) \\ 
		& & & & \\ 
		Constant & $-$0.616$^{***}$ & $-$0.811$^{***}$ & $-$0.680$^{***}$ & $-$0.735$^{***}$ \\ 
		& (0.113) & (0.047) & (0.060) & (0.068) \\ 
		& & & & \\ 
		\hline \\[-1.8ex] 
		Observations & 48 & 648 & 576 & 288 \\ 
		\hline 
		\hline \\[-1.8ex] 
		\textit{Note:}  & \multicolumn{4}{r}{$^{*}$p$<$0.1; $^{**}$p$<$0.05; $^{***}$p$<$0.01}. \\ 
	\end{tabularx} 
\end{table}